\documentclass[11pt, letterpaper]{article}

\usepackage[style=alphabetic,maxnames=99,minnames=99,natbib=true]{biblatex}
\usepackage{blindtext}
\usepackage[margin=1in]{geometry}
\usepackage{setspace}
\usepackage{amsmath,amssymb,amsthm}
\usepackage{thmtools,thm-restate}
\usepackage{booktabs}

\newcommand{\blue}[1]{\textcolor{blue}{#1}}

\usepackage{titlesec}
\titlespacing*{\paragraph}{0pt}{5pt}{1em}
\usepackage{float}
\usepackage{graphicx}      
\usepackage{footnote}
\usepackage[ruled,vlined,linesnumbered]{algorithm2e}

\SetKwInOut{Input}{Input}
\SetKwInOut{Output}{Output}

\usepackage[colorlinks=true, allcolors=blue]{hyperref}
\usepackage[nameinlink,noabbrev,capitalise]{cleveref}
\usepackage{xcolor}

\usepackage[suppress]{color-edits}
\definecolor{darkgreen}{rgb}{0.0,0.4,0.0}
\addauthor{mb}{red}
\addauthor{gf}{orange}

\DeclareUnicodeCharacter{0302}{}
\definecolor{darkgreen}{rgb}{0.0,0.4,0.0}

\newtheorem{theorem}{Theorem}[section]
\newtheorem*{theorem*}{Theorem}
\newtheorem{lemma}{Lemma}[section]
\newtheorem{definition}{Definition}[section]
\newtheorem{corollary}{Corollary}[section]

\newtheorem{observation}{Observation}[section]
\newtheorem{claim}{Claim}[section]
\newtheorem{remark}{Remark}[section]

\newcommand{\IME}{``\textbf{I}ndividually \textbf{MM}S-satisfying or EF\textbf{X}-satisfying (IMMX)"\xspace}
\newcommand{\IMMX}{IMMX\xspace}
\newcommand{\EIME}{``\textbf{I}ndividually $(1-\varepsilon)$-\textbf{MM}S-satisfying or EF\textbf{X}-satisfying ($(1-\varepsilon)$-IMMX)"\xspace}
\newcommand{\EIMMX}{$(1-\varepsilon)$-IMMX\xspace}
\newcommand{\EIMMXCHORES}{$(1+\varepsilon)$-IMMX\xspace}
\newcommand{\realloc}{\textbf{REALLOC}\xspace}

\title{Near-Optimal Best-of-Both-Worlds Fairness for Few Agents}
\author{Moshe Babaioff\thanks{The Hebrew University of Jerusalem (HUJI). Email: {\tt moshe.babaioff@mail.huji.ac.il}. Supported  by a Golda Meir Fellowship and the Israel Science Foundation (grant No. 301/24). }
\and Gefen Frosh\thanks{The Hebrew University of Jerusalem (HUJI). Email: {\tt  gefen.frosh@mail.huji.ac.il}. Supported  by  the Israel Science Foundation (grant No. 301/24).}
}
\date{\today}

\begin{document}
 \maketitle
 \thispagestyle{empty}
\setcounter{page}{0}

\begin{abstract}

We consider the problem of fair allocation of indivisible goods among agents with additive valuations, 
aiming for \emph{Best-of-Both-Worlds (BoBW)} fairness: 
a distribution over allocations that is ex-ante fair, and additionally, it is supported only on deterministic allocations that are ex-post fair. 
Existing BoBW algorithms are far from achieving  the best possible ex-post fairness guarantees, even in the well studied special case when there are only few agents.  
We focus on BoBW for few agents, and our main result is the design of the first 
poly-time BoBW algorithms achieving near-optimal fairness for three agents. 
We also present optimal poly-time BoBW results  for two agents.

For three agents, we prove that there exists an ex-ante proportional
distribution over at most six allocations, each of which is Epistemic EFX
(EEFX) and gives every agent at least $\tfrac{9}{10}$ of her maximin share
(MMS). Since MMS allocations need not exist, some agent may fall below her
MMS; we guarantee that any such agent is EFX-satisfied -- a new criterion we
call \IME. We complement this with an FPTAS preserving all envy-based  guarantees, and also preserving all value-based guarantees up to $(1-\varepsilon)$. 
Furthermore, we present an FPTAS which guarantees \emph{exact} ex-ante proportionality, while guaranteeing each agent receives $(1-\varepsilon)$ of her MMS or is EFX-satisfied ex-post -- notable, as
computing EFX allocations in polynomial time is open even without BoBW
constraints.
For two agents, we give an FPTAS that is ex-ante envy-free, ex-post EFX, and
guarantees each agent $(1-\varepsilon)$ of her MMS, matching the strongest
guarantees achievable in polynomial time.

 \end{abstract}

\pagebreak

\section{Introduction}

A classical problem in fair division is to divide a set $M$ of $m$ goods among a set $N$ of $n$ equally-entitled agents, each having an additive valuation {function} over the goods.
When resources are divisible, the celebrated cake-cutting literature shows that strong fairness benchmarks such as proportionality and envy-freeness can always be guaranteed \cite{azizcake2016}. However, in many modern applications -- from dividing an inheritance or household chores, to allocating computing resources, scholarships, or housing slots -- the items are \emph{indivisible}. This indivisibility fundamentally complicates the problem of fair division and an extensive literature in Theoretical Computer Science focuses on this problem \cite{ akrami2024breaking, akrami2025general, amanatidisapprox2017, amanatidis2024pushing, aziz_new, Aziz2018, aziz2020simultaneously, babaioff2021bobw, babaioff2022fairshares, barman2020approximation, bu2024bestofbothworldsfairallocationindivisible, caragiannis2019unreasonable, chaudhury2021little, chaudhury2024efx, feige2021tight, feige2022improvedmaximinfairallocation, azizfreemannew, garg2019approximating, garg2024bestofbothworldsfairnessenvycycleeliminationalgorithm, feldman2024breaking, Kurokawa2018}. 

In the context of indivisible goods, fairness can be interpreted and guaranteed in various different ways.  
Different frameworks have been proposed to capture what it means for an allocation to be ``fair":  
\emph{Share-based} approaches focus on ensuring that each agent receives an allocation of high enough value, while \emph{envy-based} approaches focus 
on trying to eliminate \emph{envy} among agents. These perspectives give rise to several formal fairness benchmarks, the most fundamental of which are proportionality and envy-freeness.

\paragraph{Fundamental fairness benchmarks.}
Two criteria serve as fundamental benchmarks for fairness. \emph{Proportionality} requires that each agent receives at least a $1/n$-fraction of her total value for the set of all goods. \emph{Envy-freeness (EF)} requires that no agent prefers another agent's bundle to her own. With divisible goods, both properties can be guaranteed simultaneously. With indivisible goods, however, simple examples show that {it may be impossible to obtain either benchmark}: 
if there is only one good and two agents, 
giving the good to either agent violates proportionality and envy-freeness for the other.

\paragraph{Relaxations of ex-post fairness.}
This impossibility has spurred the development of relaxations that capture meaningful notions of fairness in indivisible-items settings. For envy-freeness, two widely studied relaxations are \emph{envy-free up to one good (EF1)} \cite{lipton2004approximately} and \emph{envy-free up to any good (EFX)} \cite{caragiannis2019unreasonable}. EF1 requires that any envy towards another agent can be eliminated by removing a single item from the envied bundle. EFX strengthens this: envy must disappear when \emph{any} item is removed. For share-based fairness, the \emph{maximin share (MMS)} benchmark \cite{budishmms2011} requires that each agent receives at least the value she can guarantee herself by partitioning the goods into $n$ bundles, and receiving the worst one. Although in the case of two agents MMS allocations always exist, via the classic ``cut and choose" protocol \cite{bouvaret2014scale}, \citet{Kurokawa2018} show that even for the case of three additive agents, MMS allocations do not always exist. 

EFX allocations exist for three agents \cite{chaudhury2024efx}, however, despite significant research efforts, it is not known whether EFX exists for more than three agents. Recently, \citet{caragiannis2022new} introduced a relaxation of EFX, called  \emph{Epistemic EFX (EEFX)}. Intuitively, an allocation is \emph{Epistemic EFX (EEFX)} if each agent can view her own bundle as part of some EFX allocation -- that is, she can be persuaded that there exists a way 
to reallocate the remaining items such that she would not envy any other agent after the removal of any single item from their bundle.\footnote{Equivalently, one may think of the agent being able to reshuffle the items among the other agents’ bundles while keeping her own bundle fixed, until envy towards any other bundle can be eliminated by removing any single item from that bundle.}
These relaxations, however, are inherently \emph{ex-post}: they apply to realized allocations. 
Only guaranteeing ex-post fairness can still be perceived by some agents as unfair. Indeed, suppose there is only one good and two agents, Alice and Bob. Deterministically giving the good to Alice is EFX and MMS ex-post. Yet, Bob always receives nothing, which Bob might perceive as unfair, and argue that the item should be randomly assigned to one of the agents. 

\paragraph{Ex-ante fairness.}

By choosing allocations at random, we can ensure that each agent's \emph{expected} value  meets certain fairness criteria. For example, giving the single item to Alice or Bob with equal probability ensures that both agents receive their proportional share in expectation. This captures a notion of \emph{ex-ante proportionality}. Yet, guaranteeing only ex-ante proportionality can still be perceived as allowing for unnecessary unfairness. 
To illustrate this, imagine a case with 101 identical items. While randomly assigning all items either to Alice or to Bob guarantees ex-ante proportionality, 
it unnecessarily violates any reasonable notion of ex-post fairness (as opposed to giving each agent 50 items, and flipping a coin as to who receives the last item).

\paragraph{Best-of-Both-Worlds fairness.}
These tensions motivate the \emph{Best-of-Both-Worlds (BoBW)} paradigm. The goal is to combine the strengths of both approaches by allocating according to a distribution over allocations that is fair, and also guarantee fairness for each realized allocation. In other words, we seek a randomized allocation that is ex-ante fair (in the sense of envy freeness or proportionality of the expected values), while ensuring that every allocation in the support satisfies strong ex-post fairness criteria. The BoBW approach has recently received growing attention \cite{azizfreemannew, babaioff2021bobw, feldman2024breaking}. It provides a natural answer to the shortcomings of only ex-ante fairness, or only ex-post fairness. However, existing BoBW results fall short of matching the natural benchmarks of ex-ante proportionality\footnote{
One might consider \emph{ex-ante envy-freeness} to be a more natural benchmark; this benchmark is stronger and thus harder to obtain with strong ex-post guarantees. 
We leave the problem of obtaining it for future work, see \cref{sec:conclusion}.}, combined with ex-post envy-based fairness and a guarantee of a large fraction of the MMS\footnote{Beyond the relatively weak guarantees that are derived from the envy-based fairness itself.}. This holds even in the fundamental settings with only few agents.

\paragraph{Limitations of existing BoBW results.}
Two prominent lines of work illustrate the current boundaries. 

\emph{Ex-ante EF + ex-post EF1:} \citet{azizfreemannew} show that ex-ante envy freeness (EF) can be combined with ex-post EF1. 
While elegant, the ex-post guarantee has a weakness: 
consider the following example -- there are $n$ agents with identical additive valuations: there are $n-1$ valuable items of value $n$ each ("diamonds"), and $n$ non-valuable items of value $1$ each ("rocks"). 
An ex-ante EF distribution that is ex-post EF1 might be the following: give a random agent a rock, and each of the other agents receives both a rock and a diamond (Observe that every such allocation is EF1 -- when removing the diamond from the bundle of an agent who received both a diamond and a rock, the agent with a rock is no longer envious). 
Therefore, an ex-ante EF distribution that is supported on ex-post EF1 allocations might only guarantee every agent a $1/n$-fraction of the MMS.  
We note that this setting has an allocation that is ex-post EF (thus ex-post proportional and MMS): simply assign all rocks to one agent, and give each of the others a diamond.

\emph{Ex-ante Proportionality + ex-post $\frac{1}{2}$-MMS:} \citet{babaioff2021bobw} achieve a different tradeoff: ex-ante proportionality combined with the guarantee that every allocation in the support gives each agent a fraction of at least $\frac{1}{2}$ of her MMS. 
This approach ensures some share-based fairness, but the factor of $\frac{1}{2}$ 
is far from the known existential bound for MMS, which is currently $\frac{7}{9}$ 
\cite{huang2025fptas79approximationmaximinshare}.
For the special case of three agents, the guarantee of $\frac{1}{2}$ of the MMS \cite{babaioff2021bobw} is much smaller than the best known bound of  $\frac{11}{12}$ of the MMS \cite{feige2022improvedmaximinfairallocation}. Moreover , this approach may fail to provide any meaningful envy guarantee: for instance, if four identical goods are to be divided between Alice and Bob by randomly giving three goods to one agent and one good to the other, the resulting lottery is ex-ante proportional and ex-post $1/2$-MMS, yet clearly violates EF1, the weakest standard envy-based fairness notion (much weaker than EF and EFX).

Neither approach achieves the natural target of obtaining ex-ante proportionality together with the combination of both a strong ex-post envy-based fairness (ideally EFX), and a strong ex-post share based fairness (ideally a very large fraction of the MMS). 
We will seek such results for few agents, but before discussing those, we consider the central aspect of computational complexity.

\paragraph{Computational considerations.}
Any practical procedure for fair division must also be computationally feasible.  
In the case of two agents, the classical \emph{cut-and-choose} protocol guarantees both ex-ante proportionality and ex-post MMS and EFX fairness, yet computing the corresponding MMS partition is NP-hard.  
The challenge only intensifies with three agents: computing each agent's MMS value is NP-hard, and no polynomial-time algorithm is known for finding EFX allocations.  
A natural idea is to replace exact MMS computations with a $(1-\varepsilon)$ approximation using an FPTAS; however, doing so does not automatically preserve the fairness guarantees (other than the immediate approximate ex-post value guarantees).  
In particular, substituting {exact MMS partitions with} approximate MMS partitions may destroy ex-ante proportionality and the envy-based (EFX/EEFX) properties that {the construction with exact MMS partitions guarantees}.
Understanding how to maintain these guarantees under partitions that only $(1-\varepsilon)$ approximate the MMS is therefore a key technical challenge addressed by our work.

\paragraph{Guarantees for few agents.}
The two and three-agent settings are particularly meaningful, both theoretically and in practice. 
The case of two agents formalizes one of the oldest and most intuitive fair division problems -- the  "divide-and-choose" principle -- relevant to contexts such as divorce settlements. 
An early reference to this procedure appears in Hesiod's \emph{Theogony}, written nearly 2,800 years ago, where Prometheus and Zeus divide a sacrificial animal --   
Prometheus partitions the meat into two portions, and Zeus is invited to choose one \cite{hesiod}.

The three-agent case represents another level of complexity, yet it is still fundamental as it captures common  inheritance problems (ones with three heirs), and settings with three partners dissolving their common ownerships of a set of goods. 
Although the setting of three agents seems very restrictive, it is still rather challenging and has recently drawn much research attention   \cite{chaudhury2024efx, feige2021tight, feige2022improvedmaximinfairallocation, GOURVES201950}. 
For example, EFX existence for three agents was only proven recently \cite{chaudhury2024efx}, after much effort, and the case of four agents is still open.

Surprisingly, even for these fundamental settings with up to three agents, no prior Best-of-Both-Worlds (BoBW) algorithm achieves optimal fairness guarantees in polynomial time. For two agents, existing methods ensure ex-ante proportionality (and hence ex-ante envy-freeness, as valuations are additive) together with ex-post EFX, but the fraction of the maximin share (MMS) they guarantee is substantially below the best-known achievable bound \cite{garg2024bestofbothworldsfairnessenvycycleeliminationalgorithm, bu2024bestofbothworldsfairallocationindivisible}. For three agents, the situation is even more restrictive: no known BoBW algorithm simultaneously attains ex-ante proportionality and an ex-post envy-based guarantee stronger than EF1, while also achieving an MMS fraction close to the known existential limits (even in non-polynomial time).  This gap highlights the difficulty of reconciling ex-ante and ex-post fairness notions even in the simplest multi-agent settings with only two or three agents.

\newpage

\subsection{Our Results}

{We study the fair allocation of indivisible goods among three agents with additive valuations, and design algorithms that achieve \emph{near-optimal Best-of-Both-Worlds (BoBW)} fairness guarantees. We also present an optimal polynomial-time BoBW result for two additive agents.}
Before presenting our results in detail, let us first briefly overview our main contributions.

For three agents, we establish the existence of a distribution that is \emph{ex-ante proportional}, and whose every allocation in the support is \emph{EEFX} and guarantees each agent at least a \(\tfrac{9}{10}\)-fraction of her MMS.  
In each allocation, at most one agent may receive below her MMS value (but at least 90\% of it), and this agent is also guaranteed to be \emph{EFX}-satisfied.  
We refer to this new fairness notion as \IME: an allocation is \emph{\IMMX} if each agent either receives her MMS\ share or is \emph{EFX}-satisfied (see additional discussion of this concept in \cref{sec:IME} below).
The \IMMX\ criterion captures the strongest trade-off {that may be achievable} between share-based and envy-based guarantees when MMS allocations may fail to exist.
We discuss our results for three agents in \cref{sec:intro-3-agents} below.

The above construction for three agents is not polynomial-time, as it requires computing exact \(MMS\) partitions.  
However, by replacing these with \((1-\varepsilon)\)-approximate MMS partitions obtained via an FPTAS, we obtain an FPTAS\ for computing a BoBW distribution that remains \emph{EEFX} and also satisfies \emph{\EIME}, a relaxed notion of \IMMX that allows a $\varepsilon$-factor loss in the share-based value guarantee.   
In particular, ex-ante proportionality is not guaranteed, but only \((1-\varepsilon)\) of it. However, by dropping the \emph{EEFX} requirement we are able to obtain an FPTAS\ that guarantees \emph{exact} ex-ante proportionality and ex-post \((\tfrac{9}{10}-\varepsilon)-\mathrm{MMS}\), together with \EIMMX.  
This result substantially improves upon the previous BoBW algorithm of \cite{babaioff2021bobw}, which achieves ex-ante proportionality but only a \(\tfrac{1}{2}\)-\(MMS\) ex-post, without any envy-based fairness guarantees. 
{This construction leverages an optimal poly-time BoBW two-agent result that we discuss next and might be of independent interest.}
{Moreover, computing EFX allocations in polynomial time is an open problem even for three agents (even without any BoBW requirements). Within the polynomial-time regime, EF1 is therefore the strongest standard ex-post envy notion against which to compare. Our result lies on the Pareto frontier of BoBW guarantees: \citet{azizfreemannew} achieve ex-post EF1 together with ex-ante EF, whereas we achieve ex-post \EIMMX{} together with ex-ante proportionality. The two results are incomparable, since \EIMMX{} and EF1 are incomparable as ex-post notions - neither implies the other.}

For two agents, we present a polynomial-time algorithm (FPTAS) that, for any $\varepsilon>0$, simultaneously guarantees \emph{ex-ante envy-freeness} (and hence proportionality) and \emph{ex-post EFX}, while ensuring that each agent receives at least a \((1-\varepsilon)\)-fraction of her maximin share (MMS).  
This matches the strongest possible envy-based guarantees in both ex-ante and ex-post regimes, and attains the optimal polynomial-time MMS approximation factor of \((1-\varepsilon)\), unless \(P=NP\).  
We discuss our results for two agents in \cref{sec:intro-2-agents} below.

\begin{table}[H]
\centering
\small
\renewcommand{\arraystretch}{1.3}
\begin{tabular}{@{}llllll@{}}
\toprule
Reference & $n$ & Ex-ante & Ex-post share & Other ex-post & Poly-time \\
\midrule
\cite{azizfreemannew}   & any & EF   & $\tfrac{1}{n}$-MMS & EF1  & yes \\
\cite{babaioff2021bobw} & any & PROP & $\tfrac{1}{2}$-MMS & none (may violate EF1) & yes \\
\cite{garg2024bestofbothworldsfairnessenvycycleeliminationalgorithm}
                         & $2$ & EF   & $\tfrac{2}{3}$-MMS$^{\ddagger}$ & EFX  & yes \\
\cite{bu2024bestofbothworldsfairallocationindivisible}
                         & $2$ & EF   & $\tfrac{2}{3}$-MMS$^{\ddagger}$ & EFX  & yes \\
\midrule
\cref{cor:main}           & $2$ & EF   & $(1-\varepsilon)$-MMS & EFX  & FPTAS \\
\cref{thm:three-agents-main}
                          & $3$ & PROP & $\tfrac{9}{10}$-MMS & EEFX $+$ IMMX & no$^{\dagger}$ \\
\cref{thm:bobw-poly-approx}
                          & $3$ & $(1-\varepsilon)$-PROP & $(\tfrac{9}{10}-\varepsilon)$-MMS & EEFX $+$ \EIMMX & FPTAS \\
\cref{thm:bobw-poly}      & $3$ & PROP & $(\tfrac{9}{10}-\varepsilon)$-MMS & \EIMMX & FPTAS \\
\bottomrule
\end{tabular}
\caption{BoBW guarantees for additive agents. Every guarantee holds for
every allocation in the support (ex-post) or in expectation (ex-ante).
The ``other ex-post'' column lists envy-based and hybrid guarantees:
EEFX additionally implies the share-based MXS guarantee
\cite{caragiannis2022new}, and IMMX requires each agent to receive
her MMS \emph{or} be EFX-satisfied. For $n=3$, no allocation -- even
ignoring ex-ante requirements -- can guarantee more than
$\tfrac{39}{40}$-MMS \cite{feige2021tight}.
$^{\dagger}$Polynomial-time given oracles for MMS partitions
(\cref{lem:poly-time-3}).
$^{\ddagger}$Implied by ex-post EFX for two agents \cite{amanatidis2018comparing}; 
\citet{garg2024bestofbothworldsfairnessenvycycleeliminationalgorithm} do not guarantee more than $\tfrac{37}{42}$-MMS, while \cite{bu2024bestofbothworldsfairallocationindivisible} do not guarantee more than $\tfrac{17}{20}$-MMS
-- see \cref{sec:no-min-bound}.}
\label{tab:comparison}
\end{table}

{\subsubsection{The \IME\ Criterion: Motivation and Interpretation}\label{sec:IME}

Our new fairness notion, denoted by \IME, aims to capture a balanced compromise between share-based and envy-based guarantees.  
An allocation is said to be \emph{\IMMX} if every agent either receives her \(MMS\) share or is \emph{EFX}-satisfied (and thus guaranteed at least \(4/7\)-MMS~\cite{amanatidis2018comparing}).  
The motivation for this notion stems from two fundamental observations.  
First, \(MMS\) allocations are not guaranteed to exist, even for three additive agents~\cite{Kurokawa2018}; hence, 
some agents may inevitably fail to get their \(MMS\) value. We would like to ensure that any agent who fails to receive her \(MMS\) value 
be compensated with the strongest feasible envy-based guarantee -- that of \emph{EFX}-satisfaction.  
Second, we believe that insisting on \emph{EFX} for all agents may be unnecessarily restrictive: agents who already receive high-valued bundles, above their \(MMS\) share, have little justification for complaining.

To illustrate, consider three agents and five items \(\{g_1, g_2, g_3, g_4, g_5\}\), with additive valuations defined by the following item values:
\[
\begin{array}{c|ccccc}
 & g_1 & g_2 & g_3 & g_4 & g_5 \\
 \hline
 v_1(g) & 100 & 101 & 2 & 0 & 0 \\[2pt]
 v_2(g) & 10 & 4 & 4 & 2 & 2 \\[2pt]
 v_3(g) & 10 & 4 & 4 & 2 & 2
\end{array}
\]

Consider the allocation \((\{g_1\}, \{g_2,g_3\}, \{g_4,g_5\})\), where agent~1 receives \(\{g_1\}\), agent~2 receives \(\{g_2,g_3\}\), and agent~3 receives the remaining items.  Both agent 2 and agent 3 are EFX-satisfied, just like in an EFX allocation. 
Agent~3 does not get her MMS, but that is consistent with the fact that agents cannot expect to always get their MMS, as MMS allocations do not always exist \cite{Kurokawa2018}.
The allocation is not EFX as agent~1 is not EFX-satisfied -- she envies agent~2 even after removing \(g_3\).
Yet, agent~1 receives value \(100\),  above her \(MMS\) value of \(2\), so it seems like any complaint she might raise about not being treated fairly (as she is not EFX-satisfied) is  unwarranted, as she receives a good bundle, one with value above her fair share (her MMS).  

When judged solely by \(MMS\) or solely by \(EFX\), this allocation might appear unfair, but under our \IMMX criterion it is deemed fair: every agent is either \(MMS\)-satisfied or \emph{EFX}-satisfied.

{\paragraph{IMMX beyond additive valuations.} Recently, two constructions of 3-agent instances (with monotone valuations) for which EFX does not exist have been presented. Interestingly, while these instances also do not have any MMS allocations, we show that they have \IMMX allocations, highlighting the attractiveness of this new fairness concept: it exists for some instances for which there are no EFX and no MMS allocations.  Specifically ~\citet{akrami2026counterexampleefxnge} presented instances with three agents with monotone valuations and eight goods, in which an EFX allocation does not exist, while ~\citet{mackenzie2026counterexamplesefxsubmodularsubadditive} show that 
the problem that EFX does not exist  persists even for weighted coverage functions (a subclass of submodular functions).
In \cref{immx-proof}, we show that the submodular instance admits no MMS allocation under any cardinal valuation consistent with the ordinal rankings of the original example. For the \cite{akrami2026counterexampleefxnge} instance, we verify by exhaustive search that no MMS allocation exists. Nevertheless, we show that both instances admit an \IMMX allocation.
To the best of our knowledge, these are the first examples exhibiting a simultaneous failure of both EFX and MMS.
As \IMMX exists for these instances, it still has the potential of being generally feasible (while MMS and EFX are not).}

}

{
\paragraph{IMMX for chores.}
Our IMMX existence result extends beyond the setting of goods 
to the setting of \emph{chores} -- indivisible items that impose costs on the agents receiving them rather than positive value. Specifically,  in \Cref{sec:immx-chores} we show that the construction underlying \Cref{thm:three-agents-main} can be adapted to the setting of chores. That is, we establish that for three agents with additive cost functions, an IMMX allocation always exists: every agent either incurs a cost of at most her maximin share, or is EFX-satisfied. 
We note that for chores, the IMMX criterion carries added value beyond its appeal for goods.
MMS allocations for chores are known not to exist already for three agents \cite{ARSW17}. As for EFX, while EFX allocations for goods are known to exist for three agents \cite{chaudhury2024efx}, the existence of EFX allocations for chores remains open even for three agents. Moreover, it was recently shown that they need not exist for four or more agents \cite{he2026efxadditivechoresnonexistence}.
Therefore, for three agents with chores, while each of the two fairness benchmarks of MMS and EFX are currently unattainable (either infeasible or not known to exist), IMMX is guaranteed to exist.
}

We view \IMMX\ as a natural and robust middle ground between the two dominant fairness perspectives of share-based and envy-based fairness.  
It extends the applicability of fair allocations beyond strict \(MMS\) or strict \(EFX\) requirements, enlarging the set of allocations that can be justifiably regarded as fair, while preserving strong and interpretable guarantees for all agents.

\subsubsection{Near-Optimal BoBW for Three Additive Agents}\label{sec:intro-3-agents}

Our main focus is the problem of fair allocation of indivisible goods among three agents with additive valuations, for which we establish the existence of a near-optimal BoBW distribution.

For three agents,  while EFX allocations are known to exist \cite{chaudhury2024efx}, no polynomial-time algorithm is known, and beyond three agents existence is unresolved. We therefore turn to the weaker -- but still powerful -- notion of \emph{Epistemic EFX (EEFX)} \cite{caragiannis2022new}. EEFX retains envy-based structure, implies share-based fairness (MXS), and guarantees at least $2/3$-MMS for three agents \cite{amanatidis2018comparing}. Our aim is to push these guarantees higher, and as close as possible to the best possible guarantees.

\citet{garg2024bestofbothworldsfairnessenvycycleeliminationalgorithm} show that even for three agents with additive valuations,  
a natural random variant of the envy-cycle elimination algorithm may fail to guarantee ex-ante proportionality with ex-post EEFX. In contrast, we establish the existence of a distribution that is ex-ante proportional and whose every allocation in the support is simultaneously \emph{EEFX} and guarantees each agent at least a $\frac{9}{10}$-fraction of her MMS. {Moreover, we guarantee \IMMX: each agent in the allocation either receives her MMS share, or is EFX-satisfied. Since as mentioned, MMS allocations do not always exist for three agents, this is essentially the best envy-based guarantee one can offer to the agent not getting her MMS (the stronger notion of envy-freeness is clearly impossible to guarantee)}. {We summarize these guarantees formally in the following theorem:}

\begin{restatable}{theorem}{threemainthm}
\label{thm:three-agents-main}
Consider settings with three agents with additive valuations over a set $M$ of indivisible goods.

There exists a procedure that given three additive valuations $(v_1,v_2,v_3)$, constructs a probability distribution \( \mu \) over at most six deterministic allocations 
such that:
{
\begin{enumerate}
    \item The distribution \( \mu \) is \emph{ex-ante proportional}.
    \item In every deterministic allocation in the support of \( \mu \): 
    \begin{enumerate}
        \item One agent is EFX-satisfied and receives her proportional share.
        \item One agent is EFX-satisfied and receives at least 9/10 of her MMS.
        \item One agent is EEFX-satisfied and receives at least her MMS.
    \end{enumerate}

\end{enumerate}
}
\end{restatable} 
We note this implies that every deterministic allocation in the support of $\mu$ is \IMMX, and in each allocation each agent receives at least $\frac{9}{10}$ of her MMS share, and is also EEFX-satisfied.
Our result is the first BoBW result to guarantee EEFX for more than 2 agents. Additionally, our algorithm is combinatorial (in contrast to the LP-based algorithm of \cite{babaioff2021bobw}).
Our result shows that BoBW distributions that are ex-ante proportional and are supported on $9/10$-MMS allocations exist. This is near-optimal as the best fraction of the MMS cannot be larger than $39/40$ \cite{feige2021tight}. 
Our $9/10$ approximation of the MMS substantially improves on \citet{babaioff2021bobw} guarantee of $\frac{1}{2}$-MMS.\footnote{We note that \cite{babaioff2021bobw} provide a polynomial-time algorithm guaranteeing ex-ante proportionality and ex-post $\tfrac{1}{2}$-MMS. 
We later extend our result and show that ex-ante proportionality and ex-post $(\tfrac{9}{10}-\varepsilon)$-MMS can also be achieved in polynomial time.}
Moreover, we note that \cite{caragiannis2022new} show that MMS implies EEFX, but this implication holds only when there are no items with zero utility. In this paper we allow for items of zero utility, and as MMS does not necessarily imply EEFX in that general setting\footnote{To illustrate, consider three agents with identical additive valuations over five items with values $(10,10,2,0,0)$. Allocating the item of value 2 to agent 1 ensures that she receives at least her MMS value. However, agent 1 is not EEFX-satisfied: no redistribution of the remaining four items among the other agents yields an allocation that eliminates agent 1's envy up to any item (including zero-valued ones).}, it is non-trivial to guarantee EEFX on top of MMS for our model.

This result is near optimal from several other perspectives; first, the distribution is ex-ante proportional, and it is impossible to guarantee every agent more than her proportional share ex-ante. 
Second, if an agent does not get her MMS she is guaranteed the best-possible envy-based guarantee of EFX satisfaction {-- and thus the allocation is \IMMX}. Moreover, two out of the three agents are EFX-satisfied (while the third is EEFX-satisfied). 
Finally, it is desirable to satisfy ``equal treatment of equals'' (agents with identical valuations are treated the same). Clearly this cannot be guaranteed ex-post, so the best we can hope for is satisfying this ex-ante, which our procedure does. We also note that the procedure is \emph{robust:} all fairness guarantees hold for every agent, independently of the behavior or reports of others.

\paragraph{Computation.}  
While \cref{thm:three-agents-main} provides a constructive proof of existence, the construction does not run in polynomial time.
It is important to note that the \textbf{only} non-polynomial component of the construction is the computation of MMS partitions for each of the three given valuations. 
The entire construction runs in polynomial time if for each of the three valuations it is given access to an oracle that can compute MMS partitions.  
In particular, if all values are {integers} bounded by a polynomial in the number of goods $|M|$, 
then the BoBW result is obtained as stated in \cref{thm:three-agents-main} in polynomial time - for further details, see
\cref{sec:fptaspolynomial}.

{
\paragraph{Polynomial-Time Approximation via FPTAS} We note that for any $\varepsilon>0$ we can employ an FPTAS to compute $(1-\varepsilon)$-MMS partitions, thereby making the entire algorithm polynomial-time, while preserving the EEFX guarantee, the {\EIMMX guarantee,} and also ensuring a $1-\varepsilon$ fraction of all share-based guarantees (both ex-ante and ex-post). 
See \cref{thm:bobw-poly-approx} for a formal statement of this result. 

}

In this FPTAS result (\cref{thm:bobw-poly-approx}),  ex-ante proportionality is not guaranteed, but only \((1-\varepsilon)\) of it. We show that by dropping the EEFX requirement it is  possible to restore \textbf{exact} ex-ante proportionality, while maintaining all ex-post share-based guarantees within a factor of $1-\varepsilon$  {(including guaranteeing \EIMMX)}. This result, in particular, improves upon the BoBW result obtained by \cite{babaioff2021bobw}, which guarantees ex-ante proportionality, but only $\frac{1}{2}$-MMS. We formalize this result in \cref{thm:bobw-poly}.
This result bridges the gap between existential and computational guarantees for the case of three agents. 
It establishes the existence of an FPTAS that guarantees exact ex-ante proportionality alongside strong share-based ex-post fairness guarantees (including \EIMMX and near-optimal approximations of the MMS).

We now turn to the case of two additive agents, and present a polynomial-time algorithm (FPTAS) with optimal fairness guarantees.  
This two-agent optimal BoBW result {serves} as a building block in transforming the BoBW algorithm for three additive agents, into a polynomial-time BoBW algorithm that guarantees \emph{exact} ex-ante proportionality (\cref{thm:bobw-poly}).

\subsubsection{Optimal Poly-time BoBW for Two Additive Agents}\label{sec:intro-2-agents}

For two additive agents we present a simple, polynomial-time (FPTAS) \emph{best-of-both-worlds} algorithm with multiple optimal guarantees. 
Prior polynomial-time procedures \cite{bu2024bestofbothworldsfairallocationindivisible, garg2024bestofbothworldsfairnessenvycycleeliminationalgorithm} achieve ex-post EFX together with ex-ante envy-freeness (which is equivalent to ex-ante proportionality in this setting), but we observe that any one of them may lose a constant fraction of the maximin share (MMS).\footnote{Loss of more than 11  percent, see \cref{sec:no-min-bound} for more details.}
We show that a slight modification of the algorithm presented in \cite{bu2024bestofbothworldsfairallocationindivisible} avoids this constant loss, and it obtains an optimal MMS poly-time approximation of $1-\varepsilon$, on top of these two guarantees.

Specifically we show:

\begin{restatable}{theorem}{twoagentsmain}
    \label{cor:main} 
Fix any $\varepsilon>0$. Consider settings with two agents who have additive valuations over a set of indivisible goods.
There exists a fully polynomial-time approximation scheme (FPTAS) 
that outputs a distribution over at most two allocations satisfying the following properties:
\begin{itemize}
    \item \textbf{Ex-ante envy-free:} The randomized allocation is ex-ante envy free (and thus ex-ante proportional).

    \item \textbf{Ex-post EFX:} Every allocation in the support of the distribution is EFX. 
    
    \item \textbf{Ex-post value guarantee:} Every allocation in the support of the distribution guarantees each agent \( i \) a value of at least $(1-\varepsilon)$ of $MMS_i$. 

\end{itemize}
\end{restatable}

Note that the classical cut-and-choose protocol with a random cutter achieves ex-ante envy-freeness (proportionality), ex-post EFX, and exact MMS, but is not polynomial-time computable (unless $P=NP$), as computing the MMS is NP-hard. 
Moreover, while we can simulate, in polynomial time, the randomized cut-and-choose protocol using any given FPTAS to compute $(1-\varepsilon)$-MMS partitions, this may only guarantee an expected value of $(1-\varepsilon)$ fraction of the proportional share, and it \textbf{does not} guarantee ex-ante envy freeness, nor ex-post EFX\footnote{See \cref{subsec:ptas-fail} for an example.}. 
In contrast, our algorithm runs in polynomial time and achieves the same guarantees of ex-ante envy-freeness and ex-post EFX, and only suffers an inevitable small loss in the MMS approximation factor. 

{{These} guarantees are essentially optimal: we simultaneously achieve ex-ante envy-freeness and ex-post \emph{EFX} -- the strongest widely studied feasible notions of fairness in their respective regimes -- while ensuring each agent receives at least a \((1-\varepsilon)\) fraction of the MMS value ex-post.  
No algorithm can guarantee more than the proportional share ex-ante (under identical valuations), and our algorithm attains this bound; moreover, its running time essentially matches that of the best available FPTAS, up to an additional polynomial term independent of~\(\varepsilon\).}

\subsection{Our Techniques}

{We focus on our techniques for the case of three agents}. In a previous attempt to combine ex-ante proportionality with ex-post epistemic EFX, \citet{garg2024bestofbothworldsfairnessenvycycleeliminationalgorithm} showed that a natural randomized variant of the \emph{Envy-Cycle Elimination (ECE)} algorithm -- the tool used to establish the existence of EEFX allocations and adopted by \citet{feldman2024breaking} for their BoBW guarantees -- fails to ensure EEFX even for three agents with additive valuations. On the share-based side, \citet{babaioff2021bobw} showed that the techniques behind the envy-centered BoBW algorithm of \cite{azizfreemannew} guarantee only a $1/n$ fraction of the MMS, ruling them out for stronger ex-post share-based guarantees.

To overcome these limitations, we introduce a structural technique that combines the reasoning of EEFX with the robustness of MMS partitions. We first build, for one agent, an MMS partition into three bundles that is also EFX under her own valuation; whichever of the three she receives, she attains her MMS value and is EEFX-satisfied, regardless of how the remaining items are split between the other two agents. The key insight is that the epistemic nature of EEFX lets us \emph{repartition} the remaining items without violating EEFX: once one agent's MMS and EEFX constraints are met, the other bundles can be redivided arbitrarily to satisfy the remaining agents. Concretely, from an MMS-and-EEFX partition for one agent we repartition the other two bundles so that one remaining agent receives her proportional share and the other receives at least \(\tfrac{9}{10}MMS\) while being EFX-satisfied. To achieve ex-ante proportionality, we construct six such allocations in three pairs -- each pair MMS and EEFX for one agent and desirable for the other two. Averaging over the pairs yields an ex-ante proportional distribution in which every deterministic allocation is EEFX and near-optimal for MMS, while any agent below her MMS value remains \emph{EFX}-satisfied.

Our construction is related to the ``coarse atomic partial search'' algorithm of \cite{feige2022improvedmaximinfairallocation, amanatidisapprox2017, GOURVES201950}, but those algorithms guarantee fairness only for a realized allocation and are thus inherently \emph{ex-post}. Seeking guarantees over the distribution of allocations (\emph{ex-ante}) demands different structural invariants, yielding both a conceptually different \textbf{BoBW} framework and non-trivial modifications of the algorithm itself. 

To the best of our knowledge, this is the first technique to achieve a near-optimal envy-based guarantee (EEFX) together with a strong MMS guarantee\footnote{EEFX already implies a mild $\frac{2}{3}$ MMS approximation for three agents, but we aim much closer to the upper bound of $39/40$.}, while giving each agent either her MMS share or EFX satisfaction, all within a unified BoBW framework that attains ex-ante proportionality.

We note that our algorithm shares similarities with the three-agent cake-cutting procedure introduced by Selfridge and Conway (see \cite{robertson1998cake}), mainly in the idea of role assignment. 
However, their algorithm quickly introduces concepts that are specific to the setting of divisible goods, for which we find no clear analog in the context of indivisible items; we therefore depart from their approach and develop new techniques tailored to the indivisible setting.

{
\paragraph{Achieving Exact Ex-Ante Proportionality in Polynomial Time.}
As demonstrated in \cref{subsec:ptas-fail}, replacing exact \(MMS\) partitions with \((1-\varepsilon)\)-approximations does not necessarily guarantee ex-ante proportionality.  
To address this challenge, {we utilize our two-agent algorithm} (\cref{cor:main}) to overcome cut-and-choose errors, and additionally, introduce {\emph{adoption procedures}}, which allow agents to ``adopt'' a three-way partition generated throughout the algorithm and yield a better value of the least-valuable bundle than what their own FPTAS generated. The bundles of that adopted partition are carefully reassigned to the three agents, and we {meticulously} show that such reassignments ensure ex-ante proportionality while maintaining all other share-based guarantees up to a factor of \((1-\varepsilon)\), and \emph{EFX} satisfaction for the agent who receives {less than her} \((1-\varepsilon)\cdot MMS\) value.
}

\subsection{Additional Related Work}\label{sec:related}

The body of research on fair allocation of resources is very extensive, so we only discuss the most related literature. For comprehensive overviews, the reader is referred to the books \cite{Brams_Taylor_1996}, \cite{moulin2003} and the survey \cite{survey}. In the remainder of this section, we confine our discussion to the works most closely related to the present study.

\paragraph{EF1 and EFX}
Envy-free up to one good (EF1) originated in the work of \citet{lipton2004approximately}, who proved that such allocations always exist for any number of agents with monotone valuations, and was later formalized by Budish \cite{budishmms2011}. 
The stronger notion of envy-free-up-to-any-good (EFX) was introduced by \citet{caragiannis2019unreasonable}. 
\citet{PlautRoughgarden2020} established that EFX allocations always exist for two agents with monotone valuations, while \citet{chaudhury2024efx}\ recently showed existence for three additive agents. 
However, no polynomial-time algorithm is known even for this case, and existence remains open for four or more agents. 
Several relaxations have been explored, including EFX with charity \cite{chaudhury2021little}, restricted valuation domains \cite{ghosal2025}, and near-EFX variants \cite{amanatidis2024pushing}.

\paragraph{MMS approximation.}
The \emph{maximin share (MMS)} criterion, introduced by Budish \cite{budishmms2011}, captures the value an agent can secure by partitioning goods into $n$ bundles and receiving the least valuable one. 
Exact MMS allocations do not always exist, even for three additive agents \cite{Kurokawa2018}, motivating a rich line of work on approximate MMS guarantees. 
Early results established $2/3$-MMS allocations \cite{Kurokawa2018}, later improved through a sequence of refinements \cite{ garg2019approximating, barman2020approximation, akrami2023simplification, heidari2025improvedmaximinshareguarantee}. 
The current best existential guarantee for general $n$ is $\frac{7}{9}$-MMS \cite{huang2025fptas79approximationmaximinshare}, while for $n=3$, Feige and Norkin \cite{feige2022improvedmaximinfairallocation} obtained an even stronger bound -- $\frac{11}{12}$-MMS (we also note that Feige and Norkin analyzed a slightly weaker algorithm of \cite{amanatidisapprox2017,GOURVES201950}, and showed that it constructs $\frac{9}{10}$-MMS allocations for three agents. We use their analysis in order to prove our guarantees). 
\citet{feige2021tight} established complementary upper bounds, showing that for $n=3$ no allocation can guarantee more than $39/40$-MMS, and for $n \ge 4$, no allocation exceeds $(1 - n^{-4})$-MMS. 
Additionally, \citet{amanatidis2018comparing} proved that EFX implies $2/3$-MMS for $n=3$ and $4/7$-MMS for $n \ge 4$, illuminating the quantitative connection between envy-based and share-based fairness.

\paragraph{Best-of-Both-Worlds fairness.}
The study of \emph{Best-of-Both-Worlds (BoBW)} fairness seeks randomized allocations that reconcile two complementary goals: fairness of the distribution (ex-ante) and fairness of realized outcomes (ex-post). \citet{azizfreemannew} formally introduced the term ``Best-of-Both-Worlds'' in the context of indivisible goods and developed a polynomial-time algorithm for additive agents that guarantees ex-ante envy-freeness (EF) and ex-post EF1. 
Beyond envy-based fairness, \citet{babaioff2021bobw} extended the framework to share-based fairness, providing a polynomial-time algorithm that ensures ex-ante proportionality and ex-post $1/2$-MMS guarantees. 
This approach captures certain share-based fairness benchmarks but sacrifices envy-based structure, 
and its $1/2$-MMS bound remains far from tight. 
Subsequent work has explored BoBW algorithms for specific valuation domains, such as subadditive \cite{feldman2024breaking}, the case of indivisible and mixed goods \cite{bu2024bestofbothworldsfairallocationindivisible}, matroid rank valuations \cite{Babaioff_Ezra_Feige_2021}, and the case of agents with arbitrary entitlements \cite{hoefer2024best}. 
Our work continues this line by tightening the tradeoff between ex-ante proportionality and ex-post envy-based guarantees, and by aligning these with stronger MMS-type bounds.

\paragraph{Epistemic EFX and related notions.}
Epistemic EFX (EEFX), introduced by \citet{caragiannis2022new}, relaxes EFX by requiring that each agent be ``EFX-satisfied'' with respect to some reshuffling of the other agents’ bundles. Intuitively, EEFX {demands} that, while keeping her own bundle fixed, agent~$i$ could rearrange the remaining items so that she would no longer envy any other bundle up to any good. 
Moreover, although the existence of EFX allocations for more than three agents remains open, \citet{akrami2025general} recently showed that EEFX allocations always exist for all monotone valuation functions.
\citet{caragiannis2022new} also introduced the closely related concept of the \emph{Minimum EFX Share (MXS)}, proving that every EEFX allocation is also MXS. 

\paragraph{Computational complexity.}
Computing the value of the MMS (or an MMS allocation) is known to be  NP-hard \cite{Woeginger1997APA}, though fully polynomial-time approximation schemes (FPTAS) exist \cite{Mittalsantaclause, woegingerFPTAS}. 
Regarding EFX allocations, though existence is guaranteed for three agents \cite{chaudhury2024efx} their algorithm to compute such allocations is not polynomial, and there is no known polynomial-time algorithm for this problem. Our work addresses these computational barriers by designing FPTAS algorithms that achieve near-optimal approximations to these fairness criteria in the BoBW framework.

\subsection{Organization}

\cref{sec:prelim} introduces notation and fairness notions. \cref{sec:three-agents-existence} proves near-optimal BoBW existence results for three agents, \cref{sec:two-agents} develops our optimal poly-time BoBW result for two agents,  \cref{sec:three-agents-poly} presents our polynomial-time BoBW construction for three agents, and \cref{sec:conclusion} concludes with a discussion, and lists some open questions.

\section{Preliminaries}\label{sec:prelim}
\subsection{The Model and Basic Notations}

We consider the problem of fair allocation of a set $M$ of $m$ indivisible goods to a set $N$ of $n$ agents, each with a normalized additive valuation function $v_i:2^M\rightarrow \mathbb{R}_{\geq 0}$ over $M$.
That is, $v_i(\emptyset)=0$ (normalization), $v_i({g})\geq 0$ for every $g\in M$ (items are goods),  
and for any $S\subseteq M$ it holds that $v_i(S)=\sum_{g\in S} v_i(\{g\})$ (additivity). 
A \emph{partition} $\mathcal{P}$ (of a set of goods $M$ to $n$ agents) is a collection $\{P_1,P_2,..,P_n\}$  of $n$ disjoint subsets of $M$, such that $\cup_{i=1}^n P_i=M$. An \emph{allocation} $\mathcal{A}$ (of a set of goods $M$ to $n$ agents) is an \textbf{ordered} partition $(A_1,A_2,..,A_n)$ of $M$, where $A_i\subseteq M$ is the bundle given to agent $i$. We refer to an allocation as \emph{partial} if only a subset $S\subseteq M$ of the goods are allocated to the set of agents $N' \subseteq N$. 
A \emph{randomized allocation} is a distribution over allocations. 

\subsection{Fairness Criteria}
 
We next describe several fairness notions we will consider.
The \emph{proportional share} of agent $i\in N$ with an additive valuation $v_i$ when there are $n$ agents, is defined to be $PROP(v_i,n)=\frac{v_i(M)}{n}$. When $n$ is clear from the context we denote $PROP_i= PROP(v_i,n)$. 
An \emph{allocation is proportional} if every agent gets a bundle that she values at least as much as her proportional share. For $\alpha >0$, an allocation is \emph{$\alpha$-proportional} if every agent $i$ receives a bundle she values at least as  much as $\alpha\cdot PROP(v_i,n)$. A randomized allocation is \emph{ex-ante proportional} if the \emph{expected value} of every agent is at least her proportional share. 

Budish \cite{budishmms2011}  has suggested the maximin share (MMS) as an appropriate relaxation of the proportional share for indivisible goods:
\begin{definition}[Maximin share (MMS)]\label{def:mms}
The \emph{maximin share (MMS)} of agent $i$ with a valuation $v_i$ over the set of items $M$, and a given number of agents $n$ is defined to be:
\[
\text{MMS}(v_i,n) 
:= \max_{(Z_1,\dots,Z_n) \in \Gamma_n} \min_{j \in [n]} v_i(Z_j),
\]
where $\Gamma_n$ is the set of all partitions 
of the set of goods $M$ into $n$ bundles. When $n$ is clear from the context we denote $MMS_i= MMS(v_i,n)$.

An allocation is an \emph{MMS allocation} if every agent $i$ receives a bundle of value at least $\text{MMS}(v_i,n)$. 
An allocation is an \emph{$\alpha$-MMS allocation} if every agent $i$ receives a bundle of value that is at least $\alpha\cdot \text{MMS}(v_i,n)$.
\end{definition}

We say that a partition $\mathcal{P}=\{P_1,\ldots,P_n\}$ provides an agent with valuation $v$ a \emph{better worst-case-value guarantee} than partition $\mathcal{Q}=\{Q_1,\ldots,Q_n\}$ if $$min_{P\in \mathcal{P}}\ v(P) > min_{Q\in \mathcal{Q}}\ v(Q).$$ 
Clearly, if $\mathcal{P}$ is an MMS partition with regard to $v$ then it provides a better worst-case-value guarantee than any partition which is not an MMS partition.

\citet{caragiannis2019unreasonable} have introduced the notion of envy-free-up-to-any-item (which they called EFX). Their definition requires envy to be eliminated whenever  a positive-valued item is removed (but not when a zero-valued item is removed). We will use $EFX_+$ for this notion, to differentiate between this notion and the stronger notion of $EFX$ introduced by \citet{PlautRoughgarden2020}. This stronger notion requires envy to be eliminated even when removing zero-valued goods. We will use $EFX$ for this stronger notion.  As we prove positive results for this stronger notion, this only strengthens our result. 
\begin{definition}[Envy-freeness-up-to-any-item (EFX) {\cite[Definition 2.3]{PlautRoughgarden2020}}] \label{def:EFX}
Agent $i$ with a valuation $v_i$ is \emph{EFX-satisfied} by allocation $\mathcal{X} = (X_1, \dots, X_n)$ if for any other agent $k$, and any item $g \in X_k$, it holds that 
\[
v_i(X_i) \geq v_i(X_k \setminus \{g\}).
\]
The allocation $\mathcal{X}$ is an \emph{EFX allocation} if every agent is EFX-satisfied by $\mathcal{X}$.

\end{definition}

We say that a bundle $X$ \emph{EFX-dominates} a bundle $Y$ for agent $i$ with valuation $v_i$ if for every $g\in Y$ it holds that $v_i(X) \geq v_i(Y \setminus \{g\})$. Note that agent $i$ is EFX-satisfied by allocation $\mathcal{X} = (X_1, \dots, X_n)$ if her bundle $X_i$ EFX-dominates every other bundle in the partition $\mathcal{X}$.

\begin{observation}\label{obs:top-is-EFX}
Given an allocation $(X_1, X_2, \ldots, X_n)$, if agent~$i$ values her own bundle as the highest among all bundles, i.e.,
\[
v_i(X_i) \ge \max_{k \in [n]} v_i(X_k),
\]
then agent~$i$ is EFX-satisfied.
\end{observation}

The notions of Epistemic EFX (EEFX) and of $MXS$ were defined in \cite{caragiannis2022new} using $EFX_+$ (the weaker EFX notion). 
To differentiate them from the stronger notions that we use, we denote them as $EEFX_+$ and $MXS_+$. 
We use the notions of $EEFX$ and $MXS$ with respect to the $EFX$ definition (\cref{def:EFX}) that requires envy to be eliminated even when zero-valued goods are removed. 

\begin{definition}[Epistemic EFX (EEFX) and EEFX certificates]
For a fair division instance, an allocation $\mathcal{X} = (X_1, X_2, \dots, X_n)$ (of $M$ to $n$ agents) is called \emph{Epistemic EFX (EEFX)} if for every agent $i \in [n]$, there exists an allocation $\mathcal{Y} = (Y_1, Y_2, \dots, Y_n)$   (of $M$ to $n$ agents) such that $Y_i = X_i$ and agent $i$ is EFX-satisfied by $\mathcal{Y}$. 

We refer to such an allocation $\mathcal{Y}$ as an \emph{EEFX certificate} of agent $i$ for bundle $X_i$ (with respect to the set of items $M$ and $n$ agents).

We say that an agent is \emph{EEFX-satisfied by bundle $X_i$} 
if there exists an EEFX-certificate for that bundle (with respect to the set $M$ and $n$ agents). A bundle $X_i$ is \emph{EEFX-satisfactory} for an agent $i$ if the agent is EEFX-satisfied by that bundle (with respect to the set of items $M$ and $n$ agents).

\end{definition}

Note that if agent $i$ is EFX-satisfied by an allocation $\mathcal{X} $, then she is EEFX-satisfied by the bundle $X_i$ she receives in that allocation, with $\mathcal{X}$ being the certificate for $X_i$. 

\begin{definition}[Minimum EFX share and MXS allocations]
For a fair division instance with $n$ agents, define the set $\mathbb{EFX}(v_i,n)$ to be the set of allocations in
which agent i is EFX-satisfied:
\[
\mathbb{EFX}(v_i,n) := \left\{(Z_1, \dots, Z_n) \in \Gamma_n \;:\; agent\ i\ is\ EFX-satisfied\ by\ Z_i \right\}.
\]

The \emph{minimum EFX share} of agent $i$, denoted by $\text{MXS}(v_i,n)$, is defined to be 
\[
\text{MXS}(v_i, n) := \min_{(Z_1,...,Z_n) \in \mathbb{EFX}(v_i,n)} v_i(Z_i).
\]
When $n$ is clear from the context we denote $MXS_i =MXS(v_i,n)$.
An allocation $(Z_1, \dots, Z_n)$ is called an \emph{MXS allocation} if $v_i(Z_i) \geq \text{MXS}_i$ for every agent $i \in [n]$.
\end{definition}

\citet{caragiannis2022new} have shown that an \emph{$EEFX_+$ allocation} is also an \emph{$MXS_+$ allocation}.
\begin{theorem}[$EEFX_+$ $\Rightarrow$ $MXS_+$ (\cite{caragiannis2022new} Theorem 1)]
An \emph{$EEFX_+$ allocation} is also an \emph{$MXS_+$ allocation}. 
\end{theorem}

It is simple to observe that the same argument implies the following observation:
\begin{observation}\label{obs:eefx-mxs}
    [$EEFX$ $\Rightarrow$ $MXS$]
An \emph{EEFX allocation} is also an \emph{MXS allocation}. 
\end{observation}

By definition, it is immediate that an \emph{EFX allocation} is also an \emph{$EFX_+$} allocation, and that an EEFX allocation is also an \emph{$EEFX_+$} allocation.
\begin{observation}
[$EFX \Rightarrow EFX_+$]
An \emph{EFX allocation} is also an \emph{$EFX_+$} allocation.
\end{observation}

\begin{corollary}
    [$EEFX \Rightarrow EEFX_+$] An \emph{EEFX allocation} is also an \emph{$EEFX_+$} allocation.
\end{corollary}

\section{Near-Optimal BoBW for Three Additive Agents}\label{sec:three-agents-existence}
In this section, {we study BoBW fairness when there are three agents with additive valuations over goods.}

{Our main computational result is a polynomial-time algorithm, built on an existential} theorem establishing a Best-of-Both-Worlds (BoBW) distribution for three agents with additive valuations: a distribution that is ex-ante proportional and has the following ex-post guarantees: EEFX and $\frac{9}{10}$-MMS. Moreover, every allocation in its support is \IMMX: each agent either receives a bundle worth at least her MMS value, or is EFX-satisfied by the allocation.

\threemainthm*

Since MMS allocations might not exist, at least one agent may fail to obtain her MMS value. 
In our construction, we ensure that in any allocation, \textbf{at most one} agent does not receive her MMS share\footnote{We note that it is not always possible to "spread the pain" and instead of one agent suffering a ``big drop'' in her MMS fraction, have multiple agents each suffering a ``smaller drop'' in their MMS fraction. Indeed, 
\cite{feige2021tight} presents an example with three agents in which one agent does not get her MMS (but only a $\tfrac{39}{40}$ fraction of it), while the other agents receive their full MMS. In this example, any improvement to that agent's MMS fraction always results in another agent suffering the same drop in her MMS fraction -- i.e, there is no way to ``spread the pain".}, and to compensate, we guarantee that she is EFX-satisfied. 
This constitutes the strongest possible envy-based guarantee for her, as achieving envy-freeness would require she receives at least her proportional share (and therefore her MMS).
{Moreover, across the six allocations we generate, each agent receives less than her MMS value with probability at most $\tfrac{1}{3}$, which is essentially the best possible; MMS allocations do not exist and in every allocation at least one agent must fall short of her MMS. 
Each agent also receives her proportional share (ex-post) with probability at least $\tfrac{1}{3}$, which is the best attainable bound (for example -- the case where there is only one item).
}

Additionally, we note that as every supporting allocation is EEFX, then every  supporting allocation guarantees each agent her \emph{MXS} share (this follows directly from \Cref{obs:eefx-mxs}).
We further show that our algorithm also guarantees each agent her \emph{RMMS} share (a share introduced by \citet{feige2025residualmaximinshare}), which may exceed \(\tfrac{9}{10}\) of her MMS share (thus strengthening the ex-post value guarantee). See \cref{sec:RMMS} for the RMMS definition and our RMMS-related results.

We now turn to the computational perspective. 
While the result provides a constructive proof of existence, the construction does not run in polynomial time, since it relies on computing MMS partitions -- which is NP-hard.
It is important to note that the \textbf{only} non-polynomial component of the construction is the computation of MMS partitions for each of the three given valuations.  We thus think of our algorithm as one that interacts with three oracles, one for each valuation, where each  \emph{oracle} is able to output MMS partitions for one of these valuations. 
We show that with such oracles we can compute our distribution in polynomial time.
This approach is particularly meaningful in settings where MMS partitions can be computed efficiently (e.g., when the number of goods is not too large).
Additionally, NP-hardness is a notion that captures the difficulty of  computation ``in the worst-case''. Yet, for some specific valuations, the problem of computing  MMS partitions  might be ``easy''. 
For example, when the value of the set $M$ is bounded by $poly(|M|)$ (another simple example is when the agent views all items as identical), {see \cref{sec:fptaspolynomial}}.

The next lemma formalizes that, given these MMS oracles, the corresponding BoBW distribution can be computed in polynomial time.
\begin{restatable}{lemma}{polytimethree}\label{lem:poly-time-3}
Suppose that for each agent $i$ there exists an oracle capable of computing the MMS for $v_i$ when
\begin{itemize}
    \item the set of items is $M$ and $n=3$, and  
    \item the set of items is a given $S\subseteq M$, and $n=2$.
\end{itemize}
    Then the probability distribution $\mu$ can be computed in polynomial time.
\end{restatable}

We prove \cref{lem:poly-time-3}, and provide a computational analysis of our algorithm in  \cref{sec:timecomplex}. 

The rest of this section is structured as follows.

We begin by introducing our main {existential} result (\cref{thm:three-agents-main}) using \cref{lem:main-lemma}, which establishes the ex-post properties of each supporting allocation.  
We then present the main algorithm (\cref{alg:three-agent-construction}) used in the proof of \cref{lem:main-lemma} and provide the corresponding analysis. Finally, we present a poly-time variant (an \emph{FPTAS}) of our main theorem in \cref{subsec:FPTAS}.

\subsection{Proof of \cref{thm:three-agents-main}} 

In this section, we present our main {existential} result, \cref{thm:three-agents-main}. 

We begin with some central definitions that relate to the combination of EFX-satisfaction and of $(1-\varepsilon)$-MMS. For our main existence result we focus on the $\varepsilon=0$ case, while our polynomial time algorithms (in \cref{subsec:FPTAS} and \cref{sec:three-agents-poly}) will use $\varepsilon>0$, so we present these definitions for a general \(\varepsilon\geq  0\).

\begin{definition}[$\mathcal{MMS^{(\varepsilon)}-EFX}$ {Partition}] 
We say that a partition $\{P_1,...,P_n\}$ of $M$ to $n$ bundles is an \emph{$MMS^{(\varepsilon)}-EFX$ partition} for agent $i$ with a valuation $v_i$ if for every bundle $P\in \{P_1,...,P_n\}$ it holds that $$v_i(P) \geq  (1-\varepsilon) \cdot MMS{(v_i, n)}$$ and moreover, agent $i$ is EFX-satisfied with bundle $P$.
We use \emph{$\mathcal{MMS^{(\varepsilon)}-EFX}^i(n)$} to denote such a partition, and denote  \emph{$\mathcal{MMS-EFX}^i=\mathcal{MMS}^{(0)}\mathcal{-EFX}^i(n)$} when $n$ is clear from the context and $\varepsilon=0$ .
\end{definition}

{An $\mathcal{MMS^{(\varepsilon)}-EFX}$ partition always exists for every normalized and monotone valuation (as \cite{PlautRoughgarden2020} has shown that the leximin++ allocation for $v_i$ is an $MMS-EFX$ allocation for $i$).
In \Cref{sec:efx-reallocation}, we present a result based on the work of \cite{PlautRoughgarden2020} showing that, given an additive valuation and an MMS partition, there exists a polynomial-time algorithm that constructs an allocation forming an $MMS-EFX$ partition for the agent. 

The notion of $\mathcal{MMS^{(\varepsilon)}-EFX}$ partition defined above, presents central properties (EFX and MMS approximation) of a partition with respect to one valuation. 
We next move to define a new fairness concept of an \emph{allocation} with respect to valuations for all agents, a notion that combines share-based and envy-based guarantees, and requires that an agent that does not get a ``high enough value'' (a $(1-\varepsilon)$ fraction of her MMS) is EFX-satisfied by her bundle.

\begin{definition}[\textbf{\EIME Allocation}] 
Given valuations $(v_1, v_2,\ldots, v_n)$ over a set \(M\) of indivisible goods, we say that allocation \(\mathcal{X} = (X_1, \ldots, X_n)\) is an \emph{\EIME allocation} if for every agent \(i \in [n]\),  
either
\[
v_i(X_i) \;\geq\; (1-\varepsilon) \cdot MMS(v_i,n),
\]
or agent~\(i\) is \emph{EFX-satisfied} with her bundle \(X_i\).

When $\varepsilon = 0$, we say that such an allocation is \IME.
\end{definition}

{
Our main {existential} result for three agents (\Cref{thm:three-agents-main}) follows from the following main lemma that captures the ex-post properties of each of the supporting six allocations. Each of these six allocations corresponds to a different assignment of the following three roles to the three agents:
\emph{divider}, \emph{subdivider}, and \emph{chooser}. 
The six allocations are grouped in three pairs, each with a fixed agent that {takes on the role of the divider}.} 
{In the following lemma, we \textbf{fix} one agent (agent $i$) as the divider, 
and describe the ex-post guarantees of the pair of allocations constructed under this fixed role.}

\begin{lemma}\label{lem:main-lemma}

Consider settings with three agents with additive valuations over a set $M$ of indivisible goods.
For each agent \( i \in N=[3] \), denoted as the \emph{divider}, {\cref{alg:three-agent-construction} constructs} two (not necessarily distinct) allocations 
\(\mathcal{X}^i = (X^i_1, X^i_2, X^i_3)\) and \(\mathcal{Y}^i = (Y^i_1, Y^i_2, Y^i_3)\) such that:
\begin{enumerate}
    \item \textbf{Divider’s value guarantee.}
    The divider (agent $i$),  
    receives a bundle of value at least \(MMS_i \) in both allocations; that is,
    \[
    v_i(X^i_i) \ge MMS_i 
    \quad \text{and} \quad 
    v_i(Y^i_i) \ge MMS_i.
    \]
    \item \textbf{Subdivider and chooser value guarantees.}
    For each agent \( l \neq i \) it holds that: 
    
    \[
    v_l(X^i_l) + v_l(Y^i_l) \ge   v_l(M) - MMS_l.
    \]
    and additionally:
    \begin{enumerate}
        \item \textbf{Subdivider guarantees}: the subdivider agent, denoted as agent $j$, is guaranteed $\frac{9}{10}MMS_j$, and is EFX-satisfied. 
        \item \textbf{Chooser guarantees}: the chooser agent, denoted as agent $k$, is guaranteed at least her proportional share. 
    \end{enumerate}
    \item \textbf{Envy-Based guarantees.}
    For each agent \( r \in N \), the bundle she receives in both allocations is EEFX-satisfactory for her. 
    Moreover, the subdivider and at least one of the other agents are EFX-satisfied. 
    
\end{enumerate}
\end{lemma}

Before proving the lemma we show that \Cref{thm:three-agents-main} follows from it.

\paragraph{Proof of \cref{thm:three-agents-main}}
\begin{proof}

By \Cref{lem:main-lemma}, there exist six allocations \[
\{ \mathcal{X}^1, \mathcal{Y}^1, \mathcal{X}^2, \mathcal{Y}^2, \mathcal{X}^3, \mathcal{Y}^3 \}
\]
such that for each \( i \in \{1,2,3\} \):

\begin{enumerate}
       \item Agent \( i \) receives a bundle of value at least \( MMS_i \) in the allocations $\mathcal{X}^i, \mathcal{Y}^i$; that is,
    \[
    v_i(X^i_i) \geq MMS_i \quad \text{and} \quad v_i(Y^i_i) \geq MMS_i.
    \] 
    \item For each $l \in N \setminus \{i\}$, 
     the combined value received by agent $i$ in the allocations $\mathcal{X}^l, \mathcal{Y}^l$ is at least the sum of the two highest value bundles (by $v_i$) in $\mathcal{MMS-EFX}^i$, i.e.,
    \[
     v_i(X^l_i) + v_i(Y^l_i) \geq v_i(M)-MMS_i \quad \text{for each } l \in N \setminus \{i\}.
    \]

\end{enumerate}

Let \( \mu \) be the uniform distribution over the six allocations
\[
\{ \mathcal{X}^1, \mathcal{Y}^1, \mathcal{X}^2, \mathcal{Y}^2, \mathcal{X}^3, \mathcal{Y}^3 \}.
\]

The above implies that the expected value of agent $i$ under \( \mu \) is at least her proportional share:

\begin{multline*}
\begin{split}
    \mathbb{E}_{(Z_1,Z_2,Z_3)\sim \mu} \left[ v_i(Z_i) \right]
& =\frac{1}{6}\left(v_i({X}_i^1)+v_i(Y_i^1)+v_i({X}_i^2)+v_i(Y_i^2)+v_i({X}_i^3)+v_i({Y}_i^3)\right)\;
\\ & \geq \frac{2}{6} \cdot MMS_i + \frac{2}{6} \cdot \left(v_i(M)-MMS_i\right)
\;=\;  \frac{v_i(M)}{3}\;=\;\mathrm{PROP}_i.
\end{split}
\end{multline*}

We conclude that \( \mu \) is ex-ante proportional.
Additionally, every allocation in the support of $\mu$ upholds the ex-post guarantees of \cref{lem:main-lemma}.

This completes the proof of \Cref{thm:three-agents-main}.
\end{proof}

 \subsection{Proof of \Cref{lem:main-lemma}}\label{subsec:support-construction}

We recall that our BoBW distribution in \cref{thm:three-agents-main} is supported on six allocations, where each of these six allocations corresponds to a different assignment of the following three roles to the three agents:
\emph{divider}, \emph{subdivider}, and \emph{chooser}. 
The six allocations are constructed in three pairs, one for each assignment of the role of the divider to one of the agents. \Cref{lem:main-lemma} focuses on one such pair. 
We next outline the algorithm 
to construct the two allocations that correspond to a fixed divider, 
and the proof that they satisfy the guaranteed properties.

Fix an agent $i\in \{1,2,3\}$, and, in order to simplify notation, we assume without loss of generality, that \(i=1\). We denote this agent as the \emph{divider}. 
Once fixed, we use \cref{alg:three-agent-construction} to construct two EEFX allocations that satisfy the desired properties as follows:

\begin{enumerate}
    \item The divider partitions her items into an MMS-EFX partition according to her valuation. We denote it as  $\{A,B,C\}$.
    \item We carefully choose two bundles (from the divider's partition) for the subdivider to re-partition into two bundles according to her MMS.
    \item The chooser chooses her favorite bundle out of the two repartitioned bundles and the remaining third bundle.
    \item Based on the choosers's choice, we then allocate a bundle from the original partition to the divider, and the remaining items to the subdivider. 
\end{enumerate}
The allocations satisfy the following properties: 
\begin{enumerate}
    \item Agent 1 (the \emph{divider}) receives one of the bundles $\{A,B,C\}$, and thus a bundle of value at least \(MMS_1\) -- and therefore is also guaranteed her MMS share;
    \item Each of the other agents \(l \in \{2,3\}\) receives, over the two allocations combined, bundles whose total value is at least \(v_l(M) - MMS_l\); 
    \item Across the two allocations, agents~2 and 3 exchange the roles of \emph{subdivider} and \emph{chooser}: in each allocation, the chooser receives at least her proportional share and is EEFX-satisfied, and the subdivider receives at least a $\tfrac{9}{10}$-MMS value and is EFX-satisfied.
\end{enumerate}

We present our algorithm, whose pseudocode is given in \cref{alg:three-agent-construction}.

\begin{algorithm}[H]
\SetAlgoNlRelativeSize{-1}   
\setstretch{0.95}           
\small                        
\caption{Constructing Allocations $\mathcal{X}^1$ and $\mathcal{Y}^1$ for Agent~1}
\label{alg:three-agent-construction}
\DontPrintSemicolon

\KwIn{Three additive valuations $v_1,v_2,v_3$ over a set of goods $M$.}
\KwOut{Two EEFX allocations \(\mathcal{X}^1 = (X^1_1, X^1_2, X^1_3)\) and \(\mathcal{Y}^1 = (Y^1_1, Y^1_2, Y^1_3)\), 
such that agent~1 receives her \(MMS\) share. 
In \(\mathcal{X}^1\), one agent \(k \in \{2,3\}\) receives her proportional share, while the remaining agent \(j = \{2,3\} \setminus \{k\}\) receives at least a \(\tfrac{9}{10}MMS\) value and is \emph{EFX}-satisfied. 
In \(\mathcal{Y}^1\), the roles of agents \(j\) and \(k\) are reversed. 
Moreover, for each \(l \in \{2,3\}\),
\[
v_l(X^1_l) + v_l(Y^1_l) \;\geq\; v_l(M) - MMS_l.
\]
Additionally, for each agent \(t \in \{1,2,3\}\), we output a certificate \(C_t^1\), 
which allows the agent to efficiently verify that she is \emph{EEFX}-satisfied in any allocation where she is not \emph{EFX}-satisfied.}

Let $\{M_1,M_2,M_3\}$ be an MMS partition of $M$ according to $v_1$ (for $n=3$)\;
Let $\{A,B,C\} = \mathcal{C}^1_1= \realloc(\{M_1,M_2,M_3\},v_1)$ \tcp*{Make allocation EFX-satisfying for $v_1$}
For $l \in \{2,3\}$, denote $\mathcal{T}_l = \operatorname{argmax}_{S \in \{A,B,C\}}v_l(S)$\;

\tcp{Case distinction based on agents 2 and 3’s top bundles:}
\uIf(\tcp*[f]{Case 1: Both can get top-valued bundle}){$\mathcal{T}_2 \neq \mathcal{T}_3$ or $|\mathcal{T}_2| > 1$}{
    Arbitrarily assign each of agents 2 and 3 one of her top-valued bundles. $X^1_j\in \mathcal{T}_j$ for $j\in \{2,3\}$\;
    Set $X^1_1$ to be $\{A,B,C\} \setminus \{X^1_2, X_3^1\}$\;
    Set $\mathcal{Y}^1 =\mathcal{C}^1_2=\mathcal{C}^1_3 = \mathcal{X}^1$\;
    \KwRet $\mathcal{X}^1, \mathcal{Y}^1, (C^1_1,C^1_2,C^1_3)$ \tcp*{Each agent receives MMS and is EFX-satisfied}
}

\tcp{Case 2: Both value the same top bundle $A$}
\For{each $l \in \{2,3\}$}{
    $\mathcal{P}^l_B = \textbf{MMS-EFX-Improved-RePartition}(A,B,v_l)$\;
    $\mathcal{P}^l_C =  \textbf{MMS-EFX-Improved-RePartition}(A,C,v_l)$\;
    Choose $Z_l \in \operatorname{argmax}_{Z \in \{B,C\}} \min_{S\in  \mathcal{P}^l_Z} v_l(S)$\;
    Let $\mathcal{P}^l =\{P^l_1,P^l_2\} = \mathcal{P}^l_{Z_l}$\;
    Let $L_l = M \setminus (A \cup Z_l)$ \tcp*{Remaining bundle from partition $\{A,B,C\}$}
}

\tcp{Case 2.A}
\uIf{there exists $r \in \{2,3\}$ such that for $l \in \{2,3\}\setminus \{r\}$ it holds $\max\{v_l(P_1^r),v_l(P_2^r)\} < v_l(L_r)$}{
    Let $\mathcal{C}^1_l=\{P_1^r, P_2^r, L_r\}$\;
    Construct $\mathcal{X}^1$: agent $r \gets A$, agent $l \gets L_r$, agent $1 \gets Z_r$\;
    \eIf{$\max\{v_r(P_1^l),v_r(P_2^l)\} < v_r(L_l)$}{
        Let $\mathcal{C}^1_r=\{P_1^l, P_2^l, L_l\}$\;
        Construct $\mathcal{Y}^1$: agent $l \gets A$, agent $r \gets L_l$, agent $1 \gets Z_l$\;
    }{
        $\mathcal{Y}^1 =\mathcal{C}^1_r= \mathcal{X}^1$\;
    }
    \KwRet $\mathcal{X}^1, \mathcal{Y}^1, (C^1_1,C^1_2,C^1_3)$ \tcp*{MMS + EFX/EEFX-satisfied}
}
\Else(\tcp*[f]{Case 2.B}){
    \For{each $j \in \{2,3\}$ \tcp*[f]{agent $j$ is the \emph{subdivider}}}{
        Denote $k = [3] \setminus \{1,j\}$ as the \emph{chooser}\;
        Create an allocation $\mathcal{R}^j$ as follows:\;
        Agent $k$ chooses her top valued bundle in $\mathcal{P}^j$\;
        Agent $j$ takes the remaining bundle in $\mathcal{P}^j$\;
        Agent 1 (the \emph{divider}) receives $L_j$\;
    }
    Assign $\mathcal{X}^1=\mathcal{C}_3^1=\mathcal{R}^2, \mathcal{Y}^1=\mathcal{C}_2^1=\mathcal{R}^3$\;
    \KwRet $\mathcal{X}^1, \mathcal{Y}^1, (C^1_1,C^1_2,C^1_3)$\;
}
\end{algorithm}

\begin{algorithm}[H]
\caption{\textsc{MMS-EFX-Improved-RePartition}$(A, B, v)$}
\label{alg:partition}
\DontPrintSemicolon

\KwIn{Two bundles $A, B$ and an additive valuation $v$}
\KwOut{A partition $\{A',B'\}$ of the items in $A\cup B$ to two bundles, such that $\min_{S \in \{A',B'\}}v(S) \geq \min_{S \in \{A,B\}}v(S)$, and such that both $A'$ and $B'$ are EFX satisfying for $v$.}

Let $\{X',Y'\}$ be an MMS partition of $A \cup B$ according to $v$ (for $n=2$)\;
$\{A',B'\}  = \realloc(\{X',Y'\}, v)$ \tcp*{\cref{alg:realloc}}

\If{$\min_{Z \in \{A',B'\}}v(Z) < \min_{Z \in \{A,B\}}v(Z)$}{
    $\{A',B'\}  = \realloc(\{A,B\}, v)$\;
}

\KwRet $\{A',B'\}$
\end{algorithm}

{We now describe the construction step by step, explaining the purpose of each stage and proving that all stages collectively uphold the guarantees stated {in \cref{lem:main-lemma}}. 
\paragraph{The Main 3-Agent Algorithm (\cref{alg:three-agent-construction}).}
{In Lines~1–2, we construct an \(\mathcal{MMS\text{-}EFX}^1\) partition for agent~1 as follows: we first compute an MMS partition of $M$ to \(n=3\) bundles, and then apply the \realloc\ procedure (\cref{alg:realloc}), which guarantees the resulting partition is $\mathcal{MMS\text{-}EFX}$ for agent 1 (see \Cref{sec:efx-reallocation}).  
We denote the resulting partition by \(\{A, B, C\}\).} \\
For each agent
\(l\in\{2,3\}\), we use $\mathcal{T}_l$ to denote the set of all bundles that agent $l$ views as top-valued bundles:

\[
\mathcal{T}_l := \arg\max_{S\in\{A,B,C\}} v_l(S)
\]

We now distinguish two exhaustive and mutually exclusive sub-cases, based on
how agents~2 and~3 evaluate {the bundles in $\{A,B,C\}$}. In
each case, we construct two allocations that satisfy the requirements stated
above. 

\begin{enumerate}
    \item \textbf{Both can concurrently get a top-valued bundle (Case 1).}
    We have $\mathcal{T}_2\neq \mathcal{T}_3$ or $|\mathcal{T}_2| > 1$, and we can  
    concurrently allocate to each of agents 2 and 3 one of her top-valued bundles.
    \item \textbf{Case 1 does not hold}. Each agent has a unique  top-valued bundle, and it is the same bundle for both agents. Without loss of generality we assume that bundle $A$ is that bundle.

\end{enumerate}

We now proceed to construct two allocations, \( \mathcal{X}^1 = (X^1_1, X^1_2, X^1_3) \) and \( \mathcal{Y}^1 = (Y^1_1, Y^1_2, Y^1_3) \), for each of the cases described above.
These allocations are constructed such that in each of the two allocations agent 1 receives one of the bundles in $\{A,B,C\}$, the computed $\mathcal{MMS-EFX}^1$ partition. Therefore, in every allocation, she is EEFX-satisfied (by definition of an $\mathcal{MMS-EFX}^1$ partition), and receives a value of at least $MMS_1$ (so she is guaranteed her MMS). The two other bundles might be repartitioned, but even if so, we give agent~1 the partition $\mathcal{C}^1_1=\{A,B,C\}$ as an EEFX-Certificate, and it includes the bundle she receives in any of the two final allocations. Therefore, she can use $\mathcal{C}^1_1$ to efficiently verify she is indeed EEFX-satisfied.

    \paragraph{\textbf{Case 1: Both can concurrently get a top-valued bundle.  (Lines 4 - 8)}}
    \begin{itemize}

        \item By assumption, we can give each of agents 2 and 3 a different bundle, which they each considers to be a top-valued bundle. We define $\mathcal{X}^1$ as follows -- we arbitrarily assign the two different top bundles to their corresponding agents, and define $X^1_1$ to be the remaining bundle in the original partition. We define the allocation $\mathcal{Y}^1$ to be equal to $\mathcal{X}^1$. 
        
    \item \textbf{Value guarantees for agents 2 and 3.} 
    Agents~2 and~3 each receive their top-valued bundle in both allocations. 
    By \cref{lem:guarantee-value}, each agent $l\in \{2,3\}$ receives, over the two allocations combined, bundles whose total value is at least $v_l(M)-MMS_l$.
     Moreover, since each of the two agents receives a top-valued bundle, she gets at least her proportional share, so she receives at least her MMS share.
     Moreover, we notice that as the divider's bundle is a bundle from her  
         original MMS-EFX partition, the divider (agent 1) is  also guaranteed her MMS.
     
     Thus, all three agents get their MMS.
        \item \textbf{EFX Guarantee:} 
         Since each of agents $2$ and $3$ receives one of her top-valued bundle in the allocation -- each is EFX-satisfied (\cref{obs:top-is-EFX}). 
         Moreover, we notice that as the divider's bundle is a bundle from her  
         original MMS-EFX partition, the divider (agent 1) is EFX-satisfied. 
         
         Thus,  all three agents  are EFX-satisfied. {We note that in this case, no EEFX-certificates are needed (as the allocations themselves serve as certificates).}  
    \end{itemize}
We conclude that in this case every agent gets her MMS and is EFX-satisfied.

The only case in which it is not possible to concurrently give both agent $2$ and agent $3$ their top-valued bundles is when $\mathcal{T}_2=\mathcal{T}_3=\{A\}$. Below we always assume that this holds. 

Before moving to the second case we first present a subroutine that will be used within it. 
{\paragraph{The MMS-EFX-Improved-RePartition Algorithm (\cref{alg:partition}).}  
This procedure receives a partition of two bundles, \(A\) and \(B\), and a valuation $v$  
and returns a new partition $\{A', B'\}$ such that \(A'\) and \(B'\) are an  \(MMS\) partition of \(A \cup B\) according to $v$.  
Additionally, 
it guarantees that
\[
\min_{Z \in \{A, B\}} v(Z) \;\leq\; \min_{Z \in \{A', B'\}} v(Z),
\]
and that \(A'\) and \(B'\) are \emph{EFX}-dominating with respect to one another.  
We further note that even when replacing exact \(MMS\) with approximate \((1-\varepsilon)\)-MMS partitions,  
these guarantees continue to hold due to the condition enforced in Line~3 of the algorithm, and application of the \realloc algorithm (\cref{alg:realloc}).
}

{We now proceed to the second case.}

\paragraph{\textbf{Case 2: $A$ is the unique top-valued bundle for both agents. (Lines 9 - 32).}}

   For each agent $l \in \{2,3\}$,  and $Z
   \in \{B,C\}$, we denote $$\mathcal{P}^l_Z=\textbf{MMS-EFX-Improved-RePartition}(A,Z,v_l).$$

   We then choose 
   \[
Z_l \in 
\arg\max_{Z \in \{B,C\}} \min_{S\in  \mathcal{P}^l_Z} v_l(S)
\] 
and denote $\mathcal{P}^l =\{P^l_1,P^l_2\}= \mathcal{P}^l_{Z_l}$ (i.e - we choose $Z \in \{B,C\}$ such that it maximizes the minimal-valued bundle in the partition returned by the call to the function \textbf{MMS-EFX-Improved-RePartition}$(A,Z,v_l)$).

Let $L_l = M \setminus (A \cup Z_l)$ be the remaining items.
    
    \paragraph{\textbf{Subcase 2.A (Lines 15 - 23).}}
    If there exists an agent \( r \in \{2,3\} \) such that for the other agent \( l \in \{2,3\}\setminus \{r\} \),

    \[
    \max \{ v_l(P_1^r),\, v_l(P_2^r) \} < v_l(L_r),
    \]    
    
    that is, agent $l$ strictly prefers the third bundle in the partition to any of the bundles in $\mathcal{P}^r$ (which, in turn, implies that bundle \( L_r=M \setminus (A \cup Z_r) \), the remaining bundle in the original three-way partition, must be valued by $l$ \textbf{above} her proportional share),  
    we define both \( \mathcal{X}^1 \) and \( \mathcal{Y}^1 \) as follows:  
    agent \( r \) receives bundle \( A \);  
    agent \( l \) receives \( L_r= M \setminus (A \cup Z_r) \);  
    and agent \( 1 \) receives \( Z_r \). We denote $\mathcal{C}^1_l=\{P_1^r,P_2^r,L_r\}$, and notice that it is an EEFX-certificate for agent $l$.

    We note that if both agents satisfy this condition, we instead construct $\mathcal{Y}^1$ by swapping $r$ and $l$, thereby preserving all guarantees and ensuring that each agent has an equal chance to receive her top-valued bundle in the allocation.
    
    \begin{enumerate}
    \item \textbf{Value guarantees for agents 2 and 3.} 
    Agents~2 and~3 each receive their top-valued bundle in some partition, twice. 
    By \cref{lem:guarantee-value}, each agent $t \in \{2,3\}$ receives, over the two allocations combined, bundles whose total value is at least $v_t(M)-MMS_t$. 
     Moreover, since each agent values her bundle above her proportional share, she receives at least her MMS share.
        \item \textbf{EEFX/EFX Guarantees:} 
         Agent $r\in \{2,3\}$ receives her top-valued bundle in the allocation -- and therefore she is  EFX-satisfied (\cref{obs:top-is-EFX}). Agent $l\in \{2,3\}$ receives a bundle valued at least her proportional share, and therefore, by $\cref{lem:bundle-at-least-prop-is-EEFX}$, she is EEFX-satisfied. Finally, we notice that as the divider's bundle is a bundle from her  
         original MMS-EFX partition, the divider (agent 1) is EFX-satisfied.
    \end{enumerate}
We note that in this case, every agent receives her MMS share. 

From this point onward, we assume that no such agent \( r \) exists. 

    \paragraph{\textbf{Subcase 2.B (Lines 24 - 32).}}
    Among the two non-divider agents (agents 2 and 3), 
    one acts as the \emph{subdivider} and the other as the \emph{chooser}. 
    The chooser receives at least her proportional share (and in fact, a bundle that is a top-valued bundle in the partition), the subdivider is guaranteed $\tfrac{9}{10}$-MMS and is EFX-satisfied, and the divider (agent 1) receives a bundle from her original $\mathcal{MMS-EFX}^1$  partition (thus guaranteeing her MMS and EEFX-satisfaction). 
    A second allocation is produced by swapping the roles of the subdivider and chooser.

\begin{itemize}

    \item \textbf{Allocation:} 
    Define $\mathcal{Y}^1$ as follows: Agent 2 acts as the chooser, and agent 3 as the subdivider.  
    The chooser (agent 2) takes her top bundle (above her proportional share) in $\mathcal{P}^3$, the subdivider (agent 3) the other bundle in $\mathcal{P}^3$, and the divider gets the remaining bundle $L_3$. 
    The allocation $\mathcal{X}^1$ is similarly obtained after swapping the assignment of roles of subdivider and chooser between agents 2 and 3.

    \item \textbf{Value Guarantees for agents 2 and 3:} 
    We notice the \emph{chooser} receives above her proportional share (she receives her top bundle in the partition, since we are not in case 2.A).
    The more challenging claim, proving that  
    the \emph{subdivider} receives at least $\frac{9}{10}$ of her MMS, is deferred to \cref{lem:val-910-mms} below. 
    We also defer to \Cref{lem:total-val-case-2} the proof  that each agent $t \in \{2,3\}$ receives, over the two allocations combined, bundles whose total value is at least $v_t(M) - MMS_t$.

    \item \textbf{EFX Guarantee for agents 2 and 3:} The chooser is EFX-satisfied by \Cref{obs:top-is-EFX}. 
    The more challenging claim, proving that  
    the \emph{subdivider} is EFX-satisfied, is deferred to \Cref{lem:reallocation-is-still-EFX} below.

\end{itemize}

To complete the proof, we only need to prove the following four lemmas: 

\begin{restatable}{lemma}{lemtopval}
\label{lem:guarantee-value}
Let \( \{S_1, S_2, S_3\} \)  be a partition of the item set \( M \), and let there be an agent with an additive valuation $v$. Assume without loss of generality that
\[
v(S_1) \geq \max [v(S_2),v(S_3)].
\]

Then it holds that
    $$2 \cdot v(S_1) \geq v(S_1) + \max [v(S_2),v(S_3)] \geq  v(M)-MMS(v,3)$$

\end{restatable}

\begin{restatable}{lemma}{totalvaluelemma}\label{lem:total-val-case-2}
    In Subcase 2.B, The total value that each agent $t \in \{2,3\}$ receives over the two allocations $\mathcal{X}^1, \mathcal{Y}^1$ is
    
    at least $v_t(M) - MMS_t$. That is, for $t \in \{2,3\}$ it holds that $$v_t(X^1_t)+v_t(Y^1_t)\geq v_t(M) - MMS_t$$ 
\end{restatable}

{
\begin{restatable}{lemma}{lemreallocefxl}
\label{lem:reallocation-is-still-EFX} 
    Consider agent $l$'s  $\mathcal{P}^l$ partition (Line 13), denoted by $\{P_1^l, P_2^l\}$. 
    Then each of $P_1^l$ and of $P_2^l$ is {EFX-satisfactory for agent $l$ with regards to the partition $\{P_1^l, P_2^l, L_l\}$, where $L_l=M\setminus (P_1^l \cup P_2^l)$}.
    
\end{restatable}

\begin{restatable}{lemma}{lemninemms}
\label{lem:val-910-mms}
    Consider agent $l$'s  $\mathcal{P}^l$ partition (Line 13), denoted by $\{P_1^l, P_2^l\}$.
    Then $\min \{ v_l(P_1^l),v_l(P_2^l)\} \geq \frac{9}{10} MMS_l$ (where $MMS_l = MMS_l(v_l,3)$ is the MMS value over 
    the set of goods $M$ with respect to valuation $v_l$ and $n=3$ agents).

\end{restatable}
  }
The proofs of these lemmas are deferred to~\Cref{subsec:technial-lemma}.

We also remark that the analysis of the  MMS approximation of our algorithm is tight, and prove this in \cref{tightmms}.

\subsection{{Verification of \cref{thm:three-agents-main}}}\label{subsec:cert}

{Our BoBW construction is designed to address key practical aspects of fairness and verifiability.  
First, each agent can independently verify that she was treated fairly according to the guarantees, without needing to trust the algorithm or know others’ valuations (which may be private).  
To ensure this, our distribution has a support of only six allocations. This small support allows every agent to explicitly and efficiently compute her expected value, {and} verify that it meets her proportional share. The agents can even jointly randomize over the allocations themselves (e.g., by rolling a fair die), rather than relying on the algorithm’s internal randomization.  
In contrast, prior BoBW algorithms such as \cite{azizfreemannew} require a support of size \(m \cdot n + 1\), which grows with the number of items \(m\); our construction, like that of \cite{babaioff2021bobw}, has constant-size support for \(n=3\) agents. 

Second, verifying ex-post envy-based guarantees is straightforward for EFX or EF1 allocations, but is a non-trivial problem for EEFX, as checking EEFX-satisfaction is not known to be polynomial-time computable.  
To address this, our algorithm outputs, for each agent, a \emph{certificate} of EEFX satisfaction: a partition that certifies her bundle as EFX-satisfying, allowing verification in polynomial time.  
Without such a certificate, an agent would have no efficient way to confirm or dispute her fairness guarantee.  

Finally, each agent is assured of strong value guarantees: she receives at least a \(\tfrac{9}{10}MMS\) value (and else she must be  EFX-satisfied).  
Since each agent can approximate her own \(MMS\) value using any FPTAS or PTAS, she can efficiently verify that she receives at least \((1-\varepsilon)\) of her calculated approximation of the MMS, or $\frac{9}{10}-\varepsilon$ of that value while meeting the corresponding EFX guarantee.}

We prove the following:

\begin{restatable}{theorem}{polycert}[Polynomial-Time Verifiability of Fairness Guarantees]

\label{lem:verification}
{For any fixed $\varepsilon > 0$, there exists an FPTAS that enables polynomial-time verification of following fairness guarantees for each agent $i \in \{1,2,3\}$}

\begin{enumerate}
    \item {Given the six supporting allocations, she can verify that a uniform randomization over the six allocations guarantees her proportional share ex-ante}.
    
    \item She is EEFX-satisfied with any allocation produced by the algorithm.
    \item Ex-post, she either receives at least a $(1-\varepsilon)$ fraction of the calculated approximate to her MMS value, or is EFX-satisfied while receiving at least a $(\tfrac{9}{10}-\varepsilon)$ fraction of the calculated approximate to her MMS.
    
\end{enumerate}
\end{restatable}

\begin{proof}
We verify each guarantee separately. 
\begin{enumerate}
    \item \textbf{Ex-ante proportionality.}  The proportional share is simply the value for $M$ divided by 3, so each agent can easily compute her proportional share in poly-time.
    Since the algorithm outputs only six allocations, 
    each agent can explicitly compute her expected value and verify in polynomial time that it meets her proportional share. 

    \item \textbf{EEFX satisfaction.} 
    {For $i \in \{1,2,3\}$, and every two allocations $\mathcal{X}^i,\mathcal{Y}^i$, a single certificate is presented to each agent. If the agent is not \emph{EFX}-satisfied in one of the two allocations, the certificate she receives serves as an \emph{EEFX-certificate} for the bundle assigned to her in that allocation. {We note that each certificate corresponds to a \emph{partition} that includes the bundle an agent receives in that allocation, together with a repartition of the remaining items into two bundles such that agent~\(i\)’s bundle \emph{EFX}-dominates both.}
    } 
    
    {Computing that the bundle EFX-dominates each of the two other bundles can clearly be done in polynomial time.}

    \item \textbf{Ex-post MMS guarantee.}  
    As each agent can compute a $(1-\varepsilon)$-approximation to her MMS value in polynomial time~\cite{Mittalsantaclause}, she can verify whether she receives at least $(1-\varepsilon)$-MMS. 
    Otherwise, the agent can verify in polynomial time that her allocation satisfies EFX and yields at least $(\tfrac{9}{10}-\varepsilon)$-MMS.
\end{enumerate}
\end{proof}
 
\subsection{FPTAS with Approximate Value Guarantees}\label{subsec:FPTAS}

{ As previously discussed, computing an exact MMS partition is known to be NP-hard \cite{Kurokawa2018}. 
Consequently, without access to the MMS partitions of the agents, our algorithm cannot be implemented in polynomial time. 
However, \citet{Mittalsantaclause} presented a fully polynomial-time approximation scheme (FPTAS) for the \emph{Santa Claus Problem} with a fixed number of agents. 
In the \emph{Santa Claus Problem}, there are $n$ agents and $m$ indivisible resources to be distributed among them. 
Each resource is indivisible, i.e., it can be allocated to exactly one agent, and each agent has a utility associated with every resource. 
The objective is to allocate the resources so as to maximize the minimum utility among all agents. 
We note that this problem is essentially equivalent to computing an MMS partition.

This FPTAS implies that, for any $\varepsilon > 0$, for the setting of three (or two) agents,
one can efficiently compute a $(1-\varepsilon)$-approximate MMS partition for each agent $i$. 
In \cref{sec:efx-reallocation}, we prove that such a partition can then be transformed into an $\mathcal{MMS}^{(\varepsilon)}-\mathcal{EFX}^i$ partition
allocation for agent~$i$ in polynomial time (\cref{lem:new_realloc}). 
We use this construction to prove the following:

\begin{restatable}{theorem}{threeagentsapprox}\label{thm:bobw-poly-approx}
Fix any $\varepsilon>0$.
For $n = 3$ agents with additive valuations, 
there exists an FPTAS that computes a distribution $\mu$ such that:

\begin{enumerate}
    \item $\mu$ is $(1-\varepsilon)$ proportional ex-ante. 
    \item In every deterministic allocation in the support of \( \mu \): 
    \begin{enumerate}
        \item One agent is EFX-satisfied and receives her proportional share.
        \item One agent is EFX-satisfied and receives at least  a $\frac{9}{10}-\varepsilon$ fraction of her MMS value.
        \item One agent is EEFX-satisfied and receives at least $(1-\varepsilon)$ of her MMS value.
    \end{enumerate}

\end{enumerate}

\end{restatable}

We note this implies that every deterministic allocation in the support of $\mu$ is \EIMMX, and in each allocation each agent receives at least $\frac{9}{10}-\varepsilon$ of her MMS share, and is also EEFX-satisfied.

\begin{proof}

Our modified algorithm proceeds as follows:

{In \cref{alg:three-agent-construction}, we replace Line~1 (the MMS computation) with an FPTAS computation of an $(1-\varepsilon)-$MMS partition. 
Similarly, in \cref{alg:partition}, we replace Line~1 (the MMS computation) with an FPTAS computation of an $(1-\varepsilon)$-MMS partition.
We note that in \cref{alg:partition}, this modification preserves the following properties:
\begin{enumerate}
    \item Each of the output bundles is $EFX$ with respect to the other; and
    \item The minimal value of the partition does not deteriorate.
\end{enumerate}
Therefore, \cref{lem:reallocation-is-still-EFX} remains valid, implying that all envy-based guarantees are likewise maintained.
}

For share-based guarantees, each agent in every allocation retains the same guarantees up to a factor of $(1-c\cdot \varepsilon)$ for some constant $c > 0$ (see \Cref{remark1}). 
Therefore, all ex-ante and ex-post guarantees hold up to this factor.

\end{proof}

\begin{remark}\label{remark1}
Each of the finitely many inequalities in the proofs of \Cref{lem:guarantee-value}, \Cref{lem:total-val-case-2}, \Cref{lem:reallocation-is-still-EFX}, \Cref{lem:val-910-mms} that relies on the optimality of an exact MMS partition degrades by at most a multiplicative factor of $(1 - \varepsilon)$ when the partition is replaced by a $(1 - \varepsilon)$-approximate one, while all remaining inequalities -- including all EFX/EEFX guarantees -- hold verbatim. As each derivation contains a constant number $c$ of such inequalities, all value guarantees hold up to a factor of $(1 - c\cdot \varepsilon)$; running the algorithm with $\varepsilon' =  \varepsilon / c$ yields \cref{thm:bobw-poly-approx} as stated, and the running time remains $poly(|input|, 1/\varepsilon)$.
\end{remark}

\section{Optimal Poly-time BoBW Result for Two Additive Agents}\label{sec:two-agents}

In this section, we present our optimal poly-time BoBW fairness result for two  agents with additive valuations over indivisible goods. 
Previous polynomial-time BoBW algorithms, such as those of \citet{garg2024bestofbothworldsfairnessenvycycleeliminationalgorithm} and \citet{bu2024bestofbothworldsfairallocationindivisible}, both achieve ex-ante EF and ex-post EFX, but each only ensures a constant fraction of the maximin share (MMS) (and we observe that each may lose a constant fraction as well (\cref{{sec:no-min-bound}})). 
In contrast, {by applying a small modification to the algorithm of \cite{bu2024bestofbothworldsfairallocationindivisible}}, we guarantee the same envy-based fairness benchmarks while improving the share-based guarantee to $(1-\varepsilon)$-MMS, in polynomial time. 
We discuss the techniques and limitations of these prior approaches in greater detail in \cref{sec:no-min-bound}.

We first present our main proposition for two-agent settings:
\begin{restatable}{proposition}{twoagentsmainprop}
    \label{the:main-2agents} Consider settings with two agents who have additive valuations over a set 
of indivisible goods. 

There exists a polynomial-time algorithm that given two {partitions} $\mathcal{A}^i=\{A^i_1,A^i_2\}$ for $i\in \{1,2\}$, outputs a distribution over at most two allocations satisfying the following properties:

\begin{itemize}
    \item \textbf{Ex-ante envy-free:} The randomized allocation is ex-ante envy free (and thus ex-ante proportional).

    \item \textbf{Ex-post EFX:} Every allocation in the support of the distribution is EFX. 
    
    \item \textbf{Ex-post value guarantee:} Every allocation in the support of the distribution guarantees each agent \( i \) a value of at least $$\min\{v_i(A^i_1),v_i(A^i_2)\}$$
\end{itemize}
\end{restatable}

{We present} the algorithm and the proof of \cref{the:main-2agents} in \cref{sec:2agents}. We next show that
we can use this proposition, in combination with any FPTAS algorithm for computing the MMS  (e.g., the FPTAS of \cite{woegingerFPTAS}, or of \cite{Ibarra1975}, which can be adapted to compute a \((1-\varepsilon)\)-MMS partition into two bundles), to prove the following theorem:

\twoagentsmain*

\begin{proof}
    Consider a setting with two agents, 
    where each agent $i\in \{1,2\}$ has an additive valuation $v_i:2^M\rightarrow \mathbb{R}_{\geq0}$  over the set $M$ of $m$ items. 
    
   We first run any FPTAS algorithm that computes the MMS separately for each agent's valuation on $M$. 
    For $i\in \{1,2\}$, let $\{A^i_1,A^i_2\}$ denote the partition of the FPTAS  on valuation $v_i$. 
 
    We then apply \cref{the:main-2agents} on these two partitions. 
    We observe that each of the resulting allocations satisfy the needed value guarantee,  implying \Cref{cor:main}. 
\end{proof}

Finally, we note that \cite{babaioff2022fairshares} introduced a \textbf{PTAS} for computing the MMS partition in the two-agent case (described in \cref{sub:PTAS}), guaranteeing each agent a value of at least \((1-\varepsilon)\)-\(\mathrm{MMS}\). 
In fact, the analysis of \cite{babaioff2022fairshares} implicitly yields a slightly stronger result. 
By refining their proof, we show that the algorithm satisfies the following bound:
\[
\min\{v_i(A^i_1),v_i(A^i_2)\}
\;\geq\;
\min\left\{\mathrm{MMS}_i,\,(1-\varepsilon)\cdot\mathrm{PROP}_i\right\}
\;\geq\;
(1-\varepsilon)\cdot \mathrm{MMS}_i.
\]
Consequently, we can combine this result with \cref{the:main-2agents} to obtain slightly stronger ex-post share-based guarantee of $\min\left\{\mathrm{MMS}_i,\,(1-\varepsilon)\cdot\mathrm{PROP}_i\right\}$. 
A detailed proof of this refined bound appears in the appendix (\Cref{lem:min-ptas}).

\section{Near-Optimal Poly-time BoBW for Three  Additive Agents}\label{sec:three-agents-poly}

{We now return to the case of three agents with additive valuations and show that our two-agent procedure (\cref{alg:two-agents-efx}) can be leveraged to guarantee \textbf{exact} ex-ante proportionality, all ex-post share-based guarantees up to a factor of \((1-\varepsilon)\), and \EIMMX.}

\cref{thm:bobw-poly-approx} (of \cref{sec:three-agents-existence}) captures the variant of our algorithm which runs in polynomial time, but its ex-ante guarantee is not \emph{exactly} proportional (see example in \cref{subsec:ptas-fail}).
Getting only approximate proportionality rather than an exact one is particularly problematic if the agents realize in retrospect that they have identical valuations -- in that case achieving ex-ante proportionality is the same as ex-ante envy-freeness, so {it} is particularly desirable. 
Additionally, when valuations are identical, getting ex-ante envy-freeness with $(1-\varepsilon)$-MMS approximation is simple\footnote{To get ex-ante envy-freeness with $(1-\varepsilon)$-MMS approximation for identical valuations, simply randomly assign bundles to agents in a partition that is a $(1-\varepsilon)$-MMS approximation (which can be computed efficiently).}, so in such a setting the agents would be particularly disappointed if the allocation is not ex-ante envy-free (see \cref{subsec:approx-ptas-fail}  for an example that if such identical agents use different algorithms to compute partitions that {are} $(1-\varepsilon)$-approximate to the MMS, the algorithm we present in \cref{thm:bobw-poly-approx} might result in such ex-ante envy).    
Moreover, from {an} ex-ante perspective, proportionality is the best value guarantee agents can hope for, and thus we would like to improve the result and get exact proportionality. 
Crucially, this will allow us to make our poly-time algorithm comparable to the result of \cite{babaioff2021bobw} which obtains exact proportionality but only $1/2$-MMS. Below we show that it is possible to get a BoBW result with ex-ante proportionality and with a substantially improved MMS approximation of $9/10-\varepsilon$. Moreover, the only agent that is not getting $(1-\varepsilon)$ of her MMS must be EFX-satisfied.   

{There are two main reasons why \cref{alg:three-agent-construction} fails to guarantee exact ex-ante proportionality when executed as in \cref{thm:bobw-poly-approx}.  
The first arises in the simple case where two agents perform a \emph{cut-and-choose} procedure over two bundles.  
As shown in \cref{subsec:ptas-fail}, this procedure may not only fail to guarantee that each agent receives the average value of the two bundles, but may also violate \emph{EFX} satisfaction.  

The second limitation stems from the fact that, since we cannot compute each agent’s \(MMS\) partition 
exactly, some \textbf{other} partition encountered during the algorithm may offer the agent a better worst-case guarantee than her initial partition.  
We note that, to guarantee proportionality in \cref{thm:three-agents-main}, we rely on the guarantees of \cref{lem:main-lemma}, which ensure that in allocations where an agent is not the divider, she receives the sum of the top two bundles whose total value is at least as high as those in her MMS partition.  
However, if the initial partition is no longer the optimal one -- i.e., if some other partition offers a better worst-case guarantee -- this property may fail to hold.

Consequently, the agent may lose a fraction of her proportional share.  

We address these two issues as follows:
{
(1) First, we enable an agent to "replace" her initial MMS partition if she finds out that another agent's initial partition gives her a strictly \emph{better worst-case value guarantee},
(2) second, we invoke our two-agent procedure, \cref{alg:two-agents-efx}, to resolve the cut-and-choose issue, and  
(3) third, we introduce \emph{adoption stages}, in which agents may adopt a \emph{better worst-case-value guarantee} partition encountered throughout the execution.
}

We thus present a non-trivial modification of our poly-time BoBW algorithm for three agents (the FPTAS presented in \cref{thm:bobw-poly-approx}), which, although no longer guaranteeing EEFX, does result in a polynomial-time Best-of-Both-Worlds (BoBW) algorithm that guarantees \textbf{exact} ex-ante proportionality.

\begin{restatable}{theorem}{threeagentspoly}[Polynomial-time BoBW for three agents]\label{thm:bobw-poly}
  Fix any $\varepsilon>0$.
For $n = 3$ agents with additive valuations, 
there exists a fully polynomial-time approximation scheme (FPTAS) that computes a distribution $\mu$ such that:

\begin{enumerate}
    \item $\mu$ is proportional ex-ante. 
    \item {In every allocation in the support:}
    \begin{enumerate}
        \item One agent is EFX-satisfied and receives her proportional share.
        \item One agent receives at least a $\frac{9}{10}-\varepsilon$ fraction of her MMS value and is EFX-satisfied, or receives at least $(1-\varepsilon)$ of her MMS.
        \item One agent receives at least $(1-\varepsilon)$ of her MMS.
    \end{enumerate}

\end{enumerate}

\end{restatable}
We note this implies that every deterministic allocation in the support of $\mu$ is \EIMMX, and in each allocation each agent receives at least $\frac{9}{10}-\varepsilon$ of her MMS share.

We now prove \Cref{thm:bobw-poly}.

The algorithm works as follows:

\paragraph{Stage 1: Picking $ (1-\varepsilon) \cdot MMS$ Partitions.} 

{
First, for each agent $t \in \{1,2,3\}$, we generate a $(1-\varepsilon) \cdot MMS$ partition. 
{As previously mentioned, such a partition 
can be computed using any FPTAS or PTAS (e.g, the FPTAS presented by \citet{Mittalsantaclause} for the \emph{Santa Claus problem} with a fixed number of agents).}

Each agent~$t$ then evaluates all three partitions and selects the one whose minimal-valued bundle (according to her own valuation) is maximal (the {best} {worst-case-value guarantee} partition). We denote this partition as $\mathcal{M}^t$. 
After this selection step, it follows that no agent prefers another agent's partition in terms of the minimal-valued bundle. 
}
\paragraph{Stage 2 : Run \cref{alg:three-agent-construction}.}{We then execute the three-agent algorithm \cref{alg:three-agent-construction} as described in \cref{thm:bobw-poly-approx}, 
modifying Line~2 (the initialization of $\{A,B,C\}$) so that $\{A,B,C\}$ equals $\mathcal{M}^t$ for each agent $t$. We do not run \cref{alg:realloc} in this case. 
In other words, instead of initializing Line~2 independently for each agent with the output of \cref{alg:realloc} on an approximate MMS allocation generated by the FPTAS, we instead set $\{A,B,C\}$ in Line 2 to equal the agent's chosen partition $\mathcal{M}^t$. 

The algorithm outputs six allocations, consisting of $\mathcal{X}^i$ and $\mathcal{Y}^i$, and 9 certificates $\mathcal{C}^j_i$  for every $i,j\in [3]$.

\paragraph{Stage 3 -- Using \cref{alg:two-agents-efx} to fix cut-and-choose errors.}

As previously noted, when partitioning bundles between two agents using the approximate cut-and-choose method, we may fail to guarantee \emph{exact} \emph{ex-ante proportionality}.  
To address this, we \emph{fix} any such problematic partitions using our two-agent subroutine (\cref{alg:two-agents-efx}).  
We note that, by construction, the cut-and-choose procedure is only applied in \emph{Subcase~2.B}.  
Accordingly, we apply the two-agent procedure as follows: 

For each agent $r \in \{1,2,3\}$, if there exists an agent $i \neq r$ such that $\mathcal{X}^i$ and $\mathcal{Y}^i$ were allocated according to Subcase~2.B:
{

Let $l=[3]\setminus \{i,r\}$ be the other non-dividing agent.
W.l.o.g., agent~$r$ is the \textbf{subdivider} in $\mathcal{X}^i$, and agent~$l$ is the \textbf{subdivider} in $\mathcal{Y}^i$.
If $$\min_{Z\in \{X_r^i, X_l^i\}}v_r(Z) < \min_{Z\in \{Y_r^i, Y_l^i\}}v_r(Z)$$
(i.e -- agent $r$ prefers the minimum-valued bundle in the two-bundle partition $\{Y_r^i, Y_l^i\}$, created when agent $l$ was the subdivider), we proceed to construct new $\mathcal{X}^i$ and $\mathcal{Y}^i$ allocations as below. (Else, we continue to Stage 4). 
We invoke the two-agent procedure (\cref{alg:two-agents-efx}) with the following input:
\begin{enumerate}
    \item $v_1 = v_r, v_2=v_l$
    \item $\mathcal{X}^1=\mathcal{X}^2=\{Y^i_r,Y^i_l\}$
\end{enumerate}
We note that this implies (by \cref{obs:repartition-larger}) that for both agents $r$ and $l$, bundles $Y^i_r,Y^i_l$ are the top valued bundles in $\mathcal{Y}^i$.
We denote $T_i=Y^i_i$ and $\overline{T}_i = M\setminus T_i = Y^i_r \cup Y^i_l$. 

We recall that \cref{alg:two-agents-efx} outputs
two partitions $\{A_1^l,A_2^l\}, \{A_1^r,A_2^r\}$ {over the set of items $\overline{T}_i=Y^i_r \cup Y^i_l$} (it is possible that the two partitions are identical), where agent $l$ chooses a bundle first from $\mathcal{A}^r$, and vice versa. 

We construct $\mathcal{X}^i$ and $\mathcal{Y}^i$ using the two resulting allocations, as follows:
\begin{enumerate}
    \item $\mathcal{X}^i$: Agent $l$ chooses her favorite bundle in $\{A_1^r,A_2^r\}$, agent $r$ receives the remaining bundle in $\{A_1^r,A_2^r\}$, and agent $i$ receives $T_i$. We denote $\mathcal{C}^i_l=\mathcal{X}^i$.
     \item $\mathcal{Y}^i$: Agent $r$ chooses her favorite bundle in $\{A_1^l,A_2^l\}$, agent $l$ receives the remaining bundle in $\{A_1^l,A_2^l\}$, and agent $i$ receives $T_i$. We denote $\mathcal{C}^i_r=\mathcal{Y}^i$.
\end{enumerate}
}

In \cref{lem:stage-3-guarantees}, we show that all ex-post guarantees from \cref{thm:bobw-poly-approx}, except the EEFX guarantee, hold for the allocations constructed up to stage~3.

\paragraph{Stage 4 -- Agents Adopting a Permuted Partition of Another  (First time).}

Next, each agent \( k \in \{1,2,3\} \)  considers a set of candidate partitions, composed of the two certificates $\mathcal{C}_k^r$ for $r \in  [3]\setminus \{k\}$; 
Among these, she selects a candidate partition whose minimal-valued bundle has the highest value -- i.e, it has the best worst-case-value guarantee. We denote this partition as $\mathcal{S}^k$.
\\
If this minimal value exceeds the current bundle received by agent $k$ in $\mathcal{X}^k$ and $\mathcal{Y}^k$, i.e., 
\[ 
\min_{Z \in \mathcal{S}^k} v_k(Z) \;>\; v_k(X^k_k)
\ \ \ or\ \ \  
\min_{Z \in \mathcal{S}^k} v_k(Z) \;>\;  v_k(Y^k_k),
\]
then, the agent $k$, which we now also call \emph{the copier}, replaces her two derived allocations \( \mathcal{X}^k \) and \( \mathcal{Y}^k \), by a permutation of $\mathcal{S}^k$ as follows:

\begin{enumerate}
    \item In \( \mathcal{X}^k \), we pick any agent \( j \neq k \) (\emph{first-chooser}) and she selects her favorite bundle first, the remaining agent $i$ (not $k,j$) (\emph{second-chooser}) chooses second, and $k$ get the third bundle from $\mathcal{S}^k$. We denote $\mathcal{C}^k_j=\mathcal{X}^k$. 
    \item In \( \mathcal{Y}^k \), the roles of agents \( i \) and \( j \) are reversed, so the two switch their order of picking bundles from $\mathcal{S}^k$, and $k$ still gets the remaining bundle. We denote $\mathcal{C}^k_i=\mathcal{Y}^k$.
\end{enumerate}

\paragraph{Stage 5:  Agents Adopting a Permuted Partition of Another  (Second time).} Again, each agent \( k \in \{1,2,3\} \)  considers {a set of candidate partitions, composed of} the two certificate (partitions) $\mathcal{C}_k^r$ for $r \in [3]\setminus \{k\}$ {(which we note may have been updated during Stage 4)};  {We then proceed in an identical manner to Stage 4.}

In \cref{sec:3agents_poly_proof}, we prove that this algorithm guarantees \emph{ex-ante proportionality} for every agent (\cref{lem:two-copy-prop}),  
and that each allocation in the resulting distribution upholds all previously established guarantees as in \cref{thm:bobw-poly-approx} (\cref{lem:copy-stage-guarantees}).

\section{Conclusion}\label{sec:conclusion}

{
We have established that Best-of-Both-Worlds fairness guarantees can achieve substantially stronger outcomes than previously demonstrated, particularly for settings with few agents. 

For three agents, we made substantial progress on multiple fronts. We proved the existence of a distribution that is ex-ante proportional while guaranteeing ex-post EEFX and $\frac{9}{10}$-MMS, coming close to known upper bounds. Crucially, we introduced and achieved the \IME criterion, ensuring that any agent who falls short of her MMS receives instead the strongest possible envy-based guarantee. We complemented these existential results with an FPTAS that preserves nearly all guarantees while achieving exact ex-ante proportionality (and \EIMMX) -- a significant improvement over previous polynomial-time BoBW constructions, that guarantee only $\frac{1}{2}$-MMS \cite{babaioff2021bobw}.

{As a building block for our three-agent FPTAS, we presented the first Fully Polynomial-Time Approximation Scheme (FPTAS) that, for two agents, simultaneously achieves ex-ante envy-freeness, ex-post EFX, and \((1-\varepsilon)\)-MMS guarantees.  
This result is essentially tight for poly-time algorithms, and it matches the strongest feasible fairness notions across all three dimensions.}

Beyond these algorithmic contributions, our work advances the practical applicability of BoBW algorithms. First, our constructions maintain small support sizes (at most six allocations), enabling agents to verify ex-ante fairness explicitly and even conduct their own randomization if desired. Second, we provide polynomial-time verifiable certificates for computationally complex properties like EEFX, allowing agents to independently confirm that fairness guarantees hold without trusting the algorithm or accessing other agents' valuations.

Several directions remain open for future work. The most pressing challenge is extending our results beyond three agents while maintaining strong guarantees -- particularly achieving BoBW distributions that are ex-ante proportional and ex-post EEFX with near-optimal MMS approximations. For three agents, strengthening the ex-ante guarantee to envy-freeness (rather than just proportionality) while preserving strong ex-post guarantees remains an interesting open problem. Additionally, it would be valuable to close the gap between our $\frac{9}{10}$-MMS guarantee and the upper bound of $\frac{39}{40}$. {Beyond three agents, it will be interesting} to determine whether our \IMMX criterion can be guaranteed for $n \geq 4$ agents, since the existence of EFX  allocations remains unresolved for four or more agents.

Our results demonstrate that the apparent tension between ex-ante and ex-post fairness is less severe than the gaps left by prior BoBW results, at least for small numbers of agents. By carefully exploiting the structure of additive valuations and developing new algorithmic techniques, we have shown that it is possible to achieve near-optimal fairness guarantees across multiple dimensions simultaneously, while maintaining computational efficiency and practical verifiability.

}

\printbibliography

\appendix
\crefalias{section}{appendix}
\crefname{section}{appendix}{appendices}
\Crefname{section}{Appendix}{Appendices}

\section{EFX Re-allocation Algorithm for $n$ Agents} \label{sec:efx-reallocation}

In this section, we describe \textsc{Realloc}, a minor adaptation of Algorithm~6.1 introduced by \citet{PlautRoughgarden2020} and, for completeness, prove the following lemma:
\begin{lemma}\label{lem:new_realloc}
Fix an additive valuation $v$, any number of agents $n$, and a partition $(B_1,B_2,\ldots,B_n)$. 
Then there is a poly-time algorithm that transforms $(B_1,B_2,\ldots,B_n)$ to 
a partition 

$(B_1',B_2',\ldots,B_n')$ that is EFX with respect to $v$, and for which
$$min (v(B'_1),v(B'_2),\ldots,v(B_n'))\geq min (v(B_1),v(B_2),\ldots,v(B_n))$$    
\end{lemma}

Before we prove \Cref{lem:new_realloc}, we derive the following corollary:

\begin{corollary}\label{cor:poly_mms}
There exists a polynomial-time algorithm that given 
an additive valuation $v$ over $M$, a number of agents $n$, and any $MMS(v,n)$ partition of the items $M$ for the agent, 
outputs a partition that is $\mathcal{MMS-EFX}$ for the agent.

\end{corollary}

\begin{proof} 
Starting from the agent's MMS partition $MMS(v,n)$, apply the algorithm of \Cref{lem:new_realloc} to the bundles in this partition with respect to $v$.  
By \Cref{lem:new_realloc}, the resulting partition is EFX for $v$ and guarantees that the minimum bundle value does not decrease.  
Thus, each bundle in the transformed partition EFX-satisfies the agent, and guarantees her MMS share, i.e., the partition is an $\mathcal{MMS-EFX}$ partition for the agent.  
\end{proof}

\paragraph{Proof of \Cref{lem:new_realloc}}

\begin{proof}  
Let $(B_1,B_2,\ldots,B_n)$ be a partition of a set $M$ of $m$ goods, and $v$ an additive valuation over $M$. 
Assume w.l.o.g. that the items are sorted in non-increasing order of their values, that is,  $v(\{w_1\})\geq v(\{w_2\})\geq \ldots \geq v(\{w_m\})$. 

We now describe \textsc{Realloc}, a minor adaptation of Algorithm~6.1 in \cite{PlautRoughgarden2020}.  
Our procedure differs in two respects: it takes as input a partition of the goods (rather than starting from $n$ empty bundles), and it is restricted to the case of identical valuations, whereas Algorithm~6.1 handles additive but non-identical valuations under the assumption of identical preference orderings.

\begin{algorithm}[H]
\caption{\textsc{Realloc} -- EFX Re-allocation}
\label{alg:realloc}
\DontPrintSemicolon

\KwIn{
    \begin{itemize}
        \item A partition $\mathcal B=(B_{1},\dots,B_{n})$ of a set $M$ of $m$ goods to $n$ bundles,
        \item an additive valuation $v:2^{M}\!\to\mathbb{R}_{\ge 0}$.
    \end{itemize}
}
\KwOut{A partition that weakly improves the minimum bundle value, and is also EFX for $v$.}

Let $(w_{1}, w_{2}, \dots, w_{m})$ be a non-increasing ordering of the set of goods $M$ ($v(\{w_1\}) \geq \dots \geq v(\{w_m\})$)\;

\For{$i\gets 1$ \textbf{to} $m$}{
    \textbf{identify the bundle $B_{k}$ currently holding $w_i$:}\;
    $k \gets \text{index such that } w_{i}\in B_{k}$\;
    
    \textbf{remove item $w_{i}$ from bundle $B_{k}$:}\;
    $B_{k} \gets B_{k}\setminus\{w_{i}\}$\;
    
    \textbf{find the bundle $B_{j}$ with smallest value:}\;
    $j \gets \arg\min_{t\in \{1,\dots,n\}} v(B_{t})$ \tcp*[f]{break ties arbitrarily}
    
    \textbf{insert item $w_{i}$ to $B_{j}$: }\;
    $B_{j} \gets B_{j}\cup\{w_{i}\}$\;
}

\KwRet $(B_{1},\dots,B_{n})$
\end{algorithm}

We establish two structural lemmas and an
immediate corollary.

\begin{lemma}[Monotone worst–bundle value]\label{lem:monotone-min}
Let 
\(\mathcal B^{(i)}=(B^{(i)}_{1},\dots,B^{(i)}_{n})\) be the partition after
the \(i\)-th iteration of \textsc{Realloc}, where
\(i\in\{0,\dots,m\}\) (where $ B^{(0)}=(B^{(0)}_{1},\dots,B^{(0)}_{n})$ denotes the input to the algorithm)  and \(m=|M|\).
Then
\[
   \min_{t\in[n]} v\!\bigl(B^{(i)}_{t}\bigr)
   \;\;\le\;\;
   \min_{t\in[n]} v\!\bigl(B^{(i+1)}_{t}\bigr)
   \quad\text{for every } i=0,\dots,m-1.
\]
In other words, the value of the lowest-valued bundle never decreases as the
algorithm proceeds.
\end{lemma}

\begin{corollary}[Preservation of the initial minimum]\label{cor:preservs}
Let \(\mathcal B=(B_{1},\dots,B_{n})\) be the input partition and
\(\mathcal B'=(B'_{1},\dots,B'_{n})\) the partition returned by
\textsc{Realloc}.  Then
\[
   \min_{t} v(B'_{t})
   \;\ge\;
   \min_{t} v(B_{t}),
\]
\end{corollary}

\begin{lemma}[EFX guarantee of the output]\label{lem:efx-output}
Let $B'$ be the partition produced by \textsc{Realloc}. Then every bundle in $B'$ is
\emph{envy-free-up-to-any-item} (EFX) with respect to the valuation~\(v\).
\end{lemma}

\paragraph{Proof of Lemma \ref{lem:monotone-min}}
\begin{proof}
Let \( w_i \) denote the item moved in round \( i \), and let \( B_k \) be the bundle containing \( w_i \) in the initial allocation. Let us remove $w_i$ from $B_k$ and define:
$$B^\star=\arg\min_{B_j}\{ v(B_j)\}$$

If this is not unique, choose one of the bundles at random.

First, observe that in round \( i \), only the bundles \( B_k \) and \( B^\star \) could potentially change, while all other bundles remain unchanged. Therefore, their values remain the same.

Consider two cases:
\begin{enumerate}
    \item If \( B_k = B^\star \), then by the algorithm's definition, \( w_i \) remains in \( B_k \), and no bundle's value changes. Therefore, \( \min_{B_j} v(B_j) \) remains unchanged.
    \item Otherwise, \( w_i \) is moved from bundle \( B_k \) to bundle \( B^\star \). By definition, we have:
\[
v(B^\star \cup \{w_i\}) \geq v(B^\star) = \min_{B_j} v(B_j).
\]

Additionally, from the selection of \( B^\star \), we know:
\[
\min_{B_j} v(B_j) = v(B^\star) \leq v(B_k \setminus \{w_i\}).
\]

Therefore, after reallocating \( w_i \), every bundle has a value at least as large as the original minimum \( \min_{B_j} v(B_j) \). Thus, the minimum bundle value cannot decrease; it can only stay the same or increase.
\end{enumerate}

This completes the proof.
\end{proof}

\paragraph{Proof of Lemma ~\ref{lem:efx-output}}

\begin{proof}
Let the final allocation produced by the algorithm be $(B'_1, B'_2, \dots, B'_n)$. For contradiction, assume the resulting allocation is not EFX. Then there exist two bundles \( B'_i, B'_j \) such that:
\[
v(B'_i) < v(B'_j) - \min_{w_j \in B'_j} v(\{w_j\}).
\]

Let \( w_j \) be the item of minimum value last inserted in \( B'_j \). Consider the round in the algorithm when \( w_j \) was placed into bundle \( B'_j \). At this step, \( w_j \) was either:

\begin{enumerate}
    \item Added to the bundle with the minimum current value, or
    \item Remained in its original bundle because removing it would have decreased that bundle's value below the minimal bundle value.
\end{enumerate}

Specifically, at the time \( w_j \) was placed in \( B'_j \), we had:
$$B_j'=\arg\min_{B_j}\{ v(B_j)\}$$

At this moment, by construction, the bundle \( B'_j \) satisfies the EFX condition relative to every other bundle:
\[
\forall k \in [n],\quad v(B'_j \setminus \{w_j\}) \leq v(B_k).
\]

Once \( w_j \) is placed, there are only two ways the EFX condition for \( B'_j \) could later fail:

\begin{itemize}
    \item \textbf{Case 1 (increase of value in \( B'_j \)):}  
    The value of \( B'_j \) could become strictly larger. However, since \( w_j \) is by assumption the smallest-value item in \( B'_j \), no new smaller-valued item can later be added to increase the discrepancy that violates EFX (according to the running of the algorithm). 

    \item \textbf{Case 2 (decrease of value in another bundle):}  
    The value of another bundle \( B_k \) could become smaller than \( v(B'_j \setminus \{w_j\}) \). However, at the time \( w_j \) was placed, \( v(B'_j \setminus \{w_j\}) \) was less than or equal to the minimal bundle value. By Lemma~\ref{lem:monotone-min}, we know the minimum bundle value cannot decrease over the course of the algorithm. Thus, no bundle can decrease below \( v(B'_j \setminus \{w_j\}) \).
\end{itemize}

Since both cases are impossible, no EFX violation can occur at any later stage. Thus, our assumption leads to a contradiction, and the final allocation must indeed be EFX.
\end{proof}

\end{proof}

\section{Omitted Proofs for Three Agents}

\subsection{Proof of Technical Lemmas}\label{subsec:technial-lemma}

In this section we prove \Cref{lem:guarantee-value}, \Cref{lem:reallocation-is-still-EFX}, \cref{lem:total-val-case-2} and \Cref{lem:val-910-mms}.

Let \( \{A, B, C\} \) be a partition of the item set \( M \), and let agent \( j \in \{1,2,3\} \) be fixed. Assume without loss of generality that
\[
v_j(A) \geq \max [v_j(B),v_j(C)].
\]

We first state and prove a useful lemma:

\begin{lemma}\label{lem:bundle-at-least-prop-is-EEFX}
Let $X \subseteq M$. If $v_j(X)\geq PROP_j$, then agent $j$ is EEFX-satisfied with $X$. 
\end{lemma}

\begin{proof}
We construct an EEFX certificate for bundle $X$ and agent $j$, by presenting a partition of the set $M \setminus X$ to two bundles $Y$ and $W$, such that agent $j$ is EFX-satisfied with any bundle in $\{X,Y,W \}$.

Let $S = M \setminus X$ be the set of remaining goods. 
First, we notice that if $v_j(S) \leq v_j(X)$, then by defining $W=S,Y=\emptyset$, $(W,Y,X)$ is an EFX-certificate for $X$. Therefore, from now on we assume $$v_j(S) > v_j(X)$$

Order the goods in $S$ as $s_1, s_2, \dots, s_t$ such that $v_j(s_1) \geq v_j(s_2) \geq \dots \geq v_j(s_t)$. 

We initialize $Y = \emptyset$ and add goods from $S$ one by one according to this order until the total value of $Y$ exceeds $v_j(X)$.  That is, let $k$ be the minimal index such that $v_j(\{s_1, \dots, s_k\}) > v_j(X)$, and define:
\[
Y = \{s_1, \dots, s_k\}, \quad W = S \setminus Y.
\]

Let $g = s_k$ be the least-valued item in $Y$. Since $v_j(Y \setminus \{g\}) \leq v_j(X) < v_j(Y)$  by construction, it follows that $X$ EFX-dominates $Y$, i.e., $v_j(X) \geq v_j(Y \setminus \{g\})$.

We now observe that

$$v_j(X) \geq PROP_j$$  
and
$$v_j(Y) > v_j(X) \geq PROP_j$$
Since the total value of all goods is $v_j(M)$ and $PROP_j = \frac{v_j(M)}{3}$ (in the 3-agent case), we must have:
\[
v_j(W) = v_j(M) - v_j(Y) - v_j(X) \leq v_j(M) - 2\cdot PROP_j = PROP_j.
\]
Thus, $v_j(W) \leq PROP_j \leq v_j(X)$.

It follows that $(Y, W, X)$ is EFX for $v_j$, so it serves as an EEFX certificate for agent $j$, therefore agent $j$ is EEFX-satisfied with $X$.
\end{proof}

\paragraph{Proof of \Cref{lem:guarantee-value}}
\lemtopval*

\begin{proof}
Let us assume w.l.o.g that $S_2=argmax\{v(S_2),v(S_3)\}$.
Assume in contradiction 
that $$v(S_1)+v(S_2)<v(M) - MMS(v,3).$$
Since $$v(M)=v(S_1)+v(S_2)+v(S_3)$$
this would force \( v(S_3)>MMS(v,3)\), contradicting the definition of MMS,
which maximizes the minimum bundle value.  Therefore the claimed inequality must hold.
\end{proof}

{

\paragraph{Proof of \cref{lem:total-val-case-2}}

We first prove a useful observation. 
\begin{observation}\label{obs:repartition-larger} 
    Let $\mathcal{P}^l= \{P^l_1,P^l_2\}$ be a partition chosen by agent $l$ in Line 13 in \cref{alg:three-agent-construction}, and recall that $L_l=M\setminus (P^l_1 \cup P^l_2)$.
    Then $$\min \{v_l(P^l_1), v_l(P^l_2) \} \geq v_l(L_l)$$
\end{observation}

\begin{proof} 
    Assume by contradiction that $$\min \{v_l(P^l_1), v_l(P^l_2) \} < v_l(L_l).$$ W.l.o.g, let us assume that $\mathcal{P}^l =\mathcal{P}^l_B $ (i.e -- it is a partition of $A \cup B$, and $L_l = C$). 
    Let us look at the partition $\mathcal{P}^l_C$. Since $v_l(A) \geq v_l(C)$, by using the \textbf{MMS-EFX-Improved-RePartition} algorithm we have $\min_{P \in \mathcal{P}^l_C}v_l(P)\geq v_l(C)$. But this leads to $$\min_{P \in \mathcal{P}^l_C}v_l(P)\geq v_l(C) > \min_{P \in \mathcal{P}_B^l}v_l(P),$$ contradicting the fact that  $$\min \{v_l(P^l_1), v_l(P^l_2) \} = \min_{P \in P^l_B} \{v_l(P)\} \geq \min_{P \in P^l_C} \{v_l(P)\}$$
\end{proof}
}

\totalvaluelemma*

\begin{proof}
We show that across allocations $\mathcal{X}^1$ and $\mathcal{Y}^1$, each agent $t \in \{2,3\}$ obtains a total value 
$v_t(X^1_t) + v_t(Y^1_t)$ which is  at least $v_t(M) - MMS_t$.

\begin{itemize}
   \item \textbf{Case 1: Both agents partition the same pair.} 
Suppose both agents choose to partition the same pair of bundles, say $A \cup B$, yielding a partition $\mathcal{P}^t=\{P^t_1,P^t_2\}$ for each $t \in \{2,3\}$. 
By \cref{obs:repartition-larger}
\[
\min\{v_t(P^t_1),v_t(P^t_2)\} \;\geq\; v_t(C) \quad \text{for each } t \in \{2,3\},
\]
Therefore $\{P^t_1,P^t_2,C\}$ is a partition of $M$, such that for each agent $t$, $P^t_1,P^t_2$ are top two valued bundles in the partition.    
The agents then perform in each allocation a cut-and-choose on $A \cup B$ (with reversed roles in each allocation), so each obtains total value across the two allocations of at least $v_t(P^t_1)+v_t(P^t_2)$.  
Therefore, by \Cref{lem:guarantee-value}, the claim follows.

    \item \textbf{Case 2: The agents partition different pairs.} 
    Without loss of generality, suppose that when agent 3 is the subdivider she 
    partitions $A \cup C$, while when agent 2 is the subdivider she 
    partitions $A \cup B$. Let $\mathcal{X}^1$ be the allocation where agent 3 is the subdivider, and $\mathcal{Y}^1$ be the allocation where agent 3 is the chooser.
    Consider agent~3. 
    The following inequalities hold:
    \begin{enumerate}
        \item $v_3(X^1_3)\geq  v_3(C) = v_3(Y^1_1)$: \\
        First, since agent~2 partitioned $A \cup B$ and we are in Subcase~2.B, agent~1 receives bundle~$C$ in $\mathcal{Y}^1$, so $ v_3(Y^1_1)=v_3(C) $. 
        Second, since agent~3 partitioned $A \cup C$, by the definition of the $\mathcal{MMS\text{-}EFX}$ partition, the least-valued bundle in that partition is worth at least $v_3(C)< v_3(A)$ (as $A$ is agent~3’s unique top bundle), and thus $v_3(X^1_3)\geq  v_3(C)$.

        \item $v_3(Y^1_3) \geq v_3(Y_2^1)$: \\
        we recall that agent 3 was the chooser in allocation $\mathcal{Y}^1$, and since we are in Subcase 2.B -- the chooser receives her favorite bundle in the partition.  
        \item $v_3(X^1_3) \geq v_3(Y^1_2)$: \\
        By definition, $\min \{ v_3(X^1_2), v_3(X^1_3)) \geq \min\{v_3(Y_2^1),v_3(Y^1_3)\}$ -- if this was not the case, agent $3$ could have created the partition $(Y_2^1,Y^1_3)$ instead.
        Recall that agent 3 was the chooser in $\mathcal{Y}^1$, and therefore $v_3(Y^1_3) \geq v_3(Y^1_2)$. i.e -- $v_3(Y^1_2)=\min\{v_3(Y_2^1),v_3(Y^1_3)\}$. Therefore -- $ v_3(X^1_3) \geq \min \{ v_3(X^1_2), v_3(X^1_3)\} \geq \min\{v_3(Y_2^1),v_3(Y^1_3)\}= v_3(Y^1_2)$
    \end{enumerate}
    Assume, for contradiction, that
    \[
    (\star) \qquad v_3(X^1_3) + v_3(Y^1_3) \;<\; v_3(M) - MMS_3.
    \]
    Since $v_3(X^1_3)\geq v_3(Y^1_1)$  [inequality (1)], this leads to 
    \[
    v_3(M) - v_3(Y^1_2) 
      \;=\; v_3(Y^1_1) + v_3(Y^1_3) 
      \;\leq\; v_3(X^1_3) + v_3(Y^1_3) 
      \;<\; v_3(M) - MMS_3,
    \]

    which implies that $v_3(Y^1_2) > MMS_3$.

    Moreover, combining $v_3(Y^1_2) > MMS_3$ and $v_3(Y^1_3) \geq \max \{v_3(Y_2^1),v_3(Y^1_3)\}\geq v_3(Y^1_2) > MMS_3$ [inequality (2)], it follows that $v_3(Y^1_1) \leq MMS_3$ (by definition of the MMS). Therefore,
    \[
    v_3(M) - v_3(Y^1_1)  
      \;\geq\; v_3(M) - MMS_3.
    \]
    But since $v_3(M) - v_3(Y^1_1) = v_3(Y^1_3) + v_3(Y^1_2)$, this leads to 
    \[
    v_3(Y^1_3) + v_3(Y^1_2) \;\geq\; v_3(M) - MMS_3.
    \]
    Adding in $v_3(X^1_3) \geq v_3(Y^1_2)$ [inequality (3)], we therefore have 
    \[
    v_3(Y^1_3) + v_3(X^1_3) \;\geq\; v_3(Y^1_3) + v_3(Y^1_2) \;\geq\; v_3(M) - MMS_3,
    \]
    contradicting $(\star)$.
\end{itemize}

Thus, in both cases each agent $t \in \{2,3\}$ receives total value at least $ v_t(M)-MMS_t$, as claimed.
\end{proof}

{
\paragraph{Proof of \Cref{lem:reallocation-is-still-EFX}.}

\lemreallocefxl*

\begin{proof}
Let us assume w.l.o.g that $v_l(P_1^l) \ge v_l(P_2^l)$. We note that due to the \realloc process in \cref{alg:partition}, the  bundle $P_2^l$ EFX-dominates $P_1^l$. 
Moreover, by \cref{obs:repartition-larger} we must also have $v_l(P_2^l) \geq v_l(L_l)$. 
It follows that $P_2^l$ EFX-dominates $L_l$. From \cref{obs:top-is-EFX}, $P_1^l$ EFX dominates both $P_2^l$ and $L_l$. Therefore both $P_2^l$ and $P_1^l$ are EFX-satisfactory for agent $l$ with regards to the set of items $M$ and $n=3$ agents. 
\end{proof}

\paragraph{Proof of \Cref{lem:val-910-mms}.}

\lemninemms*

\begin{proof}
Without loss of generality, assume that \( v_l(P_1^l) \ge v_l(P_2^l) \). 
We show that \( v_l(P_2^l) \) is bounded below by the guarantee provided by the \emph{coarse atomic partial search} algorithm of \cite{amanatidisapprox2017,GOURVES201950}, whose analysis was later strengthened in \cite{feige2022improvedmaximinfairallocation} to achieve a $\tfrac{9}{10}$-MMS allocation. 
\citet{feige2022improvedmaximinfairallocation} demonstrate that if there exists only a single bundle in $\{A,B,C\}$ of value to agent $l$ at least as high as her MMS share (bundle $A$), then when considering some partitions of \( (A,B,C) \) into \( (X,Y,W) \) where \( W \in \{B,C\} \) and maximizing the value of \( \min\{v_l(X),v_l(Y)\} \) -- agent~\( l \) can be guaranteed at least a $\tfrac{9}{10}$-MMS value. 
If there are two such bundles, \citet{feige2022improvedmaximinfairallocation} guarantee agent $l$ her MMS share by giving her one of those bundles directly.

In our algorithm, we consider the two partitions generated by the algorithm \textbf{MMS-EFX-Improved-RePartition} (\cref{alg:partition}) with either the input $\{A, B\}$ or the input $\{A,C\}$ (and valuation $v_l$), and select the partition maximizing the minimal-bundle value. Therefore if there is only one bundle guaranteeing the MMS share (bundle $A$), we guarantee $\frac{9}{10}$ of the MMS share, and if there are two such bundles -- we guarantee at least the MMS share, using the fact that the \textbf{MMS-EFX-Improved-RePartition} does not diminish the minimum-valued bundle in the partition. 

Therefore, the value of the minimal-valued bundle in our construction is at least as high as the value of the bundle guaranteed by the coarse atomic partial search of \cite{feige2022improvedmaximinfairallocation}.

\end{proof}
}

\subsection{Time Complexity}\label{sec:timecomplex}

We prove \Cref{lem:poly-time-3}:

\polytimethree*

\begin{proof}

\smallskip\noindent\textbf{Step 1: Oracle calls (MMS partitions).}
For each agent $i$, call the oracle once to obtain an MMS partition of $M$ into $3$ bundles. We then apply $\textsc{Realloc}$ (which runs in polynomial time) (\Cref{lem:new_realloc}) in order to obtain the $\mathcal{MMS-EFX}^i$ partition for each agent $i\in \{1,2,3\}$. 

\smallskip\noindent\textbf{Step 2: Construct the candidate allocations.}
For each agent $i$'s $\mathcal{MMS-EFX}^i$ allocation, our algorithm constructs at most two deterministic allocations, $\mathcal{X}^i$ and $\mathcal{Y}^i$ (and in some cases a single allocation). 
We recall that if each agent $j\neq i$ can receive a different top-valued bundle (which can be checked in polynomial time), no further calculations are needed to construct  $\mathcal{X}^i$ and $\mathcal{Y}^i$. Therefore, we are left with Case 2, the case when both agent's top-valued bundle is the same bundle in the partition, say bundle $A$. 
Each construction requires only:
\begin{enumerate}
  \item Computing, for each relevant pair of bundles $S\subseteq\{A\cup B, A\cup C\}$ and each agent $j\in\{2,3\}$, an {$\textbf{MMS-EFX-Improved-RePartition}$ with input $S$ into two bundles according to $v_j$}. This is done by:
  \begin{itemize}
    \item one oracle call to obtain an MMS partition of $S$ for $j$ (size-$2$ partition for $n=2$), and
    \item applying the $\textsc{Realloc}$ transformation of \Cref{lem:new_realloc} to make it EFX for $v_j$ without decreasing the minimum bundle value.
  \end{itemize}
  The oracle call is assumed; by \Cref{lem:new_realloc}, $\textsc{Realloc}$ runs in time polynomial in $m$.
  \item Simple value comparisons and $\arg\max$ over a constant set ($\{A\cup B, A\cup C\}$) to choose the partition that maximizes the minimum bundle value for each $j$.
  \item Assembling $\mathcal{X}$ and $\mathcal{Y}$ from these partitions.
\end{enumerate}
Thus, aside from the oracle calls, constructing $\mathcal{X}$ and $\mathcal{Y}$ is polynomial in $m$.

\smallskip\noindent\textbf{Complexity bound.}
Overall, we make at most 15 calls for computing {an agents MMS partition (3 calls, one for each agent to compute her MMS  partition of $M$ to 3 bundles, and for each $i\in [3]$, at most 4 calls for computation of MMS partitions of the other two agents to 2 bundles, in the process of constructing $\mathcal{X}^i$ and $\mathcal{Y}^i$).  As we assume that MMS partitions are handled by the oracle, then given such oracles the distribution~$\mu$ can be computed in  polynomial time.}

\smallskip
This establishes the claim.
\end{proof}

\subsection{Tightness of Our Analysis}\label{tightmms}

We remark that the analysis of the  MMS approximation of our algorithm is tight. 
\begin{remark}
There exists an instance with three agents and additive valuations in which, under our algorithm, some agent receives exactly $\tfrac{9}{10}$ of her maximin share (MMS).
\end{remark}
The instance and analysis are identical to the instance given in \cite{feige2022improvedmaximinfairallocation}, Section 3.7.
This result highlights that our analysis is tight for the proposed construction. 
Determining the best possible upper bound on the MMS fraction achievable simultaneously with EEFX and ex-ante proportionality remains an open question. 
For context, \citet{feige2021tight} showed that there exist instances in which no allocation can guarantee more than a $\tfrac{39}{40}$-fraction of MMS. 
Thus, even with no ex-ante requirement it is impossible to guarantee a fraction larger than $97.5\%$ of the MMS, and we guarantee $90\%$ (with ex-ante proportionality), leaving a gap of only $7.5\%$.

\subsection{FPTAS enables MMS computation
}\label{sec:fptaspolynomial}
If all item values are integers bounded by $\mathrm{poly}(|M|)$, then there exists a polynomial-time algorithm that outputs an allocation satisfying the BoBW guarantee of \cref{thm:three-agents-main}, as exact $MMS$ partitions can be computed in polynomial time.

We next sketch the proof for completeness.
Let $v$ be a valuation and denote $W = v(M) = \mathrm{poly}(|M|)$.  
Set $\varepsilon = \frac{1}{10W}$. By the definition of an FPTAS, we can compute a $(1-\varepsilon)$-approximation of the $MMS$ value in time polynomial in $|M|$ and $1/\varepsilon = 10W$.

Since $MMS \le W$,
\[
MMS - (1-\varepsilon){MMS} \le \varepsilon \cdot MMS \le \varepsilon W = \tfrac{1}{10}.
\]

Because all valuations are integral, the $MMS$ value is an integer. Therefore any estimate within additive error strictly smaller than $1$ must equal the exact MMS value.
Consequently, an exact $MMS$ value -- and therefore an exact $MMS$ partition -- can be computed in polynomial time.

\section{Omitted Proofs for Two Agents}

\subsection{Proof of \cref{the:main-2agents}}\label{sec:2agents}
\twoagentsmainprop*

We first present an overview of the proof of \cref{the:main-2agents}.
To obtain a polynomial-time algorithm which outputs a distribution over allocations that is ex-ante proportional, ex-post EFX, and ex-post $(1-\varepsilon)\cdot \mathrm{MMS}$, we
 \emph{slightly modify Algorithm~3}  of \cite{bu2024bestofbothworldsfairallocationindivisible} and {\emph{combine this with} any FPTAS (or PTAS) for the MMS partition problem (e.g., the FPTAS of \cite{woegingerFPTAS}, or of \cite{Ibarra1975}).

 Concretely, we 
 {\begin{enumerate}
     \item seed Algorithm~3 with agent-specific bipartitions \(\mathcal{A}^i=\{A^i_1,A^i_2\}\) (instead of \(\{\emptyset,M\}\)) with both bundles having value at least $(1-\varepsilon)\cdot MMS_i$, and
     \item add an improvement guard that updates agent $i's$ candidate split only if it remains EFX for agent $i$ and strictly increases her minimum bundle value.
 \end{enumerate}}
 The resulting procedure outputs a lottery supported on at most two allocations that is \emph{ex-ante envy-free} (hence proportional) and \emph{ex-post EFX}, and every realized allocation guarantees each agent \(i\) at least \(\min\{v_i(A^i_1),v_i(A^i_2)\}\). 
  }

{We next present our algorithm for two agents, and the proof of \cref{the:main-2agents}.}

Our algorithm slightly modifies Algorithm~3 of \cite{bu2024bestofbothworldsfairallocationindivisible}. Before presenting our modifications,
we first describe Algorithm~3 of \cite{bu2024bestofbothworldsfairallocationindivisible}. 

\paragraph{Algorithm~3 of \cite{bu2024bestofbothworldsfairallocationindivisible}:}
The algorithm begins by first creating two candidate partitions \(\mathcal{A}^1=\{A^1_1,A^1_2\}\) and \(\mathcal{A}^2=\{A^2_1,A^2_2\}\), where \(\mathcal{A}^i\) is EFX with respect to agent \(i\)'s valuation \(v_i\).
Each allocation $\mathcal{A}^i$ is computed by applying a subroutine called \textsc{LocalSearch} (\cref{alg:local-search}) on the partition $\{\emptyset, M\}$ and valuation $v_i$. The subroutine \textsc{LocalSearch} is one that given any partition \(\{A,B\}\)  and valuation \(v\), returns, in polynomial time, a partition \(\{A',B'\}\) that is EFX with regards to valuation $v$, and has a weakly smaller \emph{value gap}, i.e.,
\(|v(A')-v(B')| \le |v(A)-v(B)|\). This invariant is maintained throughout the algorithm via the subroutine \textsc{LocalSearch}\((A,B,v)\). 

The algorithm proceeds by iteratively refining the two candidate partitions. 
In each step, it attempts to improve the value that agent~$i$ assigns to her 
least-preferred bundle in $\mathcal{A}^i$, thereby increasing her guaranteed utility, 
while at the same time preserving the invariant that $\mathcal{A}^i$ remains 
\emph{EFX} with respect to $v_i$. 
At termination, the algorithm outputs either one EFX allocation (which is EF) or a lottery 
over two such allocations, which guarantees that each agent~$i$ receives 
her proportional share in expectation, thereby ensuring ex-ante proportionality (and ex-ante EF). Specifically:
if at any point throughout the algorithm, either \(\mathcal{A}^1\) or \(\mathcal{A}^2\) is (ex-post) envy-free under the profile \((v_1,v_2)\), the algorithm returns that partition immediately (getting EF ex-post).
Otherwise, compare each of the two candidate partitions under the \emph{other} agent's valuation:
if there exists \(i\in\{1,2\}\) such that the value gap of \(\mathcal{A}^i\) measured by \(v_{3-i}\) is \emph{strictly} smaller than the value gap of \(\mathcal{A}^{3-i}\) under \(v_{3-i}\), update \(\mathcal{A}^{3-i}\) by calling
\textsc{LocalSearch} on \(\mathcal{A}^i\) with valuation \(v_{3-i}\), thereby obtaining a new EFX candidate with no larger value gap. 
If no such improving update exists, output the lottery \(\{(1/2,\mathcal{A}^1),(1/2,\mathcal{A}^2)\}\).

As shown in \cite{bu2024bestofbothworldsfairallocationindivisible}, this procedure runs in polynomial time; its output distribution is ex-ante envy-free, and every supporting allocation is EFX.

In \Cref{subsec:worse-bu} we present an example showing that Algorithm~3 of \cite{bu2024bestofbothworldsfairallocationindivisible} might give an agent no more than {85\%} of her MMS. We next discuss our modified version that overcomes 
this failure to guarantee $1-\varepsilon$ fraction of the MMS -- with our modified version we are able to guarantee $1-\varepsilon$ fraction of the MMS in polynomial time, a fraction of the MMS that is the best that we can hope to guarantee with efficient computation.

\paragraph{Our Algorithmic Modification:}

We introduce two modifications to Algorithm~3 of \cite{bu2024bestofbothworldsfairallocationindivisible}.
\begin{enumerate}
  \item \textbf{Initialization.} Instead of seeding \(\mathcal{A}^1\) and \(\mathcal{A}^2\) via \textsc{LocalSearch} on \(\{\emptyset,M\}\), we accept as input two arbitrary bipartitions of \(M\), denoted \(\mathcal{X}^1=\{X^1_1,X^1_2\}\) and \(\mathcal{X}^2=\{X^2_1,X^2_2\}\), and set $\mathcal{A}^i =(A_1^i, A_2^i) \gets \textsc{LocalSearch}({\mathbf{X^i_1, X^i_2}}, v_i)$ for \(i\in\{1,2\}\).
  This initialization allows starting with a partition for each agent in which both bundles have high value (which will never drop later).
  \item \textbf{Minimal-value update rule.} In Line 11, we replace \(\mathcal{A}^i\) by the other agent’s candidate partition \(\mathcal{A}^{3-i}\) only if it \emph{strictly} improves agent \(i\)’s minimum-bundle value:
  \[
    \min\{v_i(A^{3-i}_1),\,v_i(A^{3-i}_2)\}\;>\;\min\{v_i(A^i_1),\,v_i(A^i_2)\}.
  \]
  Equivalently, agent \(i\) adopts the other partition if and only if she prefers its lowest value bundle over the lowest value bundle in her own partition. In the original algorithm, the replacement took place even when there was no minimum value improvement. 
\end{enumerate}

Both changes are made in order to guarantee a minimum value. 

The pseudocode below describes Algorithm 3 by \cite{bu2024bestofbothworldsfairallocationindivisible} with our minor changes.
In \Cref{alg:3} we also give the original pseudocode, with our modifications 
in \blue{blue}.

\begin{minipage}{\textwidth}

\renewcommand*\footnoterule{}

\begin{savenotes}
\begin{algorithm}[H]
\caption{An ex-ante EF and ex-post EFX randomized allocation for two additive agents, with minimum value improvement}
\label{alg:two-agents-efx}
\DontPrintSemicolon

\KwIn{Two additive valuations $v_1,v_2$ over a set of indivisible goods $M$. 
\blue{Two initial partitions of $M$, $\mathcal{X}^1=\{X^1_1,X^1_2\}$ and $\mathcal{X}^2=\{X^2_1,X^2_2\}$}}
\KwOut{A uniform distribution over (at most two) partitions, that is ex-ante EF. Additionally, for each partition, each agent $i$ is EFX-satisfied and receives a value of at least $\min \{v_i(X^i_1), v_i(X^i_2)\}$.}

$\mathcal{A}^1 =(A_1^1, A_2^1) \gets \textsc{LocalSearch}(\blue{\mathbf{X^1_1, X^1_2}}, v_1)$ \tcp*{$v_1(A_1^1) \leq v_1(A_2^1)$}
$\mathcal{A}^2 =(A_1^2, A_2^2) \gets \textsc{LocalSearch}(\blue{\mathbf{X^2_1, X^2_2}}, v_2)$ \tcp*{$v_2(A_1^2) \leq v_2(A_2^2)$}

\If{$\exists\, i \in \{1,2\}$ such that $v_i(A^i_1)=v_i(A^i_2)$ \textbf{ or } $v_{3-i}(A^i_1)\ge v_{3-i}(A^i_2)$}{
    \KwRet $\{(1, \mathcal{A}^i)\}$ where agent $3-i$ picks her preferred bundle\footnote{Although the usual notation assigns bundle $A^i_j$ to agent $j$, we slightly abuse notation here by treating $A^i_1$ and $A^i_2$ as a \emph{set of two bundles} rather than an allocation (with assignment to agents). Thus, agent $3-i$ selects her preferred bundle among the two, and the remaining bundle is assigned to agent $i$.} from $\{A^i_1,A^i_2\}$\;
}

\While{$\exists\, i \in \{1,2\}$ such that $v_{3-i}(A^i_1) \;>\;v_{3-i}(A^{3-i}_1)$}{
    $\mathcal{A}^{3-i} =(A_1^{3-i}, A_2^{3-i}) \gets \textsc{LocalSearch}(A^i_1, A^i_2, v_{3-i})$\;
    \If{$\exists\, j \in \{1,2\}$ such that $v_j(A^j_1)=v_j(A^j_2)$ \textbf{ or } $v_{3-j}(A^j_1)\ge v_{3-j}(A^j_2)$}{
        \KwRet $\{(1, \mathcal{A}^j)\}$ where agent $3-j$ picks her preferred bundle from $\{A^j_1,A^j_2\}$\;
    }
    \If{$\mathcal{A}^{3-i}$ is also EFX under $v_i$}{
        \If{\blue{$v_{i}(A^i_1) \;<\; v_{i}(A^{3-i}_1)$}}{
            $A^i \gets A^{3-i}$\;
        }
        \KwRet $\{(0.5, \mathcal{A}^1), (0.5, \mathcal{A}^2)\}$ where for each $i \in \{1,2\}$, agent $3-i$ picks her preferred bundle from $\{A^i_1,A^i_2\}$\;
    }
}
\KwRet $\{(0.5, \mathcal{A}^1), (0.5, \mathcal{A}^2)\}$ where for each $i \in \{1,2\}$, agent $3-i$ picks her preferred bundle from $\{A^i_1,A^i_2\}$\;

\end{algorithm}
\end{savenotes}

\end{minipage}

\begin{algorithm}[H]
\caption{\textsc{LocalSearch}$(A, B, v)$}
\label{alg:local-search}
\DontPrintSemicolon

\KwIn{an additive valuation $v$ and an allocation $(A,B)$}
\KwOut{an allocation $(A',B')$ that is EFX for $v$, and such that $v(B')\geq v(A')\geq \min \{v(A),v(B)\}$}

\If{$v(A) > v(B)$}{
    \textbf{Swap} $A$ and $B$\;
}

\While{$\exists\, g \in B$ such that $v(A \cup \{g\}) < v(B)$}{
    $g \gets \arg\max\limits_{g' \in B \;\wedge\; v(A\cup\{g'\}) < v(B)} \; v(g')$\;
    $A \gets A \cup \{g\}$,\quad $B \gets B \setminus \{g\}$\;
    \If{$v(A) > v(B)$}{
        \textbf{Swap} $A$ and $B$\;
    }
}
\KwRet $(A'=A,B'=B)$
\end{algorithm}

We now prove \cref{the:main-2agents}.

\begin{proof}

In \Cref{sec:two-agents-proof}, we prove the following lemma:

\begin{lemma}[Minimal-Value Improvement]\label{lem:min_value-lemma}
Let $v_1,v_2$ be additive valuations and let $\mathcal{X}^i=\{X^i_1,X^i_2\}$ be the given initial
partitions for $i\in\{1,2\}$.
Then the output of \cref{alg:two-agents-efx} contains at most two allocations, and for every realized
allocation $\mathcal{A}=(A_1,A_2)$ in the support, and every agent $i\in\{1,2\}$,
\[
v_i(A_i)\;\ge\; \min\{v_i(X^i_1),\,v_i(X^i_2)\}.
\]
\end{lemma}

In addition, all claims proven in \cite{bu2024bestofbothworldsfairallocationindivisible} remain valid for this variant of the algorithm. In particular (as shown in \Cref{sec:two-agents-proof}):

\begin{theorem}[Essentially Theorem~3.2 of \cite{bu2024bestofbothworldsfairallocationindivisible}]
\label{thm:alg3-ef-efx}
In the indivisible-goods setting with two agents and additive utilities,
Algorithm~\ref{alg:two-agents-efx} runs in polynomial time and returns an allocation
that is ex-ante envy-free and ex-post EFX.
\end{theorem}

Combining \cref{thm:alg3-ef-efx} and \cref{lem:min_value-lemma} finishes the proof.

\end{proof}
 
\subsection{Guarantees of Theorem~3.2 of \cite{bu2024bestofbothworldsfairallocationindivisible} that are Preserved}\label{sec:two-agents-proof}

We now prove \cref{thm:alg3-ef-efx}, showing that Theorem~3.2 of \cite{bu2024bestofbothworldsfairallocationindivisible} holds also for our variant of the algorithm.

First, we notice that subroutine \textsc{LocalSearch} (\cref{alg:local-search}) is unchanged; hence Lemma~3.1 of
\cite{bu2024bestofbothworldsfairallocationindivisible} continues to hold:

\begin{lemma}[Lemma~3.1 of \cite{bu2024bestofbothworldsfairallocationindivisible}]
\label{lem:locals-efx} Let $v$ be an additive valuation function.  The subroutine
\textsc{LocalSearch} 
runs in polynomial time, and given bundles $A,B$ and a valuation $v$, returns an EFX allocation
$(A',B')$ under $v$ such that
$
\bigl|v(B')-v(A')\bigr| \;\le\; \bigl|v(B)-v(A)\bigr|
$, which is equivalent to $min\{v_i(A'),v_i(B')\} \geq min\{v_i(A),v_i(B)\}$\footnote{Since $v(A')+v(B')=v(A)+v(B)$}.
\end{lemma}

We next notice that in the initialization, \cref{thm:alg3-ef-efx} in \cite{bu2024bestofbothworldsfairallocationindivisible} requires that the two initial allocation are EFX, which our variation guarantees. Therefore the guarantees are upheld.

Since the only return clause influenced by our modification is the return clause in Line 12, the distribution of allocations returned in all other return clauses (lines 4, 8, 13) remains ex-ante EF and ex-post EFX, as guaranteed by \cite{bu2024bestofbothworldsfairallocationindivisible}.

We therefore analyze only our modification in lines 10-11: When $\mathcal{A}^{3-i}$ is EFX under $v_i$, 
the algorithm updates $\mathcal{A}^i \leftarrow \mathcal{A}^{3-i}$
\emph{only if} agent $i$’s minimum strictly improves.
Formally, writing for any two bundles $(X,Y)$ and valuation $v$,

$$min_i(X,Y): =\min\{v_i(X),v_i(Y)\} $$

the update condition is $min_i(A^{3-i}_1,A^{3-i}_2) > min_i(A^i_1,A^i_2)$, which is equivalent
to \[
\bigl|v_i(A^{3-i}_1) - v_i(A^{3-i}_2)\bigr| \;<\; \bigl|v_i(A^{i}_1) - v_i(A^{i}_2)\bigr|.
\]

The algorithm then returns a lottery over
$\{(0.5,\mathcal{A}^1),(0.5,\mathcal{A}^2)\}$.

We notice that if the algorithm chooses to update $A^i$, then the lottery returned is identical to the one returned in the original allocation, and all guarantees hold.
\\
Else, we note:

\begin{enumerate}
    \item In each realization $\mathcal{A}^i$, agent $3-i$ chooses first (hence is ex-post EFX-satisfied), and agent $i$ receives the remaining bundle from her own allocation; by construction this bundle is EFX for $v_i$.
    \item Ex-ante, each agent $i$ chooses from $\mathcal{A}^{3-i}$ with probability $1/2$ and receives the other bundle from $\mathcal{A}^i$ with probability $1/2$. Using $v_i(A^i_1)+v_i(A^i_2)=v_i(A^{3-i}_1)+v_i(A^{3-i}_2)=v_i(M)$, we have
    \begin{multline*}
        \begin{split}
          \mathbb{E}[v_i] & = \tfrac12\bigl(v_i(M)-min_i(A^{3-i}_1,A^{3-i}_2)\bigr) \;+\; \tfrac12\,min_i(A^i_1,A^i_2) \\
      & = \frac{v_i(M)}{2} \;+\; \frac{min_i(A^i_1,A^i_2)-min_i(A^{3-i}_1,A^{3-i}_2)}{2}
      \\ &\ge \frac{v_i(M)}{2}.
        \end{split}
    \end{multline*}
    The final inequality uses $min_i(A^{3-i}_1,A^{3-i}_2)\le min_i(A^i_1,A^i_2)$ (the condition under which we do not meet the \emph{if clause} requirements). Thus, each agent attains at least her proportional share ex-ante.
\end{enumerate}

\subsection{Proof of \cref{lem:min_value-lemma}}

\begin{proof}
    Fix an agent $i\in\{1,2\}$. 

\emph{Initialization.} The algorithm sets $\mathcal{A}^i\gets \textsc{LocalSearch}(X^i_1,X^i_2,v_i)$. 
By \cref{lem:locals-efx}, 
for $v_i$, the lowest value bundle in $\mathcal{A}^i$ has value at least 
$\min\{v_i(X^i_1),v_i(X^i_2)\}$. 

\emph{Updates to $\mathcal{A}^i$.} 
The only way agent $i$'s candidate can change later is when she either runs \textsc{LocalSearch} on $\mathcal{A}^{3-i}$ under valuation $v_i$ (Line 6), or when she  copies the other partition $\mathcal{A}^{3-i}$ (Line 11).
By the update rule, this happens only if the least valued bundle in $\mathcal{A}^{3-i}$ in agent $i$'s opinion is \emph{valued strictly more} than that in $\mathcal{A}^i$. 
Since \cref{lem:locals-efx} ensures that the worse-bundle value cannot decrease under \textsc{LocalSearch}, every update to $\mathcal{A}^i$ strictly increases the value of agent~$i$'s minimum-bundle value.
Therefore, throughout the algorithm 
$$\min \{v_i(A_1^i),v_i(A_2^i)\} \ge \min\{v_i(X^i_1),v_i(X^i_2)\}.$$

\emph{Termination.} The algorithm returns either:
\begin{enumerate}
\item a single allocation $\mathcal{A}$ in which every agent receives her \emph{top valued} bundle (lines 5 and 10), therefore a value greater than her proportional share, and therefore more than $ \min\{v_i(X^i_1),v_i(X^i_2)\}$, or
\item a lottery over $\mathcal{A}^1$ and $\mathcal{A}^2$ in which agent $i$ chooses her top valued bundle in $\mathcal{A}^{3-i}$, and receives (in the worst case) her least favorite bundle in $\mathcal{A}^i$. In realization $\mathcal{A}^i$, agent $i$ may receive her \emph{worse} bundle; by the discussion above, its value is at least $\min\{v_i(X^i_1),v_i(X^i_2)\}$. 
\end{enumerate}
In both cases, agent $i$'s realized value is at least $\min\{v_i(X^i_1),v_i(X^i_2)\}$.
Since the support consists of at most the two allocations $\mathcal{A}^1$ and $\mathcal{A}^2$, the claim follows.

\end{proof}

\subsection{Algorithm 3 of \cite{bu2024bestofbothworldsfairallocationindivisible}}\label{alg:3}

We now give the original algorithm of \cite{bu2024bestofbothworldsfairallocationindivisible}, with our modifications highlighted in \textcolor{blue}{blue}.
For clarity, comments beside each modification indicate the corresponding original content from the algorithm.

\begin{algorithm}[H]
\caption{An ex-ante EF and ex-post EFX randomized allocation for two additive agents, with minimum value improvement} 
\label{alg:two-agents-efx-orig}
\DontPrintSemicolon

\KwIn{Two additive valuations $u_1,u_2$ over a set of indivisible goods $M$. 
\textbf{\blue{Two initial partitions of $M$, $X^1=(X^1_1,X^1_2)$ and $X^2=(X^2_1,X^2_2)$}} \tcp*[f]{{\color{gray}Original: No initial partitions}}}
\KwOut{A uniform distribution over (at most two) partitions, that is ex-ante EF. Additionally, for each partition, each agent $i$ is EFX-satisfied and receives a value of at least $\min \{v_i(X^i_1), v_i(X^i_2)\}$.}

\For{$i \in \{1,2\}$}{
    $A^i \gets \textsc{LocalSearch}(\blue{\mathbf{X^i_1, X^i_2}}, u_i)$ \tcp*[r]{{\color{gray}Original: $A^i \gets \textsc{LocalSearch}(M,\emptyset, u_i)$}}
}

\If{$\exists\, i \in \{1,2\}$ such that $u_i(A^i_1)=u_i(A^i_2)$ \textbf{ or } $u_{3-i}(A^i_1)\ge u_{3-i}(A^i_2)$}{
    \KwRet $\{(1, A^i)\}$ where agent $3-i$ picks her preferred bundle first\;
}

\While{$\exists\, i \in \{1,2\}$ such that $u_{3-i}(A^i_2)-u_{3-i}(A^i_1) \;<\; u_{3-i}(A^{3-i}_2)-u_{3-i}(A^{3-i}_1)$}{
    $A^{3-i} \gets \textsc{LocalSearch}(A^i_1, A^i_2, u_{3-i})$\;
    \If{$\exists\, j \in \{1,2\}$ such that $u_j(A^j_1)=u_j(A^j_2)$ \textbf{ or } $u_{3-j}(A^j_1)\ge u_{3-j}(A^j_2)$}{
        \KwRet $\{(1, A^j)\}$ where agent $3-j$ picks her preferred bundle first\;
    }
    \If{$A^{3-i}$ is also EFX under $u_i$}{
        \If{\blue{$\mathbf{u_{i}(A^i_1) \;<\; u_{i}(A^{3-i}_1)}$}}{
            $A^i \gets A^{3-i}$ \tcp*[r]{{\color{gray}Original: if True:}}
        }
        \KwRet $\{(0.5, A^1), (0.5, A^2)\}$ where agent $3-i$ picks first in realized $A^i$ for each $i \in \{1,2\}$\;
    }
}
\KwRet $\{(0.5, A^1), (0.5, A^2)\}$ where agent $3-i$ picks first in realized $A^i$ for each $i \in \{1,2\}$\;

\end{algorithm}

\subsection{The 2-Agents {\textbf{PTAS$_{ALG}$} 
Algorithm (\cite{babaioff2022fairshares})}\label{sub:PTAS}}

For a fixed $\varepsilon>0$, 
the algorithm $PTAS_{ALG}$ introduced by \cite{babaioff2022fairshares} receives an additive valuation $v_i$ and $n$, and computes, in polynomial time, an allocation with the worse bundle having a value at least $(1-\varepsilon)$-fraction of the MMS with respect to $v_i$ and $n$, as follows:
Let $k = \left\lceil \frac{3}{2\varepsilon}\right\rceil$.

Sort all items in non-increasing order of values, and let $e_m$ denote the last item in this order. 
Partition the set $M$ into a set $Z$ of the $k$ first items in the order (the most valuable items), and a set of remaining items $S=M\setminus Z$. Consider the following family $F_{m,2}$ of at most $2^km$ partitions of $M$. A partition of $M$ into two bundles $(B_1,B_2)$ is a member of $F_{m,2}$ if and only if $B_1\cap S$ is a suffix of $S$. That is, bundle $B_1$ contains an arbitrary subset of $Z$, but the set of items that $B_1$ contains from $S$ must be a (possibly empty) set of consecutive items, ending in $e_m$. The value of the $PTAS$ share for agent $i$ is $PTAS(\varepsilon)_i = max_{(B_1,B_2)\in F_{m,2}} min[v_i(B_1),v_i(B_2)]$. That is, the definition of $PTAS(\varepsilon)_i$ is similar to $MMS_i$, except that instead of maximizing the minimum value of a bundle over all possible partitions, one maximizes only over the partitions in $F_{m,2}$.

\begin{theorem}\label{lem:min-ptas}
In the two–agent additive setting, fix \(\varepsilon>0\) and let \(k=\lceil 3/(2\varepsilon)\rceil\).
The share guarantees \(\mathrm{PTAS}(\varepsilon)_i\) is computable in time polynomial in \(m\) (for fixed \(\varepsilon\)) and satisfies
\[
\mathrm{PTAS}(\varepsilon)_i \;\ge\; 
\min\!\left\{\mathrm{MMS}_i,\; \mathrm{PROP}_i\cdot \left(1-\frac{1}{\bigl(k+1\bigr)}\right)\right\}\geq
\min\!\left\{\mathrm{MMS}_i,\; \mathrm{PROP}_i\cdot \left(1-\varepsilon\right)\right\}
\]
Moreover, an allocation guaranteeing each agent a value of $PTAS(\varepsilon)_i$ can be calculated in polynomial time. 
\end{theorem}

\begin{proof}
If \(m<3/\varepsilon\), an \(\mathrm{MMS}\) allocation can be found in polynomial time, so \(\mathrm{PTAS}(\varepsilon)_i\ge \mathrm{MMS}_i\) and the bound holds. Hence assume \(m\ge 3/\varepsilon\), so \(k\le m\).

By \cite{babaioff2022fairshares}, either the maximizer over \(F_{m,2}\) is an \(\mathrm{MMS}\) partition (yielding \(\mathrm{PTAS}(\varepsilon)_i\ge \mathrm{MMS}_i\)), or there exists \((B_1,B_2)\in F_{m,2}\) such that
\[
\min\{v_i(B_1),v_i(B_2)\}\;\ge\; \frac{v_i(M)}{2}\;-\;\frac{1}{2}\max_{e\in S} v_i(e).
\]
Since the items are sorted nonincreasingly and \(S\) excludes the top \(k\) items, every \(e\in S\) has
$$v_i(e)\le \frac{v_i(Z)+v_i(e)}{k+1}\leq \frac{v_i(M)}{k+1}$$ 
Thus
\[
\min\{v_i(B_1),v_i(B_2)\}
\;\ge\;
\frac{v_i(M)}{2}\;-\;\frac{1}{2}\cdot \frac{v_i(M)}{k+1}
\;=\;
\left(1-\frac{1}{(k+1)}\right)\cdot \mathrm{PROP}_i 
\]
Since $k=\lceil 3/(2\varepsilon)\rceil$, we have 
$$ 1-\frac{1}{k+1}=1-\frac{1}{\lceil 3/(2\varepsilon)\rceil+1} \geq 1-\frac{1}{ 3/(2\varepsilon)}\geq 1-\frac{2\varepsilon}{3}\geq 1-\varepsilon$$

Maximizing over \(F_{m,2}\) gives \(\mathrm{PTAS}(\varepsilon)_i\ge \min\{\mathrm{MMS}_i,\; (1-\varepsilon) \cdot \mathrm{PROP}_i\}\).
Since \(k=\lceil 3/(2\varepsilon)\rceil\) is a constant for fixed \(\varepsilon\), enumerating \(F_{m,2}\) and evaluating \(\min\{v_i(B_1),v_i(B_2)\}\) for each partition takes time polynomial in \(m\).
\end{proof}

\subsection{Prior 2-Agent BoBW Poly-Time Algorithms Lose a Constant Fraction of the MMS} 

\label{sec:no-min-bound}

Previously, Garg and Sharma \cite{garg2024bestofbothworldsfairnessenvycycleeliminationalgorithm} employed the envy-cycle elimination routine 
of \cite{lipton2004approximately} to obtain ex-ante EF together with ex-post EFX, in polynomial time. They do this by maintaining two allocations throughout the Envy-Cycle-Elimination algorithm, $\mathcal{X}$ and $\mathcal{Y}$, which are both initially empty. The algorithm proceeds in rounds. In each round, they first resolve envy cycles in both $\mathcal{X}$ and $\mathcal{Y}$. Next, they find an unenvied agent $i_X$ in $\mathcal{X}$ and an unenvied agent $i_Y$ in $\mathcal{Y}$ such that $i_X\neq i_Y$ (these agents always exists -- as proven in \cite{garg2024bestofbothworldsfairnessenvycycleeliminationalgorithm}).
They then allocate a good to $i_X$ in $X$ and the same good to $i_Y$ in $Y$. Finally, they choose one of the allocations at random, thereby guaranteeing ex-ante EF. 

Using a different approach,\citet{bu2024bestofbothworldsfairallocationindivisible} gave a polynomial-time algorithm that guarantees ex-ante EF and ex-post EFX by invoking a local-search subroutine \textsc{LocalSearch}. Given an allocation \((A,B)\) and a valuation \(v\), \textsc{LocalSearch} returns an EFX allocation (with regards to $v$) and weakly decreases the bundle-value gap \(|v(A)-v(B)|\) relative to the input. The algorithm outputs a distribution supported on at most two allocations. Each of these allocations guarantees \emph{EFX} for both agents, while the distribution as a whole ensures that every agent i receives an expected value of at least her proportional share $PROP_i$.

In this section we observe that for small enough $\varepsilon>0$ (say, $\varepsilon<0.11$), each of the two algorithms -- \textsc{ECEG2} of
\cite{garg2024bestofbothworldsfairnessenvycycleeliminationalgorithm} and
Algorithm~3 of \cite{bu2024bestofbothworldsfairallocationindivisible}, discussed above -- 
does not guarantee that every agent gets at least $(1-\varepsilon)$ fraction of the MMS.

Equivalently, for each of these
procedures there exists a constant $\alpha<0.89<1$ and an instance such that some
agent $i$ receives utility at most $\alpha\cdot \mathrm{MMS}_i$, so the fraction of MMS lost is at least 11\%. 

\subsubsection{Algorithm~3 of \cite{bu2024bestofbothworldsfairallocationindivisible} Loses a Constant Fraction of the MMS}
\label{subsec:worse-bu}

We observe that Algorithm 3 of \cite{bu2024bestofbothworldsfairallocationindivisible} can lose a constant fraction of the MMS, even in a simple case with 4 items and identical valuations. 

\begin{lemma}[A {$0.85$–MMS} upper-bound instance for Algorithm~3 of \cite{bu2024bestofbothworldsfairallocationindivisible}]\label{lem:integer-approx}
There exists a setting with two agents with identical additive valuations over a set of $4$ goods, such that Algorithm 3 of \cite{bu2024bestofbothworldsfairallocationindivisible} gives one of the agents a fraction of only {$17/20=0.85$} of her MMS.
\end{lemma}
\begin{proof}
Consider two agents with identical additive valuations ($v=v_1=v_2$) over $M=\{g_1,g_2,g_3,g_4\}$, given by
\[
\begin{array}{c|cccc}
 & g_1 & g_2 & g_3 & g_4 \\
 \hline
 v(g) & 16 & 12 & 8 & 5 
\end{array}
\]
To prove the claim we will first show that $\mathrm{MMS}=20$, and also show that Algorithm~3 (run as in \cite{bu2024bestofbothworldsfairallocationindivisible})
gives a random agent the set $\{g_1,g_3\}$ and the other agent the set $\{g_2,g_4\}$. 
The agent that ends up with the bundle $\{g_2,g_4\}$ gets a value of only $v(g_2)+v(g_4)=12+5=17$, thus only $17/20=0.85$ of her MMS.

\textbf{MMS computation.}
The total value is $v(M)=41$. The partition
$\{g_1,g_4\}$,\ $\{g_2,g_3\}$ gives a minimum bundle value of $20$, and since $PROP = 20.5$, this means that this is an MMS allocation (since $PROP \geq MMS$, and all the items are valued as integers). 
Therefore $\mathrm{MMS}=20$, attained by $\bigl(\{g_1,g_4\},\{g_2,g_3\}\bigr)$.

\smallskip
\textbf{Behavior of \textsc{LocalSearch}.}
Run \textsc{LocalSearch}$(\emptyset,M,v)$:
(i) move item $g_1$ (value $16$) to $A$ since $0+16<41$; state $(A,B)=(\{g_1\},\{g_2,g_3,g_4\})$ with values $(16,25)$;
(ii) item $g_2$ would make $16+12=28\not<25$, so skip; item $g_3$ makes $16+8=24<25$, move it:
$(A,B)=(\{g_1,g_3\},\{g_2,g_4\})$ with values $(24,17)$; (iii) swap to maintain $v(A)\le v(B)$, yielding
$(A,B)=(\{g_2,g_4\},\{g_1,g_3\})$ with values $(17,24)$; (iv) no $g\in B$ satisfies $v(A\cup\{g\})<v(B)$,
so the subroutine returns $\bigl(\{g_2,g_4\},\{g_1,g_3\}\bigr)$. Since both agents have the same valuation,
the two calls in Algorithm~3 (one per agent) produce identical pairs $A^1=A^2=\bigl(\{g_2,g_4\},\{g_1,g_3\}\bigr)$.

\smallskip
\textbf{Algorithm~3’s return.}
The early-return condition (equal values for some $i$ or the other agent weakly prefers the smaller bundle)
does not trigger: $17\neq 24$ and $17 \not\ge 24$. The while-condition compares the two gaps, which are equal,
so the loop does not run. Hence Algorithm~3 returns the $50$–$50$ lottery over $A^1$ and $A^2$, which are the same
bipartition $\bigl(\{g_2,g_4\},\{g_1,g_3\}\bigr)$. In each realization, the first picker takes $\{g_1,g_3\}$ (value $24$),
and the other agent receives $\{g_2,g_4\}$ (value $17$). Since $\mathrm{MMS}=20$, the guaranteed (worst-case) value
is $0.85\cdot \mathrm{MMS}$, as claimed.
\end{proof}
\subsubsection{\textsc{ECEG2} of \cite{garg2024bestofbothworldsfairnessenvycycleeliminationalgorithm} Loses a Constant Fraction of the MMS}
\label{subsec:worse-eceg2}

We next show that the ECEG2 algorithm of \cite{garg2024bestofbothworldsfairnessenvycycleeliminationalgorithm} can lose a constant fraction of the MMS. 

\begin{lemma}[A $0.881$–MMS upper-bound instance for \textsc{ECEG2}]
\label{lem:eceg2-integer-upper-bound}
There exists a setting with two agents with additive valuations over a set of $7$ goods, such that algorithm \textsc{ECEG2} sometimes gives one of the agents a fraction of only $37/42$ of her MMS (and $37/42 \approx 0.881 < 0.89$).
\end{lemma}
\begin{proof}

Consider two agents with additive valuations over $M=\{g_1,\dots,g_7\}$ given by
\[
\begin{array}{c|ccccccc}
 & g_1 & g_2 & g_3 & g_4 & g_5 & g_6 & g_7 \\
 \hline
 v_1(g) & 15 & 14 & 13 & 12 & 10 & 10 & 10 \\[2pt]
 v_2(g) & 13 & 12 & 9 & 6 & 5 & 3 & 3
\end{array}
\]
Let $B=(B_1,B_2)$ denote the \emph{second} allocation maintained by \textsc{ECEG2}.
Then $\mathrm{MMS}_1=42$, while
\[
B_1=\{g_2,g_3,g_6\}
\quad\text{and}\quad
v_1(B_1)=14+13+10=37,
\]
so
\(
v_1(B_1)/\mathrm{MMS}_1 = 37/42 \approx 0.881 < 0.9.
\)

\textbf{Simulation of \textsc{ECEG2}.}
Run \textsc{ECEG2} on the ordered instance $(v_1,v_2)$, processing goods
$g=g_1,g_2,\dots,g_7$ in the displayed order. The algorithm maintains two partial allocations
$A=(A_1,A_2)$ and $B=(B_1,B_2)$ and applies envy-cycle eliminations (which never trigger in this two-agent run).

For $C\in\{A,B\}$ and any round, define the set of \emph{unenvied} agents
\[
S_C \;:=\;\bigl\{\, i\in\{1,2\} : v_{3-i}(C_{3-i}) \ge v_{3-i}(C_i) \,\bigr\}.
\]
(In words, $i\in S_C$ iff agent $3-i$ does not strictly prefer $C_i$ to her own bundle $C_{3-i}$.)
In the two-agent specialization used here, the incoming good $g$ is assigned as follows:
if $1\in S_A$ and $2\in S_B$, place $g$ into $(A_1,B_2)$; otherwise place $g$ into $(A_2,B_1)$.
Under these choices, no envy-cycle swap is triggered.

\smallskip
\noindent\textbf{Round-by-round evolution of $B$.}
\[
\begin{array}{c|c|c|c}
\text{good } g & S_A & S_B & \text{update to } B
\\\hline
g_1 & \{1,2\} & \{1,2\} & B_2 \gets \{g_1\} \\[2pt]
g_2 & \{2\}   & \{1\}   & B_1 \gets \{g_2\} \\[2pt]
g_3 & \{2\} & \{1\}   & B_1 \gets \{g_2,g_3\} \\[2pt]
g_4 & \{1\}   & \{2\}   & B_2 \gets \{g_1,g_4\} \\[2pt]
g_5 & \{1,2\} & \{2\}   & B_2 \gets \{g_1,g_4,g_5\} \\[2pt]
g_6 & \{2\} & \{1\}   & B_1 \gets \{g_2,g_3,g_6\} \\[2pt]
g_7 & \{1,2\} & \{1,2\} & B_2 \gets \{g_1,g_4,g_5,g_7\}
\end{array}
\]
Thus, at termination,
\(
B_1=\{g_2,g_3,g_6\}
\)
and
\(
B_2=\{g_1,g_4,g_5,g_7\}.
\)

\smallskip
\noindent\textbf{Computing $\mathrm{MMS}_1$.}
We have $v_1(M)=84$. Let $X=\{g_1,g_2,g_3\}$; then
\[
v_1(X)=15+14+13=42
\quad\text{and}\quad
v_1(M\setminus X)=12+10+10+10=42,
\]
so $\mathrm{MMS}_1=42$.

\smallskip
\noindent\textbf{Conclusion.}
Agent~1’s value for $B_1$ is $v_1(B_1)=14+13+10=37$, and therefore if allocation $B$ is chosen, agent 1 receives only 
\[
\frac{v_1(B_1)}{\mathrm{MMS}_1} \;=\; \frac{37}{42} \;\approx\; 0.881 \;<\; 0.89,
\]
as claimed.
\end{proof}

\section{Omitted Proofs for Three Agents - Polynomial Case}\label{sec:3agents_poly_proof}

In this section, we prove \cref{lem:stage-3-guarantees}, \cref{lem:copy-stage-guarantees} and \cref{lem:two-copy-prop}.

\begin{lemma}[Ex-Post Guarantees at the end of Stage~3]\label{lem:stage-3-guarantees}
All allocations constructed at the end of Stage~3 satisfy the ex-post guarantees established in \cref{thm:bobw-poly-approx}, except the EEFX guarantee.
Specifically, for every agent $t \in \{1,2,3\}$, every allocation guarantees at least $(\frac{9}{10}-\varepsilon)\text{-MMS}_t$, while guaranteeing \EIME.
\end{lemma}

\begin{proof}
Allocations constructed in Stage~2 are obtained directly from the algorithm described in \cref{thm:bobw-poly-approx}, and therefore each of them satisfies all these ex-post guarantees.

We now turn to allocations \textbf{updated} during Stage~3. 
Recall that such an update occurs when in case 2.B, for some given allocations $\mathcal{X}^i$ and $\mathcal{Y}^i$, there exists an agent~$r$ (the \textbf{subdivider}) in $\mathcal{X}^i$ and an agent~$l$ (the \textbf{subdivider}) in $\mathcal{Y}^i$ such that
\[
\min_{Z\in \{X_r^i, X_l^i\}} v_r(Z)
    < 
\min_{Z\in \{Y_r^i, Y_l^i\}} v_r(Z).
\]
Intuitively, agent~$r$ prefers the minimal-valued bundle in the two-bundle partition $\{ Y_r^i, Y_l^i\}$ created when agent~$l$ acted as the subdivider. 

We recall from \cref{obs:repartition-larger} that, for agent~\(l\), the bundles \(Y^i_r\) and \(Y^i_l\) are the highest-valued bundles in \(\mathcal{Y}^i\).  
We claim that the same holds for agent~\(r\): since agent~\(r\) selected \(Y^i_r\) in \(\mathcal{Y}^i\), this bundle must be her top-valued one in that partition.  
Moreover, since the pair \(\{X_r^i, X_l^i\}\) is produced by the \textbf{MMS-EFX-Improved-Repartition} procedure and maximizes the algorithm’s output on either \((A,B)\) or \((A,C)\), it follows by definition that
\[
\min_{Z \in \{X_r^i, X_l^i\}} v_r(Z) \ge v_r(Y^i_i),
\]
as \(Y^i_i\) corresponds to either \(B\) or \(C\).  
Hence,
\[
\min_{Z \in \{Y_r^i, Y_l^i\}} v_r(Z) > \min_{Z \in \{X_r^i, X_l^i\}} v_r(Z) \ge v_r(Y^i_i).
\]
Therefore, in the partition \(\mathcal{Y}^i\), both agents~\(l\) and~\(r\) regard the bundles \(\{Y_r^i, Y_l^i\}\) as their top-valued bundles.

We invoke the two-agent procedure (\cref{alg:two-agents-efx}) with the following input:
\begin{enumerate}
    \item $v_1 = v_r, v_2=v_l$
    \item $\mathcal{X}^1=\mathcal{X}^2=(Y^i_r,Y^i_l)$
\end{enumerate}

We denote by $T_i = Y^i_i=M\setminus (Y_r^i \cup Y_l^i)$  the bundle allocated to agent~$i$ in the new allocation. 
We recall that \cref{alg:two-agents-efx} outputs two partitions $\{A_1^l,A_2^l\}, \{A_1^r,A_2^r\}$ {over the set of items $\overline{T}_i=Y^i_r \cup Y^i_l$} (it is possible that the two partitions are identical), where agent $l$ chooses a bundle first from $\mathcal{A}^r$, and vice versa.

We construct $\mathcal{X}^i$ and $\mathcal{Y}^i$ using the two resulting allocations, as follows:
\begin{enumerate}
    \item $\mathcal{X}^i$: Agent $l$ chooses her favorite bundle in $\{A_1^r,A_2^r\}$, agent $r$ receives the remaining bundle in $\{A_1^r,A_2^r\}$, and agent $i$ receives $T_i$. We denote $\mathcal{C}^i_l=\mathcal{X}^i$.
     \item $\mathcal{Y}^i$: Agent $r$ chooses her favorite bundle in $\{A_1^l,A_2^l\}$, agent $l$ receives the remaining bundle in $\{A_1^l,A_2^l\}$, and agent $i$ receives $T_i$. We denote $\mathcal{C}^i_r=\mathcal{Y}^i$.
\end{enumerate}
We verify that all fairness guarantees are preserved:
First, observe that when agents~\(l\) and~\(r\) choose first, they each receive their top valued bundle in the partition, and are therefore EFX-satisfied and receive above their proportional share. 

Second, Agent~$i$ receives the same bundle $T_i$ in both $\mathcal{X}^i$ and $\mathcal{Y}^i$.  
    Since $T_i$ is the bundle she received in the \textbf{original} allocation $\mathcal{Y}^i$, she remains both MMS and EEFX-satisfied.

We therefore only need to prove for the cases where agents~\(l\) and~\(r\) choose second:
\begin{enumerate}
    \item \textbf{Agent $l$.}  
    In the original allocation $\mathcal{Y}^i$, agent~$l$ was guaranteed the bundle 
    $argmin\{v_l(Y_r^i), v_l(Y_l^i)\}$, and received at least $\frac{9}{10}-\varepsilon$ of her MMS (by \cref{thm:bobw-poly-approx}). 
    By \cref{the:main-2agents}, in each allocation constructed from the output of the two-agent procedure, she receives at least the same value and is EFX-satisfied with respect to both bundles.
    Moreover, by \cref{obs:repartition-larger}, $\min\{v_l(Y_r^i), v_l(Y_l^i)\}$ is strictly larger than $v_l(T_i)=v_l(Y^i_i)$. Therefore she is EFX-satisfied with $Y^i_l$.
    
    \item \textbf{Agent $r$.}  
    The two-agent procedure is performed only if agent~$r$ prefers $\min\{v_r(Y_r^i), v_r(Y_l^i)\}$ to $\min_{Z\in\{X_r^i, X_l^i\}}v_r(Z)$.  
    Therefore, by \cref{the:main-2agents}, her guarantees can only improve compared to $(X_r^i, X_l^i)$ (where she was guaranteed $\frac{9}{10}-\varepsilon$ of her MMS share), and as previously mentioned - $\min\{v_r(Y_r^i), v_r(Y_l^i)\}$ is valued the same or higher than $v_r(T_i)$, and therefore she is EFX-satisfied with $X^i_r$. 
\end{enumerate}

Therefore, all ex-post guarantees are upheld as in \cref{thm:bobw-poly-approx}.
\end{proof}

\begin{observation}[Adoption creates no new partitions]\label{obs:adoption-no-new}
For an allocation or certificate $\mathcal{Z}$, let $\pi(\mathcal{Z})$ denote its
underlying unordered partition of $M$. When an agent $k$ adopts a partition
$\mathcal{S}^k$ in Stage~4 or Stage~5, her new allocations
$\mathcal{X}^k, \mathcal{Y}^k$ and the two updated certificates
$C^k_j, C^k_i$ are all permutations of $\mathcal{S}^k$; no bundle is modified.
Hence the set of unordered partitions underlying the allocations and
certificates,
\[
  \Pi \;:=\; \bigl\{\, \pi(\mathcal{X}^t),\, \pi(\mathcal{Y}^t),\, \pi(C^r_t)
  \;:\; t, r \in [3] \,\bigr\},
\]
does not grow during Stages~4--5: every partition available for adoption in
Stage~5 already belongs to the set $\Pi_3$ of partitions present at the end of
Stage~3.

Consequently:
\begin{enumerate}
  \item any partition adopted in Stage~5 was itself adopted (by another agent)
        in Stage~4, so the guarantees established for Stage-4 adoptions apply
        to it;
  \item after Stage~5, every agent $k$ satisfies
        $\mathrm{val}_k \ge \min_{Z \in \mathcal{T}} v_k(Z)$ for every
        partition $\mathcal{T} \in \Pi_3$ appearing among her certificates, so
        a further adoption round would trigger no update -- two rounds
        suffice; and
  \item since adoption only permutes members of $\Pi_3$, all value guarantees
        are independent of the order in which agents are processed within a
        stage.
\end{enumerate}
\end{observation}

\begin{lemma}\label{lem:copy-stage-guarantees}

The allocations constructed by the algorithm presented in \cref{thm:bobw-poly} uphold all ex-post fairness guarantees 
specified in the theorem. 

Specifically, the \emph{divider} (i.e., the adopting agent) receives at least a \((1 - \varepsilon)\)-fraction of her \(MMS\) value,  
the \emph{first-chooser} receives her top-valued bundle in each allocation,  
and the \emph{second-chooser} receives at least \(\tfrac{9}{10}-\varepsilon\) of her MMS value and is \emph{EFX}-satisfied or receives \(1-\varepsilon\) of her MMS value.
\end{lemma}

\begin{proof}
First, observe that {by \cref{lem:stage-3-guarantees}}, all allocations constructed at the end of Stage~3 satisfy all previously established ex-post fairness guarantees.
It therefore remains to verify that the allocations produced in Stages~4 and~5 also uphold these ex-post fairness guarantees.
\begin{enumerate}
\item \textbf{Stage 4: }
We first note that, due to Stage~1, agent $k$ will not choose to adopt any candidate partition which is equal to $\mathcal{M}^t$ for an agent $t \neq k$.

Consequently, agent~$k$ never adopts a candidate partition constructed under Case~1 (i.e., when both she and the other agent receive their top two bundles in such a partition), nor a candidate partition from Subcase~2.A in which she serves as the subdivider (where she receives her top bundle in such a partition). 
Therefore, she can only adopt a candidate partition that was created either when she acted as the chooser in Subcase~2.A, or when she acted a chooser in Subcase~2.B (The adopted partition must be a certificate created for agent $k$. By definition of the certificates - the certificate equals the partition in which the agent was the \textit{chooser}).
 Let the adopted partition be denoted by \(\mathcal{S} = \{S_1, S_2, S_3\}\).  
Let agent \(k\) be the \emph{adopter} (who, as we noted, was the \emph{chooser} in \(\mathcal{S}\)),  
agent \(j\) the original \emph{subdivider} in $\mathcal{S}$, and agent \(i\) the original \emph{divider} in $\mathcal{S}$.

The adopter's (agent \(k\)) value for each of the three bundles is at least  \((1 - \varepsilon) \cdot MMS\), 

since the minimum-valued bundle in $\mathcal{S}$ has value at least her \((1 - \varepsilon)MMS\) threshold.  
Therefore, she may receive any bundle in the partition.  

We recall that we construct two allocations:  
\(\mathcal{X}^k\), where agent \(i\) chooses first and agent \(j\) chooses second,  
and \(\mathcal{Y}^k\), where the roles are reversed. 
We therefore show that agent $j$ and $i's$  top two valued bundles in $S$  uphold all guarantees.

Agent \(j\), who acted as the \emph{subdivider}, 
values her top two bundles in \(\mathcal{S}\) at least as high as \(\tfrac{9}{10}MMS_j\) and {is EFX-satisfied with both of them}. 

The remaining agent \(i\), who served as the \emph{divider} in \(\mathcal{S}\), values at least two bundles in \(\mathcal{S}\) at least as high as \((1 - \varepsilon)MMS_i\):
\begin{enumerate}
    \item the bundle she receives in \(\mathcal{S}\) coincides with one of the bundles from her original division and therefore preserves her value guarantee.  
    \item the combined value of the two remaining bundles in \(\mathcal{S}\) is at least \(2 \cdot (1 - \varepsilon) \cdot MMS_i\), since they are a repartition of two bundles her MMS partition.   
    Therefore, at least one of them must have value at least \((1 - \varepsilon)MMS_i\).  
\end{enumerate}

Therefore, both agents \(i\) and \(j\) possess at least two bundles each that uphold the fairness guarantees.  
We construct two allocations:  
\(\mathcal{X}^k\), where agent \(i\) chooses first and agent \(j\) chooses second,  
and \(\mathcal{Y}^k\), where the roles are reversed.  
In both allocations, the first chooser selects her top-valued bundle, and the second chooser receives at least her second-highest bundle.

\item \textbf{Stage 5: } If an agent decides to perform an adoption in the second round, it is necessarily because she is envious of a candidate partition that was previously \emph{adopted}, by \Cref{obs:adoption-no-new}.  
In this case, the adopting agent is \((1 - \varepsilon)MMS\)-satisfied with any bundle in the partition being adopted.  
Since this partition itself resulted from a prior adoption, there already exists another agent who is also \((1 - \varepsilon)MMS\)-satisfied with all bundles.  
As established in the first round of adoption, the remaining agent possesses at least two bundles that satisfy the fairness guarantees.  

Therefore, each of the two non-adopting agents has at least two bundles that uphold all guarantees.  
We again define two derived allocations:  
\(\mathcal{X}^k\), where one of the non adoption agents chooses first and the other second,  
and \(\mathcal{Y}^k\), where the roles are reversed.

\end{enumerate}
\end{proof}

\begin{lemma}\label{lem:two-copy-prop} The distribution generated by the algorithm in  \cref{thm:bobw-poly} guarantees each agent her proportional share ex-ante.
\end{lemma}

\begin{proof}
Assume, for the sake of contradiction, that there exists an agent \( i \) who does not receive her proportional share ex-ante.  
Let 
\[
\mathrm{val}_i = \min\{v_i(X_i^i),v_i(Y_i^i)\}
\]
denote the value of agent \( i \)'s least-preferred bundle among her own allocations $\mathcal{X}^i$ and $ \mathcal{Y}^i$.  
We also denote for a partition $\mathcal{T}$: $$min_i(\mathcal{T})=\min_{X\in \mathcal{T}}v_i(X)$$
i.e - $min_i(\mathcal{T})$ is the value of the minimal-valued bundle according to agent $i$ in partition $\mathcal{T}$.

We note that $val_i$ is equal or larger than the following values:
\begin{enumerate}
    \item $min_i(\mathcal{M}^t)$ for $t \in \{1,2,3\}$ - as a result of Stage 1. 
    \item $min_i(\mathcal{C}_i^t)$ for $t \neq i $ as a result of Stage 4 + 5 (the adoption stages). 
\end{enumerate}

Let \(\mathcal{X}^j, \mathcal{Y}^j, \mathcal{X}^k, \mathcal{Y}^k\) denote the remaining allocations in the support, where \(j,k \neq i\),  
and we assume w.l.o.g that the allocations in which agent \(i\) acts as the \emph{chooser} are the \(\mathcal{X}\)-allocations.

If, for each \( l \in \{k, j\} \), it holds that
\[
v_i(X_i^l) + v_i(Y_i^l) \;\geq\; v_i(M) - \mathrm{val}_i,
\]
then agent \(i\) is ex-ante proportional.  
Indeed, across all six allocations, her total value is at least
\[
2 \cdot [v_i(M) - \mathrm{val}_i] + 2 \cdot \mathrm{val}_i \;=\; 2 \cdot v_i(M),
\]
and since the final allocation is drawn uniformly at random among these six, her expected value is
\[
\frac{1}{6} \cdot 2 \cdot v_i(M) \;=\; \frac{1}{3} v_i(M),
\]
which equals her proportional share.

Hence, if proportionality is violated, there must exist some \( l \in \{k, j\} \) such that
\begin{equation}
    \label{eq:1}
 \space v_i(X_i^l) + v_i(Y_i^l) \;<\; v_i(M) - \mathrm{val}_i.
\end{equation}
We now examine all possible cases for the construction of 
\(\mathcal{X}^l\) and \(\mathcal{Y}^l\) during the algorithm of \cref{thm:bobw-poly}, and show that each leads to $$v_i(X_i^l) + v_i(Y_i^l) \;\geq \; v_i(M) - \mathrm{val}_i$$  a contradiction of \cref{eq:1}.

\begin{enumerate}
    \item \textbf{Case~1:}  
    The allocation \(\mathcal{X}^l\) was constructed in Case~1.  
    Here, the chooser receives her favorite bundle from \(\mathcal{X}^l\), and we have 
    \(\mathcal{X}^l = \mathcal{Y}^l = \mathcal{M}^l \).  
    Thus, across \(\mathcal{X}^l\) and \(\mathcal{Y}^l\), she receives at least the combined value of her two highest-valued bundles in \(\mathcal{X}^l\).  
    As noted
    \[
    \mathrm{val}_i \geq min_i(\mathcal{M}^l),
    \]
    and therefore
    \[
    v_i(X_i^l) + v_i(Y_i^l) = 2\cdot v_i(X_i^l)\geq v_i(M) -min_i(\mathcal{M}^l) \geq v_i(M) - \mathrm{val}_i,
    \]

    \item \textbf{Subcase~2.A:}  
    In this case, ${X}_i^l$ is agent $i's$ top-valued bundle in \(\mathcal{C}_i^l\), and by construction she either receives the same bundle in ${Y}_i^l$, or a higher valued bundle (in the case where she receives her top valued bundle in $\mathcal{Y}^l$). 
    As previously noted, $val_i \geq min_i(\mathcal{C}_i^l)$, and therefore 
       \[
    v_i(X_i^l) + v_i(Y_i^l) \geq 2\cdot v_i(X_i^l)\geq v_i(M) -min_i(\mathcal{C}_i^l) \geq v_i(M) - \mathrm{val}_i,
    \]

    \item \textbf{Subcase~2.B:}  
    The chooser selects her top bundle from the subdivider’s repartition. We recall that we assume w.l.o.g that $i$ was the chooser in $\mathcal{X}^l$.
    \begin{enumerate}
        \item If the chooser and subdivider partitioned the \emph{same} two bundles in $\mathcal{X}^l$ and $\mathcal{Y}^l$, then as a result of Stage 3 we ensure that, across \(\mathcal{X}^l\) and \(\mathcal{Y}^l\),  
        agent \(i\) receives at least the value of her two top bundles in \(\mathcal{X}^l\).  We recall that $\mathcal{C}_i^l=\mathcal{X}^l$, and therefore \(\mathrm{val}_i \geq min_i(\mathcal{X}^l)\)  (as a result of \(\mathrm{val}_i \geq min_i(\mathcal{C}_i^l)\)). We therefore have:
        \[
        v_i(X_i^l) + v_i(Y_i^l) \geq v_i(M)- min_i(\mathcal{X}^l)\geq v_i(M) - \mathrm{val}_i.
        \]
        \item If they partitioned \emph{different} bundles, then by construction in Subcase~2.B the chooser (agent $i$) receives her top bundle in \(\mathcal{X}^l\). We recall that $\mathcal{C}_i^l=\mathcal{X}^l$.
        We next note that as a result of Stage 3 (agent $i$ is the subdivider in $\mathcal{Y}^l$): 
        $$min \{v_i(Y_i^l),v_i(Y_r^l)\} \geq min \{v_i(X_i^l),v_i(X_r^l)\}$$
        (If this is upheld before Stage 3, we do not run the two agent procedure. If it is not upheld - we run the two agent procedure, which guarantees it).
                We also note (as in previous sections) that since $X_l^l$ is either $B$ or $C$,
                we are guaranteed that
                 $$min \{v_i(Y_i^l),v_i(Y_r^l)\} \geq v_i(X_l^l)$$
        Combining the two inequalities above with the fact that $X_i^l$ is agent $i's$ top valued bundle in $\mathcal{X}^l$:        
        \[
        v_i(Y_i^l) \geq \max \{v_i(X_r^l), v_i(X_l^l)\}.
        \] 

        Thus,
        \[
        v_i(X_i^l) + v_i(Y_i^l) 
        \geq v_i(X_i^l) + \max\{v_i(X_r^l), v_i(X_l^l)\}
         =  v_i(M) - min_i(\mathcal{C}_i^l)
        \geq v_i(M) - \mathrm{val}_i.
        \]

    \end{enumerate}

    \item \textbf{Adoption Stages (Stage 4 + 5):}  
    The partition \(\mathcal{X}^l\) was constructed during one of the adoption stages (Stage 4 or Stage 5).
    \begin{enumerate}
        \item If agent $l$ adopted \(\mathcal{X}^l\) in the \emph{first} adoption stage (Stage 4), then 
        agent \(i\) could have adopted this partition in the \emph{second} stage (Stage 5), achieving the same valued minimal bundle as in \(\mathcal{X}^l\).  
        Hence,
        \[
        val_i \geq \min_{X \in \mathcal{X}^l} v_i(X).
        \]
        By construction (since \(\mathcal{X}^l\) is an adopted partition), agent \(i\) receives in \(\mathcal{X}^l\) and \(\mathcal{Y}^l\) the sum of her top two bundles in \(\mathcal{X}^l\), therefore
        \[
        v_i(X_i^l) + v_i(Y_i^l) 
        \geq v_i(M) - \min_{X \in \mathcal{X}^l} v_i(X)
        \geq v_i(M) - \mathrm{val}_i.
        \]
        \item If agent $l$ adopted \(\mathcal{X}^l\) during the \emph{second} adoption stage (Stage 5), it must have been adopted from {a previously adopted} partition (by \Cref{obs:adoption-no-new}).  
        Since agent \(i\) could have chosen to adopt this partition during Stage 5 as well, the same argument applies:
        \[
        val_i \geq \min_{X \in \mathcal{X}^l} v_i(X),
        \]
        implying again that
        \[
        v_i(X_i^l) + v_i(Y_i^l) \geq v_i(M) - \mathrm{val}_i.
        \]
    \end{enumerate}
\end{enumerate}

In all possible constructions, we obtain \( v_i(X_i^l) + v_i(Y_i^l) \geq v_i(M) - \mathrm{val}_i \), contradicting \cref{eq:1}.  
Thus, after completion of the algorithm, every agent receives her proportional share ex-ante.
\end{proof}

\section{Counterexamples: $(1-\varepsilon)$-approximate Algorithms May Fail BoBW Fairness} 

\subsection{Cut-and-Choose}\label{subsec:ptas-fail}

In this section, we demonstrate that simulating the classical 2-agent \emph{cut-and-choose} procedure using an $1-\varepsilon$ approximate algorithm to the maximin share (MMS) partitions (a PTAS or an FPTAS) may fail to satisfy both ex-ante proportionality and ex-post EFX. We further demonstrate that the variant of the algorithm that allows the agent that does not get her ex-ante proportional share to replace her partition with one guaranteeing a better worst-case value (and thus guaranteeing ex-ante proportionality), still has the problem that the resulting allocation may fail to satisfy EFX ex-post. 

To show this, consider a setting with two agents, and a set of seven indivisible goods $M = \{g_1, g_2, g_3, g_4, g_5, g_6, g_7\}$.
The two additive valuation functions $v_1,v_2$ are defined by the following item values: 
\[
\begin{array}{c|ccccccc}
 & g_1 & g_2 & g_3 & g_4 & g_5 & g_6 & g_7 \\
 \hline
 v_1(g) & 5 & 3 & 2 & 4 & 2 & 0.5 & 0.5 \\[2pt]
 v_2(g) & 2 & 1 & 1 & 0.5 & 0.5 & 0.5 & 0.5
\end{array}
\]
We note agent $1$'s MMS value is 8.5 (given for example by the partition $\{\{g_1,g_2,g_6\}, \{g_3,g_4, g_5,g_7\}\}$, and agent $2$'s MMS value is $3$ (given for example by the partition  $\{\{g_1,g_2\}, \{g_3, g_4, g_5, g_6, g_7\}\}$).
Assume that $\varepsilon = \frac{1}{5}$. A possible $(1-\varepsilon)$-MMS partition for agent~1 is:

\[
\mathcal{S}^1 =\{ S^1_1, S_2^1\} =\big\{ \{g_1, g_2, g_3\}, \{g_4, g_5, g_6, g_7\} \big\}
\]
$$v_1(S^1_1)=10, v_1(S^1_2)=7$$

Note that this partition is not EFX for agent~1: if agent~1 receives $S^1_2$, she strictly prefers $S^1_1$ even after the removal of a single good (e.g., $g_3$) from $S^1_1$.

For agent~2, the PTAS may produce the partition
\[
\mathcal{S}^2 =\{ S^2_1, S_2^2\} = \big\{\{g_1, g_2\}, \; \{g_3, g_4, g_5, g_6, g_7\} \big\},
\]
corresponding, in terms of agent~1’s valuations, to bundles of values $v_1(S^2_1) = 8$ and $v_1(S^2_2) = 9$. 
Under the randomized cut-and-choose procedure, agent~1 may receive her own small bundle $S^1_2$ (valued $7$) with probability $\tfrac{1}{2}$ and agent~2’s bundle $S^2_2$ (valued $9$) with probability $\tfrac{1}{2}$, yielding expected value $ 8  < \tfrac{v_1(M)}{2} = 8.5$.
Hence, the resulting allocation is not ex-ante proportional.
Moreover, since $S^1_2$ is not EFX, the overall procedure also fails to ensure ex-post EFX.
We also note that even if agent~1 were allowed to ``switch'' to agent~2’s partition, thereby guaranteeing herself ex-ante proportionality, she would still not be \emph{EFX}-satisfied when receiving bundle \(S^2_1\).

\subsection{\texorpdfstring{$(1-\varepsilon)$-Approximation Failure under Identical Valuations}{(1-epsilon)-Approximation Failure under Identical Valuations}}
\label{subsec:approx-ptas-fail}

In this section, we present an example demonstrating that even for three agents with identical valuations, allowing them to use different algorithms to compute \((1-\varepsilon)\)-approximate \(MMS\) partitions may cause the algorithm described in \cref{thm:bobw-poly-approx} to produce ex-ante envy. 

Consider a setting with three agents and five goods $M=(g_1,g_2,g_3,g_4, g_5)$, 
with common additive valuations with items of values

\[
\begin{array}{c|ccccc}
 & g_1 & g_2 & g_3 & g_4 & g_5 \\
 \hline
 v(g) & 4 & 2 & 6 & 5 & 1 
\end{array}
\]

An exact MMS partition in this case is 
\[
(\{g_1,g_2\}, \quad \{g_3\}, \quad \{g_4,g_5\}).
\]
Suppose that each agent computes her approximate MMS partition using a different {$(1-\varepsilon)$-approximation algorithm, with $\varepsilon = \frac{1}{4}$}.
Agents~1 and~2 happen to recover the exact MMS partition above,  
while agent~3 computes a slightly different partition:
\[
(\{g_4,g_2\}, \quad \{g_3\}, \quad \{g_1,g_5\}).
\]
Following the algorithm’s execution, agent~3 receives a value of~5 in the two allocations derived from her own (approximate) partition,  
and a value of~6 in the remaining allocations.  
Consequently, her ex-ante utility is strictly below her proportional share,  
demonstrating that using $(1-\varepsilon)$-approximate MMS partitions does not necessarily preserve ex-ante proportionality even under identical valuations (when different agents use different approximation MMS partitions).

\section{The Residual MMS (RMMS)}\label{sec:RMMS}

Recently, \citet{feige2025residualmaximinshare} introduced a new share based notion -- the \textit{residual maximin share}:
\begin{definition} [Residual Maximin Share (RMMS) {\cite[Definition 1]{feige2025residualmaximinshare}}]\label{def:RMMS}
The residual maximin share (RMMS) for valuation $v_i$ and number of agents $n$, denoted by $RMMS(v_i,n)$, is the highest value $t$ that is residual self-feasible. That is, for every $0 \leq k < n$, removing $k$ bundles each of value (under $v_i$)  strictly less than $t$, there is an $(n - k)$ partition of the set of remaining items, where each part has value (under $v_i$) at least $t$.
\end{definition}

When $n$ is clear from the context, we denote $RMMS_i=RMMS(v_i,n)$

An allocation is an \emph{RMMS allocation} if every agent $i$ receives a bundle of value at least $RMMS_i$. 
 An allocation is an \emph{$\alpha$-RMMS allocation} if every agent $i$ receives a bundle of value that is at least $\alpha\cdot RMMS_i$.

Moreover , as shown in \cite{feige2025residualmaximinshare}, for any $n$ and any additive valuation $v_i$ the RMMS is at most the MMS and at least the MXS. Therefore, we have 
$$PROP(v_i,n)\geq MMS(v_i,n)\geq RMMS(v_i,n) \geq MXS(v_i,n)$$

\subsection{\cref{alg:three-agent-construction} Guarantees each Agent her RMMS}

We prove the following lemma:
\begin{lemma}
    Each allocation produced by  \cref{alg:three-agent-construction} guarantees each agent her  RMMS.
\end{lemma}

\begin{proof}
First, as previously shown, in each allocation the \textbf{chooser} and the \textbf{divider} receive their proportional share and their \(MMS\) share, respectively.  
Consequently, both also receive at least their \(RMMS\) value.  
It remains to establish the guarantee for the \textbf{subdivider}.  

In the worst case, the subdivider receives the minimum valued bundle in the partition she has constructed.  
By \cref{obs:repartition-larger}, this bundle has value at least 
as high as the third bundle in the allocation.  
If this third bundle already exceeds or is equal to her \(RMMS\) value, the claim follows immediately.  
Otherwise, the third bundle is strictly below her RMMS value. 
By the definition of \(RMMS\), her repartition of the two remaining bundles guarantees that both bundles in the new partition are equal or larger than her \(RMMS\) value.  
Thus, the subdivider also receives at least her \(RMMS\).

\end{proof}

{We note that when running the algorithm described in \cref{thm:bobw-poly-approx}, each agent achieves her value guarantees up to a factor of \((1-\varepsilon)\).  
Consequently, every agent receives at least \((1-\varepsilon)\cdot RMMS\).
}

We also note that the RMMS and \( \tfrac{9}{10}\mathrm{MMS} \) guarantees are \emph{independent}, in the sense that any one of them can be larger than the other (see \cref{sec:RMMS-910MMS}). Our construction ensures that each agent receives a value at least as large as the \textbf{maximum} of the two.

 \subsection{Comparing RMMS and $\frac{9}{10}$-MMS.}\label{sec:RMMS-910MMS} 
We present examples demonstrating that the RMMS can be either higher or lower than  \( \frac{9}{10} \) of the MMS, thereby showing that the notions are independent. 

First, observe that in the case of three agents and three items, the RMMS coincides with the MMS, so the RMMS can be strictly larger than $9/10$ of the MMS.

To establish 
that the RMMS can be strictly smaller than $9/10$ of the MMS, 
consider an instance with three agents \( N = \{1,2,3\} \) and a set of 11 items \( M = \{g_1, \ldots, g_{11}\} \).  
All agents have identical additive valuation functions \( v_i = v \) for all \( i \in N \), where
\[
\begin{array}{c|ccccccccccc}
 & g_1 & g_2 & g_3 & g_4 & g_5 & g_6 & g_7 & g_8 & g_9 & g_{10} & g_{11} \\
 \hline
 v(g) & 7 & 7 & 8 & 1 & 1 & 1 & 1 & 1 & 1 & 1 & 1 
\end{array}
\]
The maximin share of each agent is
\[
\mathrm{MMS}_i = 10,
\]
as witnessed by the partition
\[
P = \big\{\, \{g_1,g_4,g_5,g_6\},\ \{g_2,g_7,g_8,g_9\},\ \{g_3,g_{10},g_{11}\} \,\big\},
\]
in which every bundle has a value of 10.

Assume, for the sake of contradiction, that
\[
\mathrm{RMMS}_i \ge \tfrac{9}{10}\mathrm{MMS}_i = 9.
\]
By the definition of RMMS, removing one bundle of value \( 8 \) (for instance, the bundle consisting of all items valued \( 1 \)) should leave a set of items \( \{g_1,g_2,g_3\} \) that can be partitioned into two bundles, each with value at least \( 9 \).  
However, when removing items 4 through 11, we are left with items $g_1,g_2,g_3$ -- 
no partition of these three items into two bundles can yield a minimum bundle value of at least \( 9 \).  
Hence, such a partition does not exist, and therefore
\[
\mathrm{RMMS}_i < 9 =\tfrac{9}{10}\mathrm{MMS}_i.
\]

\section{IMMX Allocations for the EFX Counterexample Instances}
\label{immx-proof}

In this section, we analyze the two recently discovered 3-agent instances 
for which no EFX allocation exists: the instance of 
\citet{akrami2026counterexampleefxnge} with eight goods and monotone 
valuations, and the instance of 
\citet{mackenzie2026counterexamplesefxsubmodularsubadditive} with monotone 
submodular valuations. We establish two results for each instance. First, 
we show that neither instance admits an MMS allocation: for the submodular 
instance, we prove that no MMS allocation exists under \emph{any} cardinal 
valuation consistent with the ordinal rankings of the original example, 
while for the \citet{akrami2026counterexampleefxnge} instance we verify 
the non-existence of an MMS allocation by exhaustive search. Second, we 
show that both instances nevertheless admit an \IMMX allocation. Together, 
these results demonstrate that \IMMX can be satisfied even in instances 
exhibiting a simultaneous failure of both EFX and MMS, leaving open the 
possibility that \IMMX is feasible in general.

\begin{claim}
The \citet{mackenzie2026counterexamplesefxsubmodularsubadditive} construction admits no MMS allocation. More generally,
the same conclusion holds for every cardinal valuation profile
\((v_0,v_1,v_2)\) which is compatible with the rank ordering, in the sense that
for every agent \(i\) and all bundles \(S,T\subseteq M\),
\[
r_i(S)>r_i(T)
\quad\Longrightarrow\quad
v_i(S)>v_i(T).
\]
\end{claim}

\begin{proof}
Recall that there are eight items,
\[
M=\{0,1,\dots,7\},
\]
which are grouped into 5 types:
\[
A=\{0,3\},\qquad B=\{1,4\},\qquad C=\{2,5\},
\qquad x=6,\qquad y=7.
\]
Agent \(0\)'s ordinal valuation is given by a rank function
\[
r_0:2^M\to \{0,1,\dots,7\}.
\]
The other two rank functions are obtained by the cyclic relabeling
\[
\sigma=(0\,1\,2)(3\,4\,5),
\]
namely
\[
r_1(S)=r_0(\sigma(S)),
\qquad
r_2(S)=r_0(\sigma^2(S)).
\]

For completeness, we reproduce the rank function $r_0$ of
\citet{mackenzie2026counterexamplesefxsubmodularsubadditive}. Recall that
$M=\{0,1,\dots,7\}$, grouped into five types: $A=\{0,3\}$, $B=\{1,4\}$,
$C=\{2,5\}$, $x=6$, and $y=7$. The rank functions $r_1,r_2$ of the other
two agents are obtained via the cyclic relabeling
$\sigma=(0\,1\,2)(3\,4\,5)$, namely $r_1(S)=r_0(\sigma(S))$ and
$r_2(S)=r_0(\sigma^2(S))$. The rank $r_0$ is defined as follows:

\begin{itemize}
    \item For singletons, set
    \[
    r_0(\{g\}) = 1 \qquad \text{for every } g \in M.
    \]

    \item For $|S| = 2$, the rank depends only on the types of the two
    goods, as shown in \cref{tab:ms-pairs}. For example, any pair with one
    good from $A$ and one good from $B$ has rank $2$, whereas any pair with
    one good from $C$ and the good $x$ has rank $1$. Diagonal entries apply
    to pairs consisting of two distinct goods of the same type.

    \begin{table}[H]
    \centering
    \begin{tabular}{@{}cccccc@{}}
    \toprule
        & $A$ & $B$ & $C$ & $x$ & $y$ \\
    \midrule
    $A$ & $1$ & $2$ & $2$ & $4$ & $6$ \\
    $B$ & $2$ & $1$ & $5$ & $1$ & $3$ \\
    $C$ & $2$ & $5$ & $1$ & $1$ & $3$ \\
    $x$ & $4$ & $1$ & $1$ & $-$ & $1$ \\
    $y$ & $6$ & $3$ & $3$ & $1$ & $-$ \\
    \bottomrule
    \end{tabular}
    \caption{Ranks $r_0$ of sets with cardinality two
    \cite{mackenzie2026counterexamplesefxsubmodularsubadditive}.}
    \label{tab:ms-pairs}
    \end{table}

    \item For $|S| = 3$, we call $S$ \emph{exceptional} if it contains one
    good from $A \cup \{x\}$, one good from $B$, and one good from $C$.
    Exceptional triples are assigned rank $7$. Every other triple inherits
    the largest rank among its three internal pairs:
    \[
    r_0(S) =
    \begin{cases}
        7, & \text{if } S \text{ is exceptional},\\[2pt]
        \max_{\{g,h\} \subseteq S} r_0(\{g,h\}), & \text{otherwise}.
    \end{cases}
    \]
    The exceptional triples are listed in \cref{tab:ms-triples}.

    \begin{table}[H]
    \centering
    \begin{tabular}{@{}ll@{}}
    \toprule
    Type  & Triples \\
    \midrule
    $ABC$ & $012,\ 015,\ 024,\ 045,\ 123,\ 135,\ 234,\ 345$ \\
    $BCx$ & $126,\ 156,\ 246,\ 456$ \\
    \bottomrule
    \end{tabular}
    \caption{Exceptional triples; these have the highest possible rank
    of $7$ \cite{mackenzie2026counterexamplesefxsubmodularsubadditive}.}
    \label{tab:ms-triples}
    \end{table}

    \item For $|S| \geq 4$, set
    \[
    r_0(S) = \max_{\substack{T \subseteq S \\ |T| = 3}} r_0(T).
    \]
    Equivalently, a set of size at least four has rank $7$ exactly when it
    contains an exceptional triple; otherwise, it inherits the largest
    rank of an internal pair.
\end{itemize}

We first observe that, at the ordinal level, the MMS rank is \(6\). For agent
\(0\), consider the partition
\[
\{0,7\},\qquad \{1,2,6\},\qquad \{3,4,5\}.
\]
In type notation this is
\[
Ay,\qquad BCx,\qquad ABC.
\]
The ranks of these three bundles for agent \(0\) are
\[
6,\qquad 7,\qquad 7.
\]
Thus agent \(0\) can guarantee rank at least \(6\). On the other hand, a
rank-\(7\) bundle must contain an exceptional triple, and in particular must
contain at least three items. Hence three rank-\(7\) bundles would require at
least \(9\) items, while there are only \(8\). Therefore agent \(0\)'s ordinal
MMS rank is exactly \(6\). The same argument applies to agents \(1\) and \(2\)
by cyclic relabeling.

We now pass to cardinal valuations. Let
\[
\mu_i=\operatorname{MMS}_i(v_i)
\]
be agent \(i\)'s cardinal maximin share. For agent \(0\), the partition above
has ranks \(6,7,7\). By compatibility with the rank ordering, the two rank-\(7\)
bundles are strictly more valuable than the rank-\(6\) bundle \(\{0,7\}\).
Hence
\[
\mu_0 \geq v_0(\{0,7\}).
\]
Now take any bundle \(S\) with \(r_0(S)\leq 5\). Since \(r_0(\{0,7\})=6\),
compatibility gives
\[
v_0(S)<v_0(\{0,7\})\leq \mu_0.
\]
Therefore any bundle that gives agent \(0\) at least her MMS value must have
rank at least \(6\):
\[
v_0(S)\geq \mu_0
\quad\Longrightarrow\quad
r_0(S)\geq 6.
\]
By cyclic relabeling, the same conclusion holds for every agent \(i\in\{0,1,2\}\):
\[
v_i(S)\geq \mu_i
\quad\Longrightarrow\quad
r_i(S)\geq 6.
\]

Thus, if an MMS allocation existed, every agent would have to receive a bundle
of rank at least \(6\). We now show that this is impossible.

For agent \(0\), a bundle has rank at least \(6\) exactly if it either contains
an \(A\)-good together with \(y\), or contains an exceptional triple. Equivalently,
\[
r_0(S)\geq 6
\quad\Longleftrightarrow\quad
S\text{ contains }Ay
\quad\text{or}\quad
S\text{ contains }B,C,\text{ and one of }A,x.
\]
Applying the cyclic relabeling, the corresponding conditions for the three
agents are:
\[
\begin{array}{c|c|c}
\text{agent} & \text{\(y\)-option} & \text{exceptional option} \\ \hline
0 & Ay & B,C,\text{ and one of }A,x \\
1 & Cy & A,B,\text{ and one of }C,x \\
2 & By & C,A,\text{ and one of }B,x .
\end{array}
\]
Here \(A,B,C\) refer to the original item groups.

Suppose, for contradiction, that there is an MMS allocation
\[
X=(X_0,X_1,X_2).
\]
Then every agent must receive a bundle of rank at least \(6\), and therefore must
satisfy either the \(y\)-option or the exceptional option above.

Not all three agents can satisfy the exceptional option, since each exceptional
bundle requires at least three items, and three such bundles would require at
least \(9\) items. Therefore, at least one agent must satisfy the \(y\)-option.
Since there is only one item \(y\), exactly one agent satisfies the \(y\)-option.

We now consider the three cases. If agent \(0\) receives the \(y\)-option, then
\(X_0\) contains \(y\) and an item from \(A\). Agents \(1\) and \(2\) must then
use their exceptional options, and both of these exceptional options require
an \(A\)-good. Thus all three bundles require an \(A\)-good, but there are only
two \(A\)-goods, namely \(0\) and \(3\), a contradiction.

Similarly, if agent \(1\) receives the \(y\)-option, then all three bundles
require a \(C\)-good, although there are only two \(C\)-goods. If agent \(2\)
receives the \(y\)-option, then all three bundles require a \(B\)-good, although
there are only two \(B\)-goods. Both cases are impossible.

Therefore there is no allocation in which every agent receives a bundle of
rank at least \(6\). Since every MMS allocation would have to satisfy this
property, no MMS allocation exists.
\end{proof}

\begin{claim}
The \citet{mackenzie2026counterexamplesefxsubmodularsubadditive} construction admits an \IMMX allocation for every monotone
cardinal valuation profile compatible with the rank ordering.
\end{claim}

\begin{proof}
Assume that the valuation profile \((v_0,v_1,v_2)\) is monotone and compatible
with the rank ordering.

We first compute agent \(0\)'s MMS value. As shown above, any partition giving
agent \(0\) three bundles of rank at least \(6\) must contain exactly one
\(Ay\)-bundle. Since the other two bundles must contain exceptional triples,
and there are only eight items in total, this \(Ay\)-bundle must have size
exactly two. Hence it is either
\[
\{0,7\}
\quad\text{or}\quad
\{3,7\}.
\]
Moreover, both possibilities are feasible:
\[
\{0,7\},\qquad \{1,2,6\},\qquad \{3,4,5\},
\]
and
\[
\{3,7\},\qquad \{1,2,6\},\qquad \{0,4,5\}
\]
both have ranks \(6,7,7\) for agent \(0\). Therefore
\[
\operatorname{MMS}_0
=
\max\{v_0(\{0,7\}),v_0(\{3,7\})\}.
\]

Let \(a\in\{0,3\}\) be chosen so that
\[
v_0(\{a,7\})
=
\max\{v_0(\{0,7\}),v_0(\{3,7\})\},
\]
and let \(a'\) be the other item in \(\{0,3\}\). Consider the allocation
\[
X_0=\{a,7\},
\qquad
X_1=\{1,2,4,6\},
\qquad
X_2=\{a',5\}.
\]

By the choice of \(a\), agent \(0\) receives her MMS value:
\[
v_0(X_0)=\operatorname{MMS}_0.
\]
It remains to show that agents \(1\) and \(2\) are EFX-satisfied.

For agent \(1\), the bundle
\[
X_1=\{1,2,4,6\}
\]
has rank
\[
r_1(X_1)=4.
\]
Both \(X_0\) and \(X_2\) have size two, so after deleting any one item from
either of them, only a singleton remains. Since every singleton has rank \(1\),
we have
\[
r_1(X_j\setminus\{g\})=1<4=r_1(X_1)
\]
for every \(j\in\{0,2\}\) and every \(g\in X_j\). By compatibility with the
rank ordering,
\[
v_1(X_1)>v_1(X_j\setminus\{g\})
\]
for every such \(j\) and \(g\). 
Therefore agent \(1\) is EFX-satisfied.

Since \(r_2(S)=r_0(\sigma^2(S))\), and since
\[
\sigma^2(5)=4\in B,
\]
the item \(5\) is a \(B\)-good for agent \(2\). Also, if
\(a'\in\{0,3\}\), then
\[
\sigma^2(a')\in\{2,5\}=C.
\]
Thus the bundle
\[
X_2=\{a',5\}
\]
is seen by agent \(2\) as a \(BC\)-type pair. Hence
\[
r_2(X_2)=5.
\]

Again, after deleting one item from \(X_0\), only a singleton remains, so
\[
r_2(X_0\setminus\{g\})=1<5=r_2(X_2)
\]
for every \(g\in X_0\).

It remains to compare \(X_2\) with deletions from \(X_1\). The possible
one-item deletions from \(X_1=\{1,2,4,6\}\) are
\[
\{2,4,6\},\qquad
\{1,4,6\},\qquad
\{1,2,6\},\qquad
\{1,2,4\}.
\]
Under agent \(2\)'s rank function, these bundles have ranks
\[
4,\qquad 4,\qquad 4,\qquad 2.
\]
Thus, for every \(g\in X_1\),
\[
r_2(X_1\setminus\{g\})\leq 4<5=r_2(X_2).
\]
By compatibility with the rank ordering,
\[
v_2(X_2)>v_2(X_1\setminus\{g\})
\qquad
\text{for every }g\in X_1.
\].

Hence agent \(0\) receives her MMS value, while agents \(1\) and \(2\) are
EFX-satisfied. The allocation is therefore \IMMX.
\end{proof}

\begin{claim} 
The \citet{akrami2026counterexampleefxnge} instance admits an \IMMX allocation. 
\end{claim} 

\begin{proof} 

The valuations of the \citet{akrami2026counterexampleefxnge} instance are specified explicitly over all $2^8$
bundles; we do not reproduce this 256-entry table here and refer the
reader to \cite{akrami2026counterexampleefxnge} for the full
specification. Below we quote only the bundle values needed to verify the
claim.

Consider the allocation \[ X_0=\{g_1,g_2\},\qquad X_1=\{g_3,g_7\},\qquad X_2=\{g_0,g_4,g_5,g_6\}. \] By exhaustive search over all partitions of the goods into three bundles, the MMS values in this instance are \[ \mu_0=77,\qquad \mu_1=75,\qquad \mu_2=73. \] In the allocation above, we have \[ v_0(X_0)=77=\mu_0 \] and \[ v_2(X_2)=126\geq \mu_2. \] Thus agents \(0\) and \(2\) receive their MMS values. It remains to check that agent \(1\) is EFX-satisfied. We have \[ v_1(X_1)=v_1(\{g_3,g_7\})=62. \] First, deleting any item from \(X_0=\{g_1,g_2\}\) leaves either \(\{g_1\}\) or \(\{g_2\}\), whose values for agent \(1\) are \[ v_1(\{g_1\})=2, \qquad v_1(\{g_2\})=6. \] Both are strictly below \(v_1(X_1)=62\). Next, deleting one item from \[ X_2=\{g_0,g_4,g_5,g_6\} \] leaves one of the following bundles: \[ \{g_4,g_5,g_6\},\quad \{g_0,g_5,g_6\},\quad \{g_0,g_4,g_6\},\quad \{g_0,g_4,g_5\}. \] Their values for agent \(1\) are respectively \[ 32,\qquad 59,\qquad 47,\qquad 57, \] all of which are strictly below \(62\). Therefore, agent \(1\) does not strongly envy either of the other two agents. Hence agents \(0\) and \(2\) are MMS-satisfied, while agent \(1\) is EFX-satisfied. The allocation is therefore \IMMX. 
\end{proof}

\section{IMMX Allocations Exist for Three Agents with Chores}
\label{sec:immx-chores}

In this section we show that the construction underlying our three-agent result (\Cref{thm:three-agents-main}) extends to the setting of \emph{chores}, establishing that for three agents with additive cost functions an IMMX allocation always exists. We find this existence result of independent interest: for chores, MMS allocations are known not to exist already for three agents \cite{ARSW17}, and the existence of EFX allocations remains open for three agents (and fails for four or more agents \cite{he2026efxadditivechoresnonexistence}). We focus here on the existence of a single (deterministic) IMMX allocation, and as we only require a single allocation, the construction takes a particularly simple form: the three roles of \emph{divider}, \emph{subdivider} and \emph{chooser} are fixed in advance (and may be assigned to the agents arbitrarily).

The construction below runs in polynomial time, with the sole exception of computing each agent's MMS value, which is NP-hard in general. For a fixed number of agents, however, this computation is exactly the problem of partitioning the chores into $n$ bundles so as to minimize the maximum bundle cost, i.e.\ makespan minimization on $n$ identical machines. This problem admits an FPTAS for any fixed number of machines \cite{Sahni76}, so each MMS value can be approximated within a factor of $(1+\varepsilon)$ in polynomial time. Substituting these approximate values for the exact MMS thresholds yields a polynomial-time procedure producing an allocation that is \EIMMXCHORES.  

\subsection{Model and Fairness Notions for Chores}

A \emph{chores instance} consists of a set $M$ of $m$ indivisible chores and a set $N$ of $n$ agents, where each agent $i$ has a normalized additive \emph{cost function} $c_i : 2^M \to \mathbb{R}_{\geq 0}$; that is, $c_i(\emptyset) = 0$, $c_i(g) \geq 0$ for every $g \in M$, and $c_i(S) = \sum_{g \in S} c_i(\{g\})$ for every $S \subseteq M$. An allocation is defined as for goods; each agent prefers bundles of \emph{lower} cost. The \emph{proportional share} of agent $i$ is $PROP(c_i, n) = \frac{c_i(M)}{n}$, and an allocation is \emph{proportional} if every agent $i$ incurs a cost of at most $PROP(c_i, n)$.

\begin{definition}[Maximin share for chores]
The maximin share (MMS) of agent $i$ with cost function $c_i$ over the set of chores $M$, and a given number of agents $n$, is defined to be
\[
MMS(c_i, n) := \min_{(Z_1, \ldots, Z_n) \in \Gamma_n} \; \max_{j \in [n]} c_i(Z_j),
\]
where $\Gamma_n$ is the set of all partitions of $M$ into $n$ bundles. An allocation $X$ is an \emph{MMS allocation} if $c_i(X_i) \leq MMS(c_i, n)$ for every agent $i$.

An allocation is an \emph{$\alpha$-MMS allocation} if every agent $i$ receives a bundle of value that is at at most $\alpha\cdot \text{MMS}(v_i,n)$.
\end{definition}

We note that for chores the direction of the comparison between the proportional share and the MMS is reversed relative to goods: since the maximum bundle cost in any partition is at least the average, for every additive cost function it holds that
\begin{equation}\label{eq:chores-prop-mms}
PROP(c_i, n) \leq MMS(c_i, n).
\end{equation}

We use the standard adaptation of EFX to chores, in which a chore is removed from the bundle of the \emph{envious} agent. As for goods, we use the stronger variant in which envy must be eliminated even when removing zero-cost chores.

\begin{definition}[EFX for chores]\label{def:chores-efx}
Agent $i$ with cost function $c_i$ is \emph{EFX-satisfied} by allocation $X = (X_1, \ldots, X_n)$ if for every other agent $k$ and every chore $g \in X_i$ it holds that
\[
c_i(X_i \setminus \{g\}) \leq c_i(X_k).
\]
The allocation $X$ is an \emph{EFX allocation} if every agent is EFX-satisfied by $X$.
We say that a bundle $X$ \emph{EFX-dominates} a bundle $Y$ for cost function $c$ if for every $g \in X$ it holds that $c(X \setminus \{g\}) \leq c(Y)$. Note that agent $i$ is EFX-satisfied by $X$ if and only if her bundle $X_i$ EFX-dominates every other bundle in $X$.
\end{definition}

\begin{observation}\label{obs:chores-min-implies-efx}
Given an allocation $(X_1, \ldots, X_n)$, if agent $i$ incurs the lowest cost among all bundles, i.e., $c_i(X_i) \leq \min_{k \in [n]} c_i(X_k)$, then agent $i$ is envy-free, and in particular EFX-satisfied (as removing a chore from her bundle can only decrease its cost).
\end{observation}

\begin{definition}[IMMX for chores]
Given cost functions $(c_1, \ldots, c_n)$ over a set $M$ of indivisible chores, an allocation $X = (X_1, \ldots, X_n)$ is an \emph{``Individually MMS-satisfying or EFX-satisfying (IMMX)''} allocation if for every agent $i \in [n]$, either
\[
c_i(X_i) \leq MMS(c_i, n),
\]
or agent $i$ is EFX-satisfied by $X$.
\end{definition}

\subsection{An EFX Re-allocation Procedure for Chores}

Our construction requires a chores analog of the \textsc{Realloc} procedure (\Cref{alg:realloc}): a polynomial-time transformation of an arbitrary partition into one that is EFX with respect to a single cost function, without increasing the cost of the most costly bundle. The procedure itself is identical to \textsc{Realloc}, applied verbatim to the cost function $c$ (process the chores in non-increasing order of cost, and insert each into a currently minimum-cost bundle); only the analysis differs. For goods, the relevant invariant is that the minimum bundle value never decreases; for chores, we establish that, in addition, the \emph{maximum} bundle cost never increases.

\begin{lemma}\label{lem:chores-realloc}
Fix an additive cost function $c$, any number of agents $n$, and a partition $(B_1, \ldots, B_n)$ of a set $M$ of chores. Then there is a polynomial-time algorithm (\textsc{Chore-Realloc}) that transforms $(B_1, \ldots, B_n)$ into a partition $(B'_1, \ldots, B'_n)$ such that:
\begin{enumerate}
\item every bundle in $(B'_1, \ldots, B'_n)$ EFX-dominates every other bundle (with respect to $c$), and
\item $\max_{t \in [n]} c(B'_t) \leq \max_{t \in [n]} c(B_t)$ \; (and additionally $\min_{t \in [n]} c(B'_t) \geq \min_{t \in [n]} c(B_t)$).
\end{enumerate}
\end{lemma}

The procedure \textsc{Chore-Realloc} is the following: let $(w_1, \ldots, w_m)$ be a non-increasing ordering of the chores by cost ($c(\{w_1\}) \geq \cdots \geq c(\{w_m\})$). For $i = 1, \ldots, m$: remove $w_i$ from the bundle currently holding it, and insert it into a bundle of currently minimum cost (breaking ties arbitrarily). We first establish the two monotonicity invariants.

\begin{lemma}\label{lem:chores-invariants}
Throughout the execution of \textsc{Chore-Realloc}, the maximum bundle cost never increases, and the minimum bundle cost never decreases.
\end{lemma}

\begin{proof}
Consider the iteration processing chore $w$, in which $w$ is removed from bundle $B_k$ and inserted into a bundle $B_j$ of minimum cost after the removal (so $c(B_j) \leq c(B_k \setminus \{w\})$ and $c(B_j) \leq c(B_t)$ for every $t \neq k$). If $j = k$ then $w$ is returned to its bundle and nothing changes, so assume $j \neq k$. Only bundles $k$ and $j$ change.

\emph{Maximum.} Bundle $k$'s cost decreases. Bundle $j$'s new cost is
\[
c(B_j) + c(\{w\}) \leq c(B_k \setminus \{w\}) + c(\{w\}) = c(B_k),
\]
which is at most the maximum bundle cost before the iteration. Hence the maximum never increases.

\emph{Minimum.} After the iteration, bundle $k$ has cost $c(B_k \setminus \{w\}) \geq c(B_j)$, bundle $j$ has cost $c(B_j) + c(\{w\}) \geq c(B_j)$, and every other bundle is unchanged with cost at least $c(B_j)$. Since bundle $j$ was unchanged by the removal, $c(B_j)$ is at least the minimum bundle cost before the iteration. Hence the minimum never decreases.
\end{proof}

\begin{proof}[Proof of \Cref{lem:chores-realloc}]
Property 2 is immediate from \Cref{lem:chores-invariants}, and the procedure clearly runs in polynomial time. It remains to prove Property 1.

Assume for contradiction that the final partition $(B'_1, \ldots, B'_n)$ is not EFX, i.e., there exist bundles $B'_i, B'_j$ and a chore $g' \in B'_i$ such that
\[
c(B'_i \setminus \{g'\}) > c(B'_j).
\]
Every chore is inserted (possibly back into its own bundle) exactly once, at the iteration in which it is processed, and does not move thereafter. Hence the chores of $B'_i$ were inserted into bundle $i$ at their respective iterations, in non-increasing order of cost. Let $g$ be the member of $B'_i$ that was inserted last; then $g$ is a minimum-cost chore of $B'_i$, and therefore
\[
c(B'_i \setminus \{g\}) \geq c(B'_i \setminus \{g'\}) > c(B'_j).
\]
Consider the iteration $t$ in which $g$ was inserted into bundle $i$, and let $S$ denote the content of bundle $i$ at that moment, just before the insertion (after $g$ was removed from its previous location). By the insertion rule, bundle $i$ was of minimum cost at that moment: $c(S) \leq c(R)$ for every other bundle $R$ at iteration $t$. Since no member of $B'_i$ is inserted after iteration $t$, and chores never move except at their own iteration, we have $B'_i \setminus \{g\} \subseteq S$, and therefore
\[
c(B'_i \setminus \{g\}) \leq c(S).
\]
Immediately after iteration $t$, every bundle has cost at least $c(S)$: bundle $i$ has cost $c(S) + c(\{g\}) \geq c(S)$, and every other bundle has cost at least $c(S)$ by the minimality of $S$. By \Cref{lem:chores-invariants}, the minimum bundle cost never decreases in subsequent iterations, and therefore
\[
c(B'_j) \geq c(S) \geq c(B'_i \setminus \{g\}),
\]
a contradiction.
\end{proof}

In analogy to the goods setting, we define:

\begin{definition}[$MMS\!-\!EFX$ partition for chores]
A partition $\{P_1, \ldots, P_n\}$ of a set of chores into $n$ bundles is an \emph{$MMS\!-\!EFX$ partition for chores} for agent $i$ with cost function $c_i$ if every bundle $P \in \{P_1, \ldots, P_n\}$ satisfies $c_i(P) \leq MMS(c_i, n)$, and every bundle EFX-dominates every other bundle with respect to $c_i$.
\end{definition}

\begin{corollary}\label{cor:chores-mms-efx}
There exists a polynomial-time algorithm that, given an additive cost function $c$ over a set of chores, a number of agents $n$, and any $MMS(c, n)$ partition, outputs an $MMS\!-\!EFX$ partition for chores for the agent.
\end{corollary}

\begin{proof}
In an MMS partition for chores, every bundle has cost at most $MMS(c, n)$. Applying \textsc{Chore-Realloc} (\Cref{lem:chores-realloc}) preserves this property (the maximum does not increase) and makes the bundles pairwise EFX-dominating.
\end{proof}

\subsection{The Construction and the Main Result}

\begin{theorem}\label{thm:chores-immx}
Consider settings with three agents with additive cost functions $(c_1, c_2, c_3)$ over a set $M$ of indivisible chores, with the three roles of divider, subdivider, and chooser assigned to the agents arbitrarily. There exists an allocation $X = (X_1, X_2, X_3)$ such that:
\begin{enumerate}
\item The divider incurs a cost of at most her MMS.
\item The chooser incurs a cost of at most her proportional share (and therefore, by \Cref{eq:chores-prop-mms}, of at most her MMS).
\item The subdivider either incurs a cost of at most her proportional share, or is EFX-satisfied.
\end{enumerate}
In particular, $X$ is IMMX: at most one agent (the subdivider) may incur a cost exceeding her MMS, and if she does, she is EFX-satisfied. 
\end{theorem}

As noted above, the only non-polynomial component of the construction in the proof of \Cref{thm:chores-immx}  is the computation of the exact MMS values. Approximating these values via the FPTAS for makespan minimization yields the following corollary:

\begin{corollary}\label{cor:chores-immx-fptas}
Fix any $\varepsilon > 0$. Consider settings with three agents with additive cost functions $(c_1, c_2, c_3)$ over a set $M$ of indivisible chores, with the three roles of divider, subdivider, and chooser assigned to the agents arbitrarily. There exists  a fully polynomial-time approximation scheme (FPTAS) that outputs an allocation $X = (X_1, X_2, X_3)$ such that:
\begin{enumerate}
\item The divider incurs a cost of at most $(1+\varepsilon)$ of her MMS.
\item The chooser incurs a cost of at most her proportional share (and therefore, by \Cref{eq:chores-prop-mms}, of at most her MMS).
\item The subdivider either incurs a cost of at most her proportional share, or is EFX-satisfied.
\end{enumerate}
In particular, $X$ is \EIMMXCHORES: at most one agent (the subdivider) may incur a cost exceeding $(1+\varepsilon)$ of her MMS, and if she does, she is EFX-satisfied. 
\end{corollary}

The construction of \Cref{thm:chores-immx} is a simplified, single-allocation variant of \Cref{alg:three-agent-construction}, with the roles fixed in advance. Assume w.l.o.g.\ that agent 1 is the \emph{divider}, agent 2 is the \emph{subdivider}, and agent 3 is the \emph{chooser}. The construction proceeds as follows:

\begin{enumerate}
\item \textbf{The divider partitions.} Let $\{A, B, C\}$ be an MMS partition of $M$ into three bundles with respect to $c_1$ (a partition minimizing the maximum bundle cost).
\item \textbf{Distinct low-cost bundles.} If there exist two \emph{distinct} bundles $S_2 \neq S_3$ in $\{A, B, C\}$ such that $c_2(S_2) \leq PROP(c_2, 3)$ and $c_3(S_3) \leq PROP(c_3, 3)$, then allocate $S_2$ to the subdivider, $S_3$ to the chooser, and the remaining bundle to the divider, and terminate (\textbf{Case I}).
\item \textbf{The subdivider repartitions.} Otherwise (\textbf{Case II}), by \Cref{obs:chores-dichotomy} below, there is a single common bundle -- w.l.o.g.\ bundle $A$ -- that is the unique bundle of cost at most the proportional share for each of agents 2 and 3. The subdivider takes her two cheapest bundles -- namely $A$ together with $Z := \arg\min_{Z' \in \{B, C\}} c_2(Z')$ (ties broken arbitrarily) -- and computes an $MMS\!-\!EFX$ partition for chores $\{P_1, P_2\}$ of $A \cup Z$ into two bundles with respect to $c_2$ (an MMS partition of $A \cup Z$ for $n = 2$, followed by \textsc{Chore-Realloc}; see \Cref{cor:chores-mms-efx}). Let $L = M \setminus (A \cup Z)$ denote the remaining bundle of the original partition.
\item \textbf{The chooser chooses.} The chooser receives her minimum-cost bundle among $\{P_1, P_2\}$, the subdivider receives the other bundle of $\{P_1, P_2\}$, and the divider receives $L$.
\end{enumerate}

We first establish that the case distinction in Step 2 is exhaustive in the following sense:

\begin{observation}\label{obs:chores-dichotomy}
For $l \in \{2, 3\}$, let $U_l = \{S \in \{A, B, C\} : c_l(S) \leq PROP(c_l, 3)\}$. If Step 2 fails -- i.e., there is no choice of distinct bundles $S_2 \in U_2$ and $S_3 \in U_3$ -- then $U_2 = U_3 = \{A\}$ for some common bundle (denoted w.l.o.g.\ by $A$). In particular, in Case II it holds for each $l \in \{2, 3\}$ that
\[
c_l(A) \leq PROP(c_l, 3) \qquad \text{and} \qquad c_l(B), c_l(C) > PROP(c_l, 3).
\]
\end{observation}

\begin{proof}
First, $U_l \neq \emptyset$ for each $l$, as a minimum-cost bundle in any partition into three bundles has cost at most the average, which equals $PROP(c_l, 3)$. Now, if $|U_2| \geq 2$, pick any $S_3 \in U_3$ and then $S_2 \in U_2 \setminus \{S_3\}$, so Step 2 succeeds; the case $|U_3| \geq 2$ is symmetric. Hence if Step 2 fails, both $U_2$ and $U_3$ are singletons, and they must be the same singleton (otherwise their distinct elements are a valid choice).
\end{proof}

The following observation plays the role of \Cref{obs:repartition-larger} in the goods construction; we note that here it follows by a particularly short argument.

\begin{observation}\label{obs:chores-B1}
In Case II it holds that
\[
\max\{c_2(P_1), c_2(P_2)\} \;\leq\; c_2(Z) \;=\; \min\{c_2(B), c_2(C)\} \;\leq\; c_2(L).
\]
\end{observation}

\begin{proof}
By \Cref{obs:chores-dichotomy}, $c_2(A) \leq PROP(c_2, 3) < c_2(Z)$, so the trivial partition $\{A, Z\}$ of $A \cup Z$ has maximum cost $\max\{c_2(A), c_2(Z)\} = c_2(Z)$. The MMS partition of $A \cup Z$ (for $n = 2$) minimizes the maximum cost, and \textsc{Chore-Realloc} does not increase it (\Cref{lem:chores-realloc}); hence $\max\{c_2(P_1), c_2(P_2)\} \leq c_2(Z)$. Finally, $c_2(Z) = \min\{c_2(B), c_2(C)\}$ by the choice of $Z$, and $L$ is the other bundle of $\{B, C\}$.
\end{proof}

\begin{proof}[Proof of \Cref{thm:chores-immx}]
We analyze the two cases of the construction.

\textbf{Case I.} The divider receives a bundle of her own MMS partition, and in an MMS partition for chores every bundle has cost at most $MMS(c_1, 3)$; hence Property 1 holds. The subdivider and the chooser receive bundles of cost at most $PROP(c_2, 3)$ and $PROP(c_3, 3)$, respectively; hence Properties 2 and 3 hold (and by \Cref{eq:chores-prop-mms}, all three agents incur a cost of at most their MMS).

\textbf{Case II.} The realized allocation is a permutation of the partition $\{P_1, P_2, L\}$.
\begin{itemize}
\item \emph{The divider.} She receives $L \in \{B, C\}$, a bundle of her own MMS partition, and therefore incurs a cost of at most $MMS(c_1, 3)$, establishing Property 1.
\item \emph{The chooser.} Her cost is $\min\{c_3(P_1), c_3(P_2)\}$, which is at most the average:
\[
\min\{c_3(P_1), c_3(P_2)\} \;\leq\; \frac{c_3(P_1) + c_3(P_2)}{2} \;=\; \frac{c_3(A \cup Z)}{2} \;=\; \frac{c_3(M) - c_3(L)}{2}.
\]
We emphasize that $Z$ was chosen according to the \emph{subdivider's} costs, so the chooser may well disagree about which of $\{B, C\}$ is the cheaper bundle; nevertheless, by \Cref{obs:chores-dichotomy} \emph{both} bundles of $\{B, C\}$ cost the chooser more than her proportional share, and in particular $c_3(L) > PROP(c_3, 3)$. Therefore
\[
\min\{c_3(P_1), c_3(P_2)\} \;<\; \frac{3 \cdot PROP(c_3, 3) - PROP(c_3, 3)}{2} \;=\; PROP(c_3, 3),
\]
establishing Property 2. Moreover, her bundle is of minimum cost in the realized partition: it costs no more than the other bundle of $\{P_1, P_2\}$ by her choice, and strictly less than $c_3(L)$ since it is below her proportional share while $c_3(L)$ is above it. By \Cref{obs:chores-min-implies-efx}, the chooser is also EFX-satisfied.
\item \emph{The subdivider.} She receives one of the two bundles of $\{P_1, P_2\}$; denote it by $X_2$ and the other bundle by $P$. We show she is EFX-satisfied, establishing Property 3 (via the EFX clause). First, $X_2$ EFX-dominates $P$ with respect to $c_2$, as $\{P_1, P_2\}$ is an $MMS\!-\!EFX$ partition for chores for the subdivider. Second, for every chore $g \in X_2$,
\[
c_2(X_2 \setminus \{g\}) \;\leq\; c_2(X_2) \;\leq\; \max\{c_2(P_1), c_2(P_2)\} \;\leq\; c_2(L),
\]
where the last inequality is \Cref{obs:chores-B1}; hence $X_2$ EFX-dominates $L$ as well. Therefore the subdivider is EFX-satisfied.
\end{itemize}
In both cases Properties 1--3 hold. Consequently, the only agent who may incur a cost exceeding her MMS is the subdivider, and only in Case II, in which she is EFX-satisfied; hence the allocation is IMMX. Finally, in Case II both the subdivider and the chooser are EFX-satisfied, as shown above. This completes the proof.
\end{proof}

\end{document}